\documentclass[superscriptaddress,twocolumn,amsmath,amssymb]{revtex4}
\usepackage{graphicx}
\usepackage{relsize}
\usepackage{color}
\usepackage{ulem}
\usepackage{braket}
\usepackage{mathrsfs}
\newcommand{\ketbra}[2]{\mathinner{|{#1}\rangle \! \langle{#2}|}}
\newcommand{\dket}[1]{\mathinner{|{#1}\rangle \! \rangle}}
\newcommand{\dbra}[1]{\mathinner{\langle \! \langle {#1} |}}
\newcommand{\dbraket}[2]{\mathinner{\langle \! \langle {#1} | {#2} \rangle \! \rangle}}
\newcommand{\dketbra}[2]{\mathinner{|{#1} \rangle \! \rangle \! \langle \! \langle{#2}|}}

\usepackage{amsmath,amsthm}
\newtheorem*{theorem}{Theorem (M. M. Wolf and D. Perez-Garcia)}
\newtheorem{lemma}{Lemma}
\newtheorem{corollary}{Corollary}
\usepackage{amsmath,bm}

\usepackage{hyperref}
\hypersetup{
colorlinks=true,
linkcolor=blue,
filecolor=magenta,
citecolor=magenta
}

\begin{document}

\title{Complete positivity violation of the reduced dynamics in higher-order quantum adiabatic elimination}

\author{Masaaki Tokieda}
\email{tokieda.masaaki.4e@kyoto-u.ac.jp}
\affiliation{Laboratoire de Physique de l'École Normale Supérieure, Inria, ENS, Mines ParisTech, Université PSL, Sorbonne Université, Paris, France}
\author{Cyril Elouard}
\affiliation{Inria, ENS Lyon, LIP, F-69342, Lyon Cedex 07, France}
\author{Alain Sarlette}
\affiliation{Laboratoire de Physique de l'École Normale Supérieure, Inria, ENS, Mines ParisTech, Université PSL, Sorbonne Université, Paris, France}
\affiliation{Department of Electronics and Information Systems, Ghent University, Belgium}
\author{Pierre Rouchon}
\affiliation{Laboratoire de Physique de l'École Normale Supérieure, Inria, ENS, Mines ParisTech, Université PSL, Sorbonne Université, Paris, France}

\begin{abstract}
  This paper discusses quantum adiabatic elimination, which is a model reduction technique for a composite Lindblad system consisting of a fast decaying sub-system coupled to another sub-system with a much slower timescale.
  Such a system features an invariant manifold that is close to the slow sub-system.
  This invariant manifold is reached subsequent to the decay of the fast degrees of freedom, after which the slow dynamics follow on it.
  By parametrizing the invariant manifold, the slow dynamics can be simulated via a reduced model.
  To find the evolution of the reduced state, we perform an asymptotic expansion with respect to the timescale separation.
  So far, the second-order expansion has mostly been considered.
  It has then been revealed that the second-order expansion of the reduced dynamics is generally given by a Lindblad equation, which ensures complete positivity of the time evolution.
  In this paper, we present two examples where complete positivity of the reduced dynamics is violated with higher-order contributions.
  In the first example, the violation is detected for the evolution of the partial trace without truncation of the asymptotic expansion.
  The partial trace is not the only way to parametrize the slow dynamics.
  Concerning this nonuniqueness, it was conjectured in \href{https://iopscience.iop.org/article/10.1088/2058-9565/aa7f3f}{[R. Azouit, F. Chittaro, A. Sarlette, and P. Rouchon, Quantum Sci. Technol. {\bf 2}, 044011 (2017)]} that there exists a parameter choice ensuring complete positivity.
  With the second example, however, we refute this conjecture by showing that complete positivity cannot be restored in any choice of parametrization.
  We discuss these results in terms of the invariant slow manifold consisting of quantum correlated states.
\end{abstract}

\maketitle

%%%%%%%%%%%%%%%%%%%%%%%%%%%%%%%%%%%%%%%%%%%%%%%%%%%%%%%%%%%%%%%%%%%%%%%%%%%%%%%%%%%%%%%%
%%%%%%%%%%%%%%%%%%%%%%%%%%%%%%%%%%%%%%%%%%%%%%%%%%%%%%%%%%%%%%%%%%%%%%%%%%%%%%%%%%%%%%%%

\section{Introduction}
\label{sec:Intro}

Any quantum system should be treated as an open system.
One reason is that perfect isolation of a quantum system is unrealistic experimentally and the influence of a surrounding environment needs to be taken into account.
Besides, perfectly isolated systems cannot be used for the purpose of quantum control.
In order to control or read out a quantum state, coupling to another system is unavoidable. A state of an open quantum system is represented by a density matrix.
To describe its evolution, various approximation methods have been developed so far.
One of the most widely used methods is based on the Markov assumption.
Starting from a system-environment Hamiltonian, the Born-Markov-Secular approximations lead to a Lindblad equation \cite{Breuer02}.

Lindblad equations can also be derived mathematically by imposing axiomatic conditions on the time evolution map.
It is reasonable to assume that the time evolution preserves the properties of density matrices, namely they are Hermitian, unit-trace, and positive semidefinite along the entire evolution.
The condition of positivity is usually replaced by complete positivity \cite{NielsenChuang}.
The complete positivity requirement in physics stems from the fact that a density matrix of an open quantum system is a reduced one, and the total density matrix including an environment should also remain positive semidefinite under the evolution.
One can show that the evolution of a density matrix is governed by a Lindblad equation if and only if
the time evolution map is a one-parameter semigroup ($\{ \Lambda_t \}_{t \geq 0}$ satisfying $\Lambda_t \circ \Lambda_s = \Lambda_{t+s}$ for all $t,s \geq 0$), the elements of which are trace preserving completely positive maps for all $t \geq 0$ \cite{GKS,Lindblad,Havel03}.
Note that the semigroup property is associated with the Markov assumption \cite{Breuer02}.

In this paper, we consider a composite open quantum system where the total evolution is governed by a Lindblad equation.
The composite system is assumed to consist of a fast decaying sub-system being weakly coupled to another system with a slower time scale.
In this setting, the time evolution typically starts with decay of fast degrees of freedom followed by a slower evolution of the remaining slow degrees of freedom.
In capturing the latter dynamics, thus, the fast degrees of freedom can be discarded.
This model reduction technique is known in quantum physics as adiabatic elimination and goes back to singular perturbation theory (see, e.g., \cite{Kotovic-review}).
Owing to the linearity of Lindblad equations, there exists in fact an invariant linear subspace, associated with the slow eigenvalues of the overall system, on which this dynamics rigorously takes place.

Adiabatic elimination offers two noteworthy aspects when applied to quantum physics.
First, it provides a model reduction technique for composite Lindblad systems.
By discarding the fast degrees of freedom, the slow dynamics can be described via a reduced model.
This enables simulations of large-dimensional systems that are otherwise infeasible.
Second, adiabatic elimination allows for reservoir engineering.
By crafting the coupling between the two sub-systems, we can design the dissipative dynamics of the slow evolution after the decay.
This aspect is important in recent developments of quantum technologies because dissipation is not necessarily the enemy of quantum technology, but can be leveraged to control a quantum state.
These two aspects of adiabatic elimination can be seen in previous studies, see below for references.

Various approaches have been developed to formulate adiabatic elimination for composite open quantum systems.
One of the earliest studies is \cite{Cirac92}, where the author applied the Born-Markov approximation to a bipartite Lindblad system, as commonly done for composite Hamiltonian systems \cite{Breuer02}.
A large body of studies has adopted a similar approach \cite{Wiseman93,Ripoll09,Reiter12,Lesanovsky13,Karabanov15,Tomita17,Damanet19,Viana22,Yang22}.
Other studies have explored alternative formulations, including the application of the Laplace transform to the projected master equation \cite{F-Shapiro20,Saideh20} and the use of the Schrieffer-Wolff transformation \cite{Kessler12,Burgarth21,Jager22}.
In contrast to Hamiltonian systems, relaxation behavior is built into the spectral properties of the generator in open quantum systems.
We can hence consider an invariant subspace, or in geometric language, a manifold to which trajectories are attracted in the long-time regime.
Formulations based on such a geometric picture were presented in \cite{Macieszczak16,Zanardi16}.
However, the applicability of these approaches is limited either to systems for which the eigenvectors of the generator can be determined or to the evaluation of contributions up to the second-order in the timescale separation.
In this paper, we focus on the formulation presented in \cite{Azouit17}, which can overcome these limitations.
It provides a geometric picture based on center manifold theory \cite{Fenichel79}.
The system according to this theory does exactly feature an invariant manifold corresponding to slow dynamics, and hence we view the model reduction to slow degrees of freedom as approximating both the manifold and the evolution once the system is initialized on it.
To formulate the model reduction based on this picture, we parametrize the degrees of freedom on the invariant manifold (see Fig.$\,$\ref{fig:Intro_formulation}).
We then seek to find two maps; one describing the time evolution of the parameters and the other assigning the parametrization to the solution of the Lindblad equation, that is, the density matrix of the total system.
To calculate these maps approximately for general problems, an asymptotic expansion with respect to the timescale separation is performed.
In this way, \cite{Azouit17} established a methodology to calculate higher-order contributions systematically.
Recently, this approach was extended to a periodically driven system where the driving frequency is comparable to the fast timescale, while the amplitude is in the order of the slow timescale \cite{Michiel23}.
In addition, numerical simulations were conducted to evaluate the reduced dynamics in a multisystem platform \cite{FM24}.

%%%%%%%%%%%%%%%%%%%%%%%%%%%%%%%%%%%%%%%%%%%%%%%%%%%%%%%%%
\begin{figure}[t]
  \includegraphics[keepaspectratio, scale=0.36]{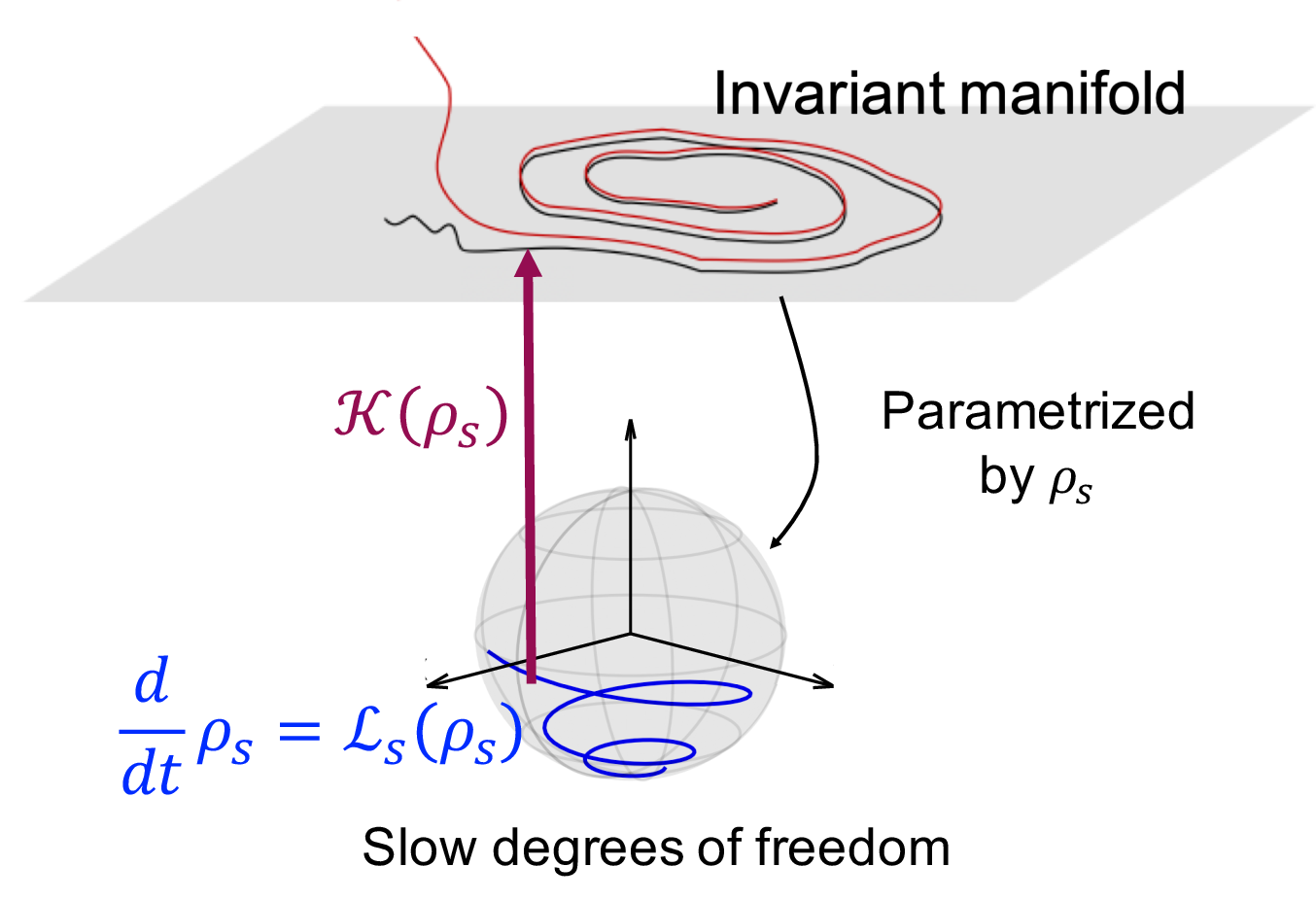}
  \caption{
  Schematic of the state evolution in a Lindblad system with timescale separation to illustrate the mechanism of adiabatic elimination. 
  From an arbitrary initial state, the system rapidly converges to the lower-dimensional invariant manifold (the gray plane), as illustrated by the red trajectory.
  In adiabatic elimination, we parametrize the degrees of freedom on this manifold (called slow degrees of freedom) by $\rho_s$ and describe the dynamics within the manifold, illustrated by the black trajectory, via the reduced dynamics, illustrated by the blue trajectory.
  Here the reduced system is assumed to be a qubit.
  The primary objectives of adiabatic elimination formulation are to calculate the map ($\mathcal{L}_s$) describing the evolution of $\rho_s$ and the map ($\mathcal{K}$) assigning the parametrization to the total state, which follows the Lindblad equation, on the invariant manifold.
  }
  \label{fig:Intro_formulation}
\end{figure}
%%%%%%%%%%%%%%%%%%%%%%%%%%%%%%%%%%%%%%%%%%%%%%%%%%%%%%%%%

In the geometric approach, adiabatic elimination includes a gauge degree of freedom associated with the nonuniqueness of the parametrization.
If the slow dynamics is parametrized via a density matrix, then one expects as a physical requirement that the two maps introduced above should preserve the quantum structure.
This expectation is behind the conjecture made in \cite{Azouit17}; the authors conjectured the existence of a gauge choice such that the reduced dynamics is governed by a Lindblad equation and the assignment is a trace preserving completely positive map (also called a Kraus map \cite{Choi75}) up to any order of the asymptotic expansion.
So far, this has been proved to be true for a general class of settings up to the second-order expansion;
it was shown in \cite{Azouit17} that the evolution equation admits a Lindblad equation and in \cite{AzouitThesis} that there always exists a gauge choice ensuring the Kraus map assignment.
Studies of the higher-order contributions have been limited so far.
For a two qubit system, \cite{Alain20} reported an example supporting the conjecture at any order.
Similar but different issues were discussed in \cite{Burgarth21}.
The authors extended the Schrieffer-Wolff transformation to open quantum systems.
They found that the effective adiabatic generator, which has the same block matrix structure as the unperturbed part and provides a first-order approximation of the total dynamics for all the time, cannot in general be put in the Lindblad form, even without truncation in the perturbation series, due to an unavoidable negative coefficient in front of one of the dissipators.
We note that they investigated the effective generator for the total system, while the conjecture is about the generator for the reduced system.
In our picture, it is clear that the total dynamics follows a Lindblad equation.
The question is whether it can be split up into a Lindblad equation on a Hilbert space equivalent to the slow sub-system, and a Kraus mapping of this parametrization to the total system.

In this paper, we challenge this conjecture by considering higher-order contributions beyond second-order.
We emphasize that recent advances in quantum technologies warrant accurate simulation of the slow dynamics.
Thus, understanding of higher-order contributions is increasingly in need.
We specifically consider two examples to investigate the conjecture.
In the first example, which is a three-level system dispersively coupled to a strongly dissipative qubit, adiabatic elimination can be performed without any truncation of the expansion series.
In such all-order analysis, we show that the parametrization via the partial trace of the total density matrix with respect to the fast sub-system yields a non-Lindblad equation violating even positivity, let alone complete positivity, of the time evolution.
This is a result for the partial trace parametrization and, according to the conjecture, there might exist a different gauge choice in which the Lindblad form is restored.
To explore the capacity of the gauge degree of freedom , we consider a qubit resonantly coupled to a strongly dissipative oscillator as the second example.
In this example, we can rigorously prove that, with fourth-order contributions, complete positivity of the reduced dynamics cannot be attained, whatever the gauge choice is.
Thus, this system serves as a counterexample to the conjecture.
The proof utilizes a result in \cite{WPG10} that reveals a constraint on the spectrum of completely positive qubit maps.
We discuss this complete positivity violation in terms of quantum correlations in states on the invariant manifold, which imposes a restriction on the initial state of the reduced system.
In contrast to our previous report \cite{Tokieda22}, which was limited to investigating an oscillator-qubit system,
this paper presents the example of a three-level system (the first example), where the all-order analysis is possible.
Furthermore, we offer more detailed discussions to elucidate the interpretation of the complete positivity violation, aiming to convince readers that such violation is not an unphysical anomaly but an anticipated consequence of quantum correlations, and we suggest the possibility of its experimental verification.

The paper is organized as follows.
In Sec.$\,$\ref{sec:AE}, the machinery of adiabatic elimination developed in \cite{Azouit17} is reviewed.
To investigate the role of the gauge degree of freedom, we derive how the time evolution and assignment maps are modified in different choices of gauge.
In Sec.$\,$\ref{sec:Diag}, we consider a dispersively coupled system.
For the partial trace, we show that complete positivity of the evolution is violated in the all-order analysis.
Next in Sec.$\,$\ref{sec:JC}, we consider an oscillator-qubit system in which the dissipative oscillator system is eliminated.
With this example, we prove the impossibility of restoring complete positivity by any gauge transformation.
Interpretations of these findings are presented in Sec.$\,$\ref{sec:diss}.
Lastly, concluding remarks are made in Sec.$\,$\ref{sec:conc}.
More details about all the claims can be found in the appendix.

%%%%%%%%%%%%%%%%%%%%%%%%%%%%%%%%%%%%%%%%%%%%%%%%%%%%%%%%%%%%%%%%%%%%%%%%%%%%%%%%%%%%%%%%

\section{Adiabatic elimination}
\label{sec:AE}

In this section, we review the machinery of adiabatic elimination developed in \cite{Azouit17}.
We consider a system consisting of a fast decaying sub-system coupled to another sub-system with a slower timescale.
Let $\mathscr{H}_A$ ($\mathscr{H}_B$) be the Hilbert space of the fast (slow) sub-system.
The density matrix of the composite system, $\mathscr{H}_A \otimes \mathscr{H}_B$, denoted by $\rho$ follows a Lindblad equation,
\begin{equation}
  \frac{d}{dt} \rho = \mathcal{L}_A \otimes \mathcal{I}_B (\rho) + \epsilon \mathcal{I}_A \otimes \mathcal{L}_B (\rho) + \epsilon \mathcal{L}_{\rm int} (\rho) \equiv \mathcal{L}_{\rm tot} (\rho).
  \label{eq:AE_master}
\end{equation}
For $\xi = A$ and $B$, $\mathcal{I}_\xi$ are the identity superoperators acting only on operators on $\mathscr{H}_\xi$.
$\mathcal{L}_A$ is a Lindbladian acting only on $\mathscr{H}_A$ and generally reads
\begin{equation}
  \mathcal{L}_A \bullet = - i [H_A, \bullet] + \sum_{k} \mathcal{D}[L_{A,k}] \bullet,
  \label{eq:AE_LA}
\end{equation}
with a Hamiltonian $H_A$ and jump operators $\{ L_{A,k} \}$, all of which are operators on $\mathscr{H}_A$. We have also introduced the commutator superoperator $[H, \bullet] = H \bullet - \bullet H$ and the dissipator superoperator $\mathcal{D}[L] \bullet = L \bullet L^\dagger - (L^\dagger L \bullet + \bullet L^\dagger L)/2$ for any operator $H$ and $L$.
We assume that the evolution only with $\mathcal{L}_A$ exponentially converges to a unique steady state $\bar{\rho}_A$.
In other words, among the spectrum of $\mathcal{L}_A$, the eigenvalue zero is simple and the other eigenvalues have strictly negative real part.
$\mathcal{L}_B$ and $\mathcal{L}_{\rm int}$ are superoperators acting on $\mathscr{H}_B$ and $\mathscr{H}_A \otimes \mathscr{H}_B$, respectively, and are assumed to contain only Hamiltonian terms.
Lastly, $\epsilon$ is a non-negative parameter representing the timescale separation.
Physically, $\mathcal{L}_A$ and $\mathcal{L}_B$ describe the internal dynamics of $\mathscr{H}_A$ and $\mathscr{H}_B$, respectively, and $\mathcal{L}_{\rm int}$ determines how the two sub-systems interact.
Note that the internal dynamics $\mathcal{L}_B$ of $\mathscr{H}_B$ is assumed slow, i.e., the model Eq.$\,$(\ref{eq:AE_master}) must hold in a frame that follows its potential fast motion.

As described in the introduction section, the goal of adiabatic elimination is to find the slow dynamics on the invariant manifold.
For linear equations such as a Lindblad equation, an invariant manifold is characterized by the (right) eigenoperators of the generator $\mathcal{L}_{\rm tot}$ in Eq.$\,$(\ref{eq:AE_master}) whose eigenvalues have real part close to zero.
Such a subspace is preserved by the operation of $\mathcal{L}_{\rm tot}$ and thus is invariant with respect to the time evolution map.
We note that, when $\mathcal{L}_{\rm int}$ is nonzero, an invariant manifold does not exactly coincide with the slow sub-system $\mathscr{H}_B$ because of correlations building up on the invariant manifold (see Sec.$\,$\ref{sec:diss} for detailed discussions).

To describe the slow dynamics, we should parametrize the degrees of freedom on the invariant manifold.
Following \cite{Azouit17}, we use a density matrix for the parametrization.
As a mathematical model reduction technique, there is no preference in that choice.
In applications to physics, on the other hand, it is convenient to employ a parametrization that facilitates interpretation of the slow dynamics.
A suitable representation in this regard is a density matrix, since most studies of open quantum systems have been based on it.
We note that the partial trace ${\rm tr}_A (\rho)$, with ${\rm tr}_A$ the trace over $\mathscr{H}_A$, has commonly been used  to represent the reduced state \cite{Breuer02,NielsenChuang}.
This is a valid gauge choice, as far as the timescales are well separated (equivalently, $\epsilon \ll 1$).
This choice plays a central role in the following discussions.
For clear distinction, we denote the partial trace by $\rho_B = {\rm tr}_A (\rho)$ and general density matrix parametrization by $\rho_s$.

Once the parametrization is fixed to $\rho_s$, we seek to find the following two maps (see Fig.$\,$\ref{fig:Intro_formulation}).
One, denoted by $\mathcal{L}_s$, describes the time evolution of $\rho_s$, namely $(d/dt) \rho_s = \mathcal{L}_s (\rho_s)$.
The other, denoted by $\mathcal{K}$, maps $\rho_s$ to the solution $\rho$ of the total Lindblad equation Eq.$\,$(\ref{eq:AE_master}), $\rho = \mathcal{K} (\rho_s)$.
Throughout this paper, we assume that $\mathcal{K}$ and $\mathcal{L}_s$ are linear and time-independent.
Since $\rho$ satisfies Eq.$\,$(\ref{eq:AE_master}), we obtain
\begin{equation}
  \mathcal{K} (\mathcal{L}_s (\rho_s)) = \mathcal{L}_{\rm tot} (\mathcal{K}(\rho_s)),
  \label{eq:AE_inv1}
\end{equation}
from which we can determine $\mathcal{K}$ and $\mathcal{L}_s$ in principle. We call this relation the invariance condition in this paper.

Except special cases (see Sec.$\,$\ref{sec:Diag}), it is difficult to find $\mathcal{K}$ and $\mathcal{L}_s$ satisfying the invariance condition Eq.$\,$(\ref{eq:AE_inv1}) exactly.
To proceed, we assume $\epsilon \ll 1$ and perform the asymptotic expansions as
\begin{equation}
  \mathcal{K} = \sum_{n = 0}^\infty \epsilon^n \mathcal{K}_n, \ \ \ \ \ \ \ \ \mathcal{L}_s = \sum_{n = 0}^\infty \epsilon^n \mathcal{L}_{s,n}.
  \label{eq:AE_asymexp}
\end{equation}
When $\epsilon = 0$, the solution of Eq.$\,$(\ref{eq:AE_master}) after the decay of the fast sub-system reads $\rho(t) = \bar{\rho}_A \otimes {\rm tr}_A (\rho (t=0))$ with the initial density matrix $\rho(t=0)$.
Therefore, the $\epsilon^0$ order elements are given by
\begin{equation}
  \mathcal{K}_{0} (\rho_s) = \bar{\rho}_A \otimes \rho_s, \ \ \ \ \ \ \ \ \mathcal{L}_{s,0} (\rho_s) = 0.
  \label{eq:AE_asymexp0}
\end{equation}

The higher-order contributions can be evaluated by inserting the expansions Eq.$\,$(\ref{eq:AE_asymexp}) into the invariance condition Eq.$\,$(\ref{eq:AE_inv1}).
Detailed calculations are presented in Appendix \ref{app:inv.cond.}.
We recall that, as the parameter choice can involve different options at any orders of $\epsilon$, the solution is not unique.
This reflects into the fact that, due to the singularity of $\mathcal{L}_A$, ${\rm tr}_A \circ \mathcal{K}_{n \geq 1}$ cannot be fully determined from the invariance condition.
In what follows, we denote $G \equiv {\rm tr}_A \circ \sum_{n = 1}^\infty \epsilon^n \mathcal{K}_n $, which can be any linear and time-independent superoperator on $\mathscr{H}_B$.
Together with Eq.$\,$(\ref{eq:AE_asymexp0}), we find 
\begin{equation}
  \rho_B = {\rm tr}_A (\rho) = {\rm tr}_A (\mathcal{K} (\rho_s)) = \rho_s + G(\rho_s).
  \label{eq:AE_rhobrhos}
\end{equation}
Throughout this paper, we assume that $(\mathcal{I}_B + G)$ is invertible, which is valid for $\epsilon \ll 1$.

Let us see how $\mathcal{K}$ and $\mathcal{L}_s$ for an arbitrary gauge choice are related to those for the partial trace.
For the sake of clarity, we write the gauge dependence explicitly as $\mathcal{K}^G$ and $\mathcal{L}_{s}^{G}$.
We note in advance that the following relations are results of general basis change and are not associated with the quantum structure.
For $\mathcal{K}^G$, note $\rho = \mathcal{K}^{G} (\rho_s) = \mathcal{K}^{G=0} (\rho_B)$.
Substituting Eq.$\,$(\ref{eq:AE_rhobrhos}) into the rightmost side gives
\begin{equation}
  \mathcal{K}^G = \mathcal{K}^{G=0} \circ (\mathcal{I}_B + G).
  \label{eq:AE_KG}
\end{equation}
For $\mathcal{L}_{s}^{G}$, the time evolution of the partial trace reads $\mathcal{L}_{s}^{G=0} (\rho_B) = (d/dt) \rho_B = (d/dt) (\rho_s + G(\rho_s)) = \mathcal{L}_{s}^{G} (\rho_s) + G (\mathcal{L}_{s}^{G} (\rho_s))$.
Comparing the leftmost and rightmost sides, we find $\mathcal{L}_{s}^{G=0} \circ (\mathcal{I}_B + G) = (\mathcal{I}_B + G) \circ \mathcal{L}_{s}^{G}$.
From the existence of $(\mathcal{I}_B + G)^{-1}$, we obtain
\begin{equation}
  \mathcal{L}_{s}^G = (\mathcal{I}_B + G)^{-1} \circ \mathcal{L}_{s}^{G=0} \circ (\mathcal{I}_B + G).
  \label{eq:AE_LG}
\end{equation}
This indicates that the spectrum of $\mathcal{L}_s$ or the decay rate inside an invariant manifold is independent of gauge choice.
This is expected since the decay rate must not change depending on the way the slow dynamics is parametrized.
Equations (\ref{eq:AE_KG}) and (\ref{eq:AE_LG}) are useful in analyzing possible transformations that the gauge degree of freedom can make (see Sec.$\,$\ref{sec:JC}).

As summarized in the introduction section, the authors of \cite{Azouit17} conjectured the existence of a gauge choice leading to reduced dynamics described by a Lindbladian,
$\sum_{j=0}^n \epsilon^j \mathcal{L}_{s,j} (\rho_s) = - i [H_s, \rho_s] + \sum_{k} \mathcal{D}[L_{s,k}] (\rho_s)$ with a Hamiltonian $H_s$ and jump operators $\{ L_{s,k} \}$,
and assignment described by a Kraus map,
$\sum_{j=0}^n \epsilon^j \mathcal{K}_{j} (\rho_s) = \sum_{k} M_k \rho_s M_k^\dagger$ with operators $M_k: \mathscr{H}_B \to \mathscr{H}_A \otimes \mathscr{H}_B$, up to $\epsilon^n$ for any positive integer $n$.
For a general class of settings, this conjecture has been proved up to $n = 2$ so far.
In the following sections, we present examples where complete positivity of the reduced dynamics is violated with higher-order ($n>2$) terms.

%%%%%%%%%%%%%%%%%%%%%%%%%%%%%%%%%%%%%%%%%%%%%%%%%%%%%%%%%%%%%%%%%%%%%%%%%%%%%%%%%%%%%%%%

\section{Complete positivity violation in all-order adiabatic elimination}
\label{sec:Diag}

\subsection{Problem setting}

In this section, we demonstrate complete positivity violation of the reduced dynamics.
In order to stress that the violation is not due to the truncation of the perturbation series, we consider an exactly solvable system where $\mathcal{K}$ and $\mathcal{L}_s$ satisfying the invariance condition Eq.$\,$(\ref{eq:AE_inv1}) can be obtained without the asymptotic expansion.
The total system consists of a target qudit ($d$-dimensional) system being coupled to another dissipative system through a single-term Hamiltonian.
To represent qudit operators, we introduce  $E_{m,n} \in \mathbb{R}^{d \times d} \ (m,n = 1,\dots,d)$  as $[E_{m,n}]_{i,j} = 1 \ (i = m \ {\rm and} \ j = n)$ and $0 \ ({\rm else})$, and $E_m = E_{m,m}$ in some canonical basis.
With these, we assume the following form of $\mathcal{L}_B$ and $\mathcal{L}_{\rm int}$ as in \cite{Alain20};
\begin{equation*}
  \epsilon \mathcal{L}_B \bullet = - i \ [ \sum_{m=1}^d \omega_m E_m, \bullet],
\end{equation*}
and
\begin{equation}
  \epsilon \mathcal{L}_{\rm int} \bullet = i  [ \sum_{m=1}^d \chi_m ( V_A \otimes E_m), \bullet],
  \label{eq:Diag_int}
\end{equation}
where $\{ \omega_m \}$ are the transition frequencies of the qudit, $\{ \chi_m \}$ are the coupling constants, and $V_A$ is an operator on $\mathscr{H}_A$.
Regarding $\mathcal{L}_A$, we only assume the existence of a unique steady state and do not specify its form in computing analytic expressions of $\mathcal{K}$ and $\mathcal{L}_s$.
When we discuss whether $\mathcal{L}_s$ is a Lindbladian later, we consider a driven-dissipative qubit system represented by
\begin{equation}
  \mathcal{L}_A \bullet = -i [ \frac{\Omega}{2} \sigma_x + \frac{\Delta}{2} \sigma_z, \bullet] + \kappa \mathcal{D} [\sigma_-] \bullet, \ \ \ V_A = \sigma_z,
  \label{eq:Diag_qubit}
\end{equation}
with the drive amplitude $\Omega$, the drive detuning from the qubit frequency $\Delta$, and the Pauli matrices $\{ \sigma_i \}_{i=x,y,z}$, and $\sigma_{\pm} = (\sigma_x \pm i \sigma_y)/2$.
Assuming the qudit to be a $d$-level approximation of an optical cavity, this Lindbladian describes a quantum nondemolition measurement of the photon number in the absence of dissipation \cite{Antoine21}.

Note that $\mathcal{L}_B$ and $\mathcal{L}_{\rm int}$ commute.
In the rotating frame with respect to the Hamiltonian $\sum_{m=1}^d \omega_m E_m$, thus, the interaction Hamiltonian does not change, while the qudit internal dynamics becomes trivial as $\mathcal{L}_B = 0$.
In the following, we consider adiabatic elimination in this frame.

\subsection{Adiabatic elimination at any order}
\label{sec:Diag_any.order}

We recall that our goal is to find maps $\mathcal{K}$ and $\mathcal{L}_s$ satisfying the invariance condition Eq.$\,$(\ref{eq:AE_inv1}).
To this end, we note that
\begin{equation}
  \mathcal{L}_{\rm tot} ( A \otimes E_{m,n} ) = \mathcal{L}_A^{(m,n)}(A) \otimes E_{m,n},
  \label{eq:Diag_LAmn}
\end{equation}
with $A$ any operator on $\mathscr{H}_A$ and
\begin{equation*}
  \mathcal{L}_A^{(m,n)} (A) = \mathcal{L}_A (A) + i (\chi_m V_A A - \chi_n A V_A).
\end{equation*}
As detailed in Appendix \ref{app:ae.qudit_all}, we then find
\begin{equation}
  \mathcal{K}^{G=0} (\rho_B) = \sum_{m,n = 1}^d Q_{m,n} \otimes E_m \rho_B E_n,
  \label{eq:Diag_KPT}
\end{equation}
and
\begin{equation}
  \mathcal{L}_s^{G=0} (\rho_B) = \sum_{m,n = 1}^d \lambda_{m,n} \ E_m \rho_B E_n,
  \label{eq:Diag_LsPT}
\end{equation}
where, for $m, n = 1 \dots d$, $\lambda_{m,n}$ is the eigenvalue of $\mathcal{L}_A^{(m,n)}$, which is in fact an eigenvalue of $\mathcal{L}_{\rm tot}$ as seen from Eq.$\,$(\ref{eq:Diag_LAmn}), with the smallest absolute real part and $Q_{m,n}$ is the corresponding (right) eigenoperator that is normalized as ${\rm tr}_A (Q_{m,n}) = 1$.
These are for the partial trace parametrization $(G = 0)$ as confirmed by the relation ${\rm tr}_A  (\mathcal{K}^{G=0} (\rho_B)) = \rho_B$ (see Eq.$\,$(\ref{eq:AE_rhobrhos})).
One can check the invariance condition Eq.$\,$(\ref{eq:AE_inv1}) using Eq.$\,$(\ref{eq:Diag_LAmn}).

In what follows, we investigate whether $\mathcal{L}_s^{G=0}$ is a Lindbladian.
According to the technical results in Appendix \ref{app:vect_Lindbladian}, this holds if and only if $S^\top \lambda S$ is positive semidefinite, which we denote as $S^\top \lambda S \geq 0$ in the following, 
with the matrix transpose $\top$ and a matrix $S$ defined in Eq.$\,$(\ref{eq:vect_MatrixS}).

As an example, let us see the case $d = 2$.
We find $S^\top \lambda S = {\rm Re}(-\lambda_{1,2})/2$.
As mentioned underneath Eq.$\,$(\ref{eq:Diag_LsPT}), $\lambda_{1,2}$ is an eigenvalue of $\mathcal{L}_{\rm tot}$.
From the stability condition, $S^\top \lambda S \geq 0$, and thus $\mathcal{L}_s^{G=0}$ always admits the Lindblad form.
This was shown in \cite{Alain20}, where the authors further proved the existence of gauge choices such that $\mathcal{K}^G$ is completely positive and surjective.

\subsection{Qutrit ($d = 3$) case}

%%%%%%%%%%%%%%%%%%%%%%%%%%%%%%%%%%%%%%%%%%%%%%%%%%%%%%%%%
\begin{figure}[t]
  \includegraphics[keepaspectratio, scale=0.42]{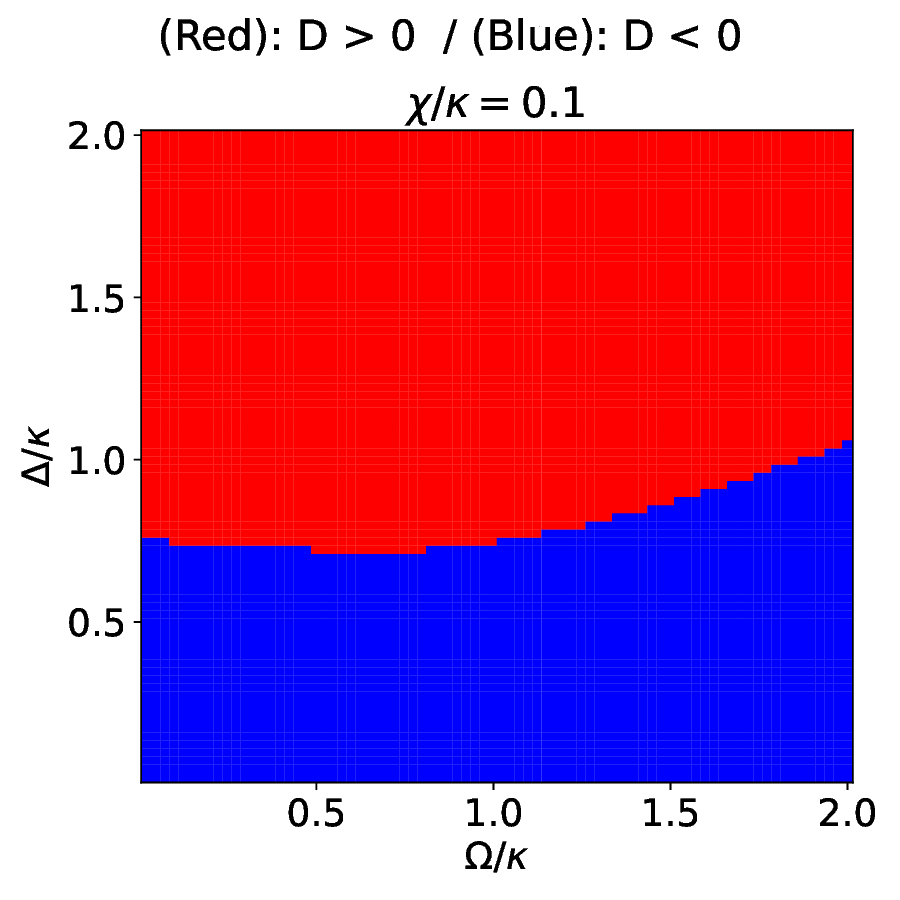}
  \caption{Parameter regions ensuring the Lindblad form of the time evolution generator $\mathcal{L}_s^{G=0}$ determined by the sign of $D$ (see Eq.$\,$(\ref{eq:Diag_lmdLindblad})).
  The Lindblad form is attained in the red region $(D > 0)$, while a dissipator comes with a negative coefficient in the blue region $(D < 0)$.
  The coupling strengths are set $(\chi_1,\chi_2,\chi_3) = (0,\chi,2\chi)$ with $\chi/\kappa = 0.1$.
  }
  \label{fig:Diag_sign}
\end{figure}
%%%%%%%%%%%%%%%%%%%%%%%%%%%%%%%%%%%%%%%%%%%%%%%%%%%%%%%%%

When $d = 3$, we find
\begin{equation}
  \begin{gathered}
    S^\top \lambda S = \\
    \begin{pmatrix}
      \frac{{\rm Re}(-\lambda_{1,2})}{2} & \frac{i {\rm Im}(\lambda_{1,2}) - \lambda_{1,3} + \lambda_{2,3}}{6} \\
      \frac{ - i {\rm Im}(\lambda_{1,2}) - \lambda_{1,3}^* + \lambda_{2,3}^*}{6} & \frac{{\rm Re}( \lambda_{1,2} - 2 (\lambda_{1,3} + \lambda_{2,3}) )}{18}
    \end{pmatrix},
  \end{gathered}
  \label{eq:Diag_CoefficientMatrix}
\end{equation}
with ${\rm Im}$ denoting the imaginary part.
$\mathcal{L}_s^{G=0}$ is a Lindbladian if and only if this matrix is positive semidefinite.
After sorting out, the condition reads $D \geq 0$ with $D$ defined by
\begin{equation}
  \begin{gathered}
    D = |\lambda_{1,2} + \lambda_{1,3} + \lambda_{2,3}|^2 \\
    - 2 (|\lambda_{1,2}|^2 + |\lambda_{1,3}|^2 + |\lambda_{2,3}|^2)
    - 4 {\rm Im} (\lambda_{1,2}) {\rm Im} (\lambda_{2,3}).
  \end{gathered}
  \label{eq:Diag_lmdLindblad}
\end{equation}
Unlike the case $d = 2$, it is not clear if this is satisfied generally.
To proceed, we numerically evaluate the sign of $D$ when the qutrit is coupled to a fast decaying qubit described by Eq.$\,$(\ref{eq:Diag_qubit}).
In this evaluation, we assume $(\chi_1,\chi_2,\chi_3) = (0,\chi,2\chi)$ with $\chi / \kappa = 0.1$.
The results are shown in Fig.$\,$\ref{fig:Diag_sign}.
We can see the blue regions where $D$ is negative.
In these regions, one of the dissipators constituting $\mathcal{L}_s^{G=0}$ has a negative coefficient.
Therefore, unlike the case $d = 2$, $\mathcal{L}_s^{G=0}$ is not a Lindbladian in general.
We note that the non-Lindblad form of $\mathcal{L}_s^{G=0}$ is obtained even with infinitesimal coupling constants as discussed in Appendix \ref{app:ae.qudit_4th}.

To be more precise, it is complete positivity of the time evolution that is violated.
As discussed in the introduction section, the generator is a Lindbladian if and only if the time evolution map satisfies the semigroup relation and is a Kraus map at any time.
In the current problem, the time evolution map, $\exp(\mathcal{L}_s^{G=0} t)$ with $t \in \mathbb{R}_{\geq 0}$, satisfies the semigroup relation and preserves the Hermitian property and trace.
Thus, the non-Lindblad form of $\mathcal{L}_s^{G=0}$ detected by the negative sign of $D$ signifies the complete positivity violation of the time evolution map $\exp(\mathcal{L}_s^{G=0} t)$.
In this qudit example, not only complete positivity, even positivity is violated.
To see this, we note that the operation of the time evolution map $\exp(\mathcal{L}_s^{G=0} t)$ reads as
\begin{equation*}
  e^{\mathcal{L}_s^{G = 0} t} (\rho_B) = \sum_{m,n = 1}^d e^{\lambda_{m,n} t} E_m \rho_B E_n.
\end{equation*}
From Lemma \ref{theorem:diagonal} in Appendix \ref{app:vect_choi}, a superoperator of this form is positive if and only if it is completely positive.
Thus, the violation of complete positivity is accompanied by that of positivity.
We present an interpretation of such a nonpositive evolution in Sec.$\,$\ref{sec:diss}.

This example demonstrates that complete positivity (and positivity) of the partial trace evolution can be violated in adiabatic elimination, even without truncation in the series expansion.
In this situation, the conjecture in \cite{Azouit17} states that the negative coefficient in front of a dissipator can be eliminated by a gauge transformation Eq.$\,$(\ref{eq:AE_LG}), and complete positivity of the reduced dynamics is restored in a different parametrization.
In the current example, it is difficult to examine this conjecture because a simple single criterion like Eq.$\,$(\ref{eq:Diag_lmdLindblad}) holds for $d=3$ only under the special conditions where $G=0$ yields a diagonal superoperator (Eq.$\,$(\ref{eq:Diag_LsPT})).

%%%%%%%%%%%%%%%%%%%%%%%%%%%%%%%%%%%%%%%%%%%%%%%%%%%%%%%%%%%%%%%%%%%%%%%%%%%%%%%%%%%%%%%%

\section{Jaynes-Cummings model with damped oscillator}
\label{sec:JC}

\subsection{Problem setting}
\label{sec:JC_setting}

To investigate roles of the gauge degree of freedom more closely, this section is dedicated to a slow qubit system being coupled to a strongly dissipative oscillator system \cite{Tokieda22}.
We assume that the Hamiltonian is given by the Jaynes-Cummings Hamiltonian and that the oscillator is coupled to a Markovian environment at finite temperature.
The qubit is assumed to be nondissipative for simplicity. In the frame rotating with the qubit frequency, we have
\begin{equation}
  \mathcal{L}_{A} \bullet = - i [\Delta_A  a^\dagger \! a, \bullet] + \gamma (1 + n_{\rm th}) \mathcal{D}[a] \bullet + \gamma n_{\rm th} \mathcal{D}[a^\dagger] \bullet,
  \label{eq:JC_LA}
\end{equation}
\begin{equation*}
  \epsilon \mathcal{L}_{\rm int} \bullet = - i [ g ( a^\dagger \otimes \sigma_{-} + a \otimes \sigma_{+} ), \bullet],
\end{equation*}
and $\mathcal{L}_B = 0$, with the oscillator detuning from the qubit frequency $\Delta_A$, the decay rate $\gamma$, the asymptotic oscillator quantum number in the absence of coupling $n_{\rm th}$ (see Eq.$\,$(\ref{eq:damposc_nth})), and the coupling constant $g$.
Operators $a$ and $a^\dagger$ are the annihilation and creation operators of the oscillator, respectively.
This form of Lindbladian is used as a benchmark when analyzing oscillator-qubit interacting systems in cavity or superconducting circuit architectures.

The full spectrum and eigenoperators of $\mathcal{L}_A$ are provided in Appendix \ref{app:damposc}.
The result confirms, as long as $\gamma > 0$, the existence of a unique steady state $\bar{\rho}_A$ given by Eq.$\,$(\ref{eq:damposc_steadystate}).

\subsection{Fourth-order adiabatic elimination}
\label{sec:JC_AE}

To our knowledge, the invariance condition Eq.$\,$(\ref{eq:AE_inv1}) for this system cannot be solved exactly.
Thus, we perform the asymptotic expansion as discussed underneath Eq.$\,$(\ref{eq:AE_asymexp}).
In this example, the timescale of the oscillator system is characterized by $\gamma^{-1}$, while that of the interaction is $|g|^{-1}$.
Thus, the timescale separation parameter $\epsilon$ reads $\epsilon = |g|/\gamma$.
Assuming $\epsilon \ll 1$, we calculate contributions up to the fourth-order.

As shown in Appendix \ref{app:damposc}, $\mathcal{L}_{s}$ for the partial trace, $\mathcal{L}_{s}^{G=0}$, reads up to the fourth-order expansion
\begin{equation}
  \begin{gathered}
    \mathcal{L}_{s}^{G=0} \bullet = - i [\frac{\omega_B^{(4)}}{2} \sigma_z, \bullet] \\
    + \gamma_{-}^{(4)} \mathcal{D}[\sigma_-] \bullet + \gamma_{+}^{(4)} \mathcal{D}[\sigma_+] \bullet + \gamma_\phi^{(4)} \mathcal{D}[\sigma_z] \bullet.
  \end{gathered}
  \label{eq:JC_LsPT}
\end{equation}
The coefficients $\omega_B^{(4)}$, $\gamma_{\pm}^{(4)}$, and $\gamma_{\phi}^{(4)}$ are real numbers defined by
\begin{equation*}
  \omega_B^{(4)} = {\rm Im} (b_- + b_+), \ \ \ \ \ \gamma_{\pm}^{(4)} = 2 {\rm Re} (b_\pm),
\end{equation*}
and
\begin{equation*}
  \gamma_\phi^{(4)} = - \frac{8g^4 n_+ n_- ( 3 - 6 (2\Delta_A/\gamma)^2 - (2\Delta_A/\gamma)^4 )}{\gamma^3 (1+(2\Delta_A/\gamma)^2)^3},
\end{equation*}
where $b_\pm$ are
\begin{equation*}
  b_\pm =  \frac{2g^2 n_{\pm}}{\bar{\gamma}} + \frac{8g^4 n_{\pm}^2}{\bar{\gamma}^3} + \frac{8g^4 n_+ n_-(1+8 i \gamma \Delta_A / |\bar{\gamma}|^2)}{\bar{\gamma}^* |\bar{\gamma}|^2},
\end{equation*}
with $n_+ = n_{\rm th}$, $n_- = 1 + n_{\rm th}$, and $\bar{\gamma} = \gamma + 2 i \Delta_A$.

The coefficient $\omega_B^{(4)}$ represents the qubit frequency shift due to the coupling with the oscillator.
Up to the second-order, $\omega_B^{(2)}$, it reads $\omega_B^{(2)} = - 4 \Delta_A g^2 (n_- + n_+)/|\bar{\gamma}|^2$.
The $\gamma_{\pm}^{(4)}$ and $\gamma_\phi^{(4)}$ terms describe the effective qubit decay induced by the coupling to the dissipative oscillator.
On the one hand, when $|g|/\gamma \ll 1$, $\gamma_{\pm}^{(4)}$ are dominated by the second-order contributions given by $\gamma_{\pm}^{(2)} = 4 g^2 \gamma n_{\pm}/|\bar{\gamma}|^2 > 0$, and thus $\gamma_{\pm}^{(4)} > 0$.
On the other hand, $\gamma_\phi^{(4)}$ involves only the fourth-order contribution and
\begin{equation*}
  \begin{gathered}
    \gamma_\phi^{(4)} < 0 \ \ \ {\rm when} \\
    n_{\rm th} > 0 \ \ {\rm and} \ \ |\Delta_A|/\gamma < \sqrt{2\sqrt{3}-3}/2 \simeq 0.34,
  \end{gathered}
\end{equation*}
even if the condition for the asymptotic expansion, $|g|/\gamma \ll 1$, holds.

Even when $\gamma_\phi^{(4)} < 0$, the stability of the time evolution can be confirmed as follows.
As calculated in Appendix \ref{app:damposc.4thL_spec}, the spectrum of $\mathcal{L}_s^{G=0}$ reads 
\begin{equation*}
    \{ 0, -1/T_2 + i \omega_B^{(4)}, -1/T_2 - i \omega_B^{(4)}, -1/T_1 \}
\end{equation*}
with $1/T_1 = \gamma_-^{(4)} + \gamma_+^{(4)}$ and $1/T_2 = 1/(2 T_1) + 2 \gamma_\phi^{(4)}$.
Since $\gamma_{\pm}^{(4)} > 0$ and $\gamma_{\pm}^{(4)} \gg |\gamma_\phi^{(4)}|$ when $|g|/\gamma$ is small, we have $T_1 > 0$ and $T_2 > 0$.
Therefore, the time evolution is stable even when $\gamma_\phi^{(4)}$ is negative.

For a similar reason, the time evolution map is positive even when $\gamma_\phi^{(4)} < 0$.
We provide a proof in Appendix \ref{app:damposc.4thL_positive}.
However, as detailed in Sec.$\,$\ref{sec:JC_gauge}, the evolution is not completely positive when $\gamma_\phi^{(4)} < 0$.
This mechanism appears to be related to the well-known example of the transpose of a matrix, which is positive but not complete positive.
Indeed, that standard example essentially says that an evolution which contracts a single Bloch vector direction fast, but the two other directions slowly, is not completely positive.
We here have this situation with $\gamma_\phi^{(4)} < 0$, slowing down the contraction of the $x$ and $y$ components of the Bloch vector.

\subsection{Gauge transformation}
\label{sec:JC_gauge}

When $\gamma_\phi^{(4)} < 0$, $\mathcal{L}_s^{G=0}$ is thus not a Lindbladian.
To be more precise, the time evolution map with $\mathcal{L}_s^{G=0}$ is not completely positive.
On the other hand, the Lindblad form might be recovered in another gauge choice, as conjectured in \cite{Azouit17}.
Here, however, we prove that this is impossible.

As shown in Eq.$\,$(\ref{eq:AE_LG}), the gauge transformation induces a similarity transformation.
To our knowledge, similarity transformations of the time evolution generator have not been discussed extensively in the literature.
To demonstrate its role, therefore, we first consider the following toy example.
Suppose a time evolution equation $(d/dt)\rho = \mathcal{L}_0 (\rho)$ with
\begin{equation*}
  \mathcal{L}_0 \bullet = - i [\frac{\omega_0}{2} \sigma_z, \bullet] + \gamma_0 \mathcal{D}[\sigma_x] \bullet - \gamma_0 \mathcal{D} [\sigma_y] \bullet,
\end{equation*}
where $\omega_0$ and $\gamma_0$ are real and positive parameters. To ensure the stability of the evolution, we assume $\omega_0 > 2\gamma_0$.
The negative sign in front of $\mathcal{D}[\sigma_y]$ indicates that $\mathcal{L}_0$ is not a Lindbladian.
This negativity can be removed by the similarity transformation
\begin{equation*}
  \mathcal{L}_0' = e^{- q_0 \mathcal{D}[\sigma_x + \sigma_y]} \circ \mathcal{L}_0 \circ e^{q_0 \mathcal{D}[\sigma_x + \sigma_y]},
\end{equation*}
with $q_0$ defined by $\tanh(4q_0) = 2\gamma_0 / \omega_0$ because $\mathcal{L}_0'$ reads
\begin{equation*}
  \mathcal{L}_0' \bullet = - i [\frac{\omega_0'}{2} \sigma_z, \bullet],
\end{equation*}
with $\omega_0' = \sqrt{\omega_0^2 - 4 \gamma_0^2}$. As a result, the above similarity transformation restores the Lindblad form.

Since $\mathcal{L}_0$ and $\mathcal{L}_0'$ have the same spectrum, so do the time evolution maps $\exp( \mathcal{L}_0 t)$ and $\exp( \mathcal{L}_0' t)$.
At infinitesimal $t$, $\exp( \mathcal{L}_0' t)$ is a Kraus map, while $\exp( \mathcal{L}_0 t)$ is not.
Thus, this example shows that one cannot judge only from the spectrum whether the map is a Kraus map or not.
On the other hand, in this example, one can anticipate the existence of a Kraus map which has the same spectrum as $\exp( \mathcal{L}_0 t)$.
Indeed, the spectrum of $\exp( \mathcal{L}_0 t)$ is given by  $\{ 1,1, \exp(i \omega_0' t), \exp(-i\omega_0' t) \}$.
This implies that the time evolution merely induces a rotation of the Bloch vector without damping, which then implies that the evolution can be described by a unitary transformation in a suitable basis.
In fact, this argument can be generalized at least for qubit maps.
That is, for a given set of four numbers, $\Lambda \in \mathbb{C}^4$, one can characterize the existence of a Kraus map whose spectrum is given by $\Lambda$.
This is guaranteed by Theorem 1 of \cite{WPG10} which states the following;
\begin{theorem}
  Given $\Lambda \in \mathbb{C}^4$, the following statements are equivalent:
  \begin{itemize}
    \item There exists a Kraus map the spectrum of which is given by $\Lambda$.
    \item $\Lambda = {1} \cup \lambda$ where $\lambda \in \mathbb{C}^3$ is closed under complex conjugation.
    Furthermore, if we define $s \in \mathbb{R}^3$ by $s_i = \lambda_i$ if $\lambda_i \in \mathbb{R}$ and $s_i = |\lambda_i|$ otherwise, then
    \begin{equation}
      s \in \mathcal{T},
      \label{eq:JC_thm}
    \end{equation}
        where $\mathcal{T} \subset \mathbb{R}^3$ is the tetrahedron whose corners are $(1,1,1)$, $(1,-1,-1)$, $(-1,1,-1)$, and $(-1,-1,1)$.
  \end{itemize}
\end{theorem}

The spectrum of the time-evolution map, $\exp(\mathcal{L}_s^G t)$, is gauge invariant and is given by 
\begin{equation}
  \{ 1, e^{-t/T_2 + i \omega_B^{(4)}t}, e^{-t/T_2 - i \omega_B^{(4)} t}, e^{-t/T_1} \},
  \label{eq:JC_LsPTspec}
\end{equation}
in any gauge choice $G$.
Using the above theorem, we can show that, when $\gamma_\phi^{(4)} < 0$ and for an infinitesimal time $t$, there does not exist a Kraus map whose spectrum reads Eq.$\,$(\ref{eq:JC_LsPTspec}).
The proof is provided in Appendix \ref{app:damposc.4thL_WPG} .
As the time-evolution with a Lindblad generator is completely positive in the entire time regime, this result indicates that $\mathcal{L}_s^G$ is not a Lindbladian when $\gamma_\phi^{(4)} < 0$, in any gauge choice $G$.

In conclusion of this section, the oscillator-qubit system discussed here serves as a counterexample to the conjecture in \cite{Azouit17}.

%%%%%%%%%%%%%%%%%%%%%%%%%%%%%%%%%%%%%%%%%%%%%%%%%%%%%%%%%%%%%%%%%%%%%%%%%%%%%%%%%%%%%%%%
\section{Discussion}
\label{sec:diss}

In this section, we delve into the interpretation of the complete positivity violation.
This section is structured as follows.
Initially, a summary of the findings obtained so far is presented, together with their connections to related previous studies, in Sec.$\,$\ref{sec:diss_summary}.
To interpret the results for the partial trace parametrization, we recall that adiabatic elimination describes the evolution on the invariant manifold.
For initial states outside the invariant manifold, the short-time transient regime is neglected.
The impact of this exclusion on the properties of the generator is addressed in Sec.$\,$\ref{sec:diss_transient}.
In the presence of the interaction term $\epsilon \mathcal{L}_{\rm int}$, the invariant manifold is in general characterized by correlated states.
The evolution on the invariant manifold thus starts with an initially correlated state and the complete positivity violation can be understood from this viewpoint.
These are further elaborated in Sec.$\,$\ref{sec:diss_initial.correlated}.
At last, several remarks on the role of the gauge degree of freedom are presented in Sec.$\,$\ref{sec:diss_gauge}.

\subsection{Summary of the results in the previous sections}
\label{sec:diss_summary}

We have seen that it is the violation of complete positivity that causes a non-Lindblad form of the generator.
The generator is given by a Lindbladian if and only if the time evolution map satisfies the semigroup relation and is a Kraus map.
The semigroup relation is guaranteed because the slow dynamics is restricted on an invariant manifold exactly.
The Hermitian and trace preservations are also satisfied generally.
Complete positivity, however, is a nontrivial condition and it can be violated as we have seen in the previous sections.
We recall that not only complete positivity but even positivity of the time evolution map can be violated as shown with the qubit-qutrit example in Sec.$\,$\ref{sec:Diag}.
In this case, the density matrix acquires a negative eigenvalue depending on the initial state, and thus the result cannot be interpreted physically on $\rho_B$ alone.

In addition, we stress that the complete positivity violation is not an artifact caused by truncating the series expansion at a finite order.
With the qubit-qutrit example presented in Sec.$\,$\ref{sec:Diag}, we have seen that the negative coefficient in front of a dissipator appears even in all-order analysis.
For the oscillator-qubit example in Sec.$\,$\ref{sec:JC}, contributions higher than the fourth-order 
bring corrections to the spectrum of $\mathcal{L}_s$. However, those corrections cannot restore complete positivity. This can be seen from the condition Eq.$\,$(\ref{eq:damposc.4thL_CPinequality}) for complete positivity given in Appendix \ref{app:damposc.4thL}.
These observations differ from the argument in \cite{Muller17}, where the authors derived a non-Lindblad master equation in a fourth-order perturbation calculation applied to a Hamiltonian system and speculated that the deviation from the Lindblad form could be attributed to the truncation in the perturbation series.

As a master equation violating positivity, let alone complete positivity, the Redfield equation is widely known \cite{Breuer02,Redfield57}.
While it was originally derived for Hamiltonian systems, the authors of \cite{Damanet19} applied the same derivation procedure for a Lindblad system and found similarly a non-Lindblad master equation for the reduced dynamics.
We note that the Redfield equation is a second-order master equation.
In our settings, on the other hand, the Lindblad form is generally guaranteed in the second-order approximation.
The difference attributes to the timescale of the internal dynamics of $\mathscr{H}_B$ (the term involving $\mathcal{L}_B$ in Eq.$\,$(\ref{eq:AE_master})).
While we assume a slow timescale and treat it perturbatively, the derivation of the Redfield equation leading to positivity violations in \cite{Damanet19} assumes a fast timescale.
It is noteworthy that, even with a slow timescale, the complete positivity violation can arise in the higher-order approximation. 
The positivity violation in the Redfield equation has been explored in the literature and we discuss connections to our findings in Sec.$\,$\ref{sec:diss_transient}.

\subsection{Transient regime}
\label{sec:diss_transient}

In order to understand the origin of a non-Lindblad form, let us recall the reason why we expected a completely positive evolution in the first place.
To this end, we write the time evolution of the partial trace via that of the total state.
If we assume a product initial state with $\bar{\rho}_A$, we have
\begin{equation}
  \rho_B (t) = {\rm tr}_A \circ e^{\mathcal{L}_{\rm tot} t} (\bar{\rho}_A \otimes \rho_B(t = 0)).
  \label{eq:diss_InitialProduct}
\end{equation}
Note that ${\rm tr}_A$, $\exp( \mathcal{L}_{\rm tot} t)$, and the map that sends $\rho_B(t = 0)$ to $\bar{\rho}_A \otimes \rho_B(t = 0)$ are all completely positive. Therefore, the time evolution map from an initial state $\rho_B(t = 0)$ to $\rho_B(t)$ is completely positive.

Here it should be recalled that, in adiabatic elimination, we initialize a state on an invariant manifold.
When $\epsilon = 0$, a set of density matrices on the invariant manifold is characterized by product states as $\mathcal{K}_0^{G=0}(\rho_B) = \bar{\rho}_A \otimes \rho_B$ (see Eq.$\,$(\ref{eq:AE_asymexp0})).
In this case, Eq.$\,$(\ref{eq:diss_InitialProduct}) describes the slow dynamics on the invariant manifold, where $\bar{\rho}_A$ is fixed.
When $\epsilon > 0$, however, the invariant manifold is not anymore of the form $\{ \bar{\rho}_A \otimes \rho_B \}$ , so the initial state of Eq.$\,$(\ref{eq:diss_InitialProduct}) does not lie on the invariant manifold, and Eq.$\,$(\ref{eq:diss_InitialProduct}) includes the transient dynamics of the fast degrees of freedom.
Long after decay time of the fast sub-system, say $t \geq t_{\rm inv}$, the total state $\rho(t)$ is approximately \cite{inv.close} on the invariant manifold.
In this situation, the time evolution map calculated in adiabatic elimination is the one that sends $\rho_B(t_{\rm inv}) = {\rm tr}_A (\rho (t_{\rm inv}))$ to $\rho_B(t) \ ( t \geq t_{\rm inv})$.
There is no guarantee that this map is completely positive, even though the map that sends $\rho_B(t=0)$ to $\rho_B(t)$ is completely positive.
This situation is illustrated in Fig.$\,$\ref{fig:diss_initialization}.

%%%%%%%%%%%%%%%%%%%%%%%%%%%%%%%%%%%%%%%%%%%%%%%%%%%%%%%%%
\begin{figure}[t]
  \includegraphics[keepaspectratio, scale=0.35]{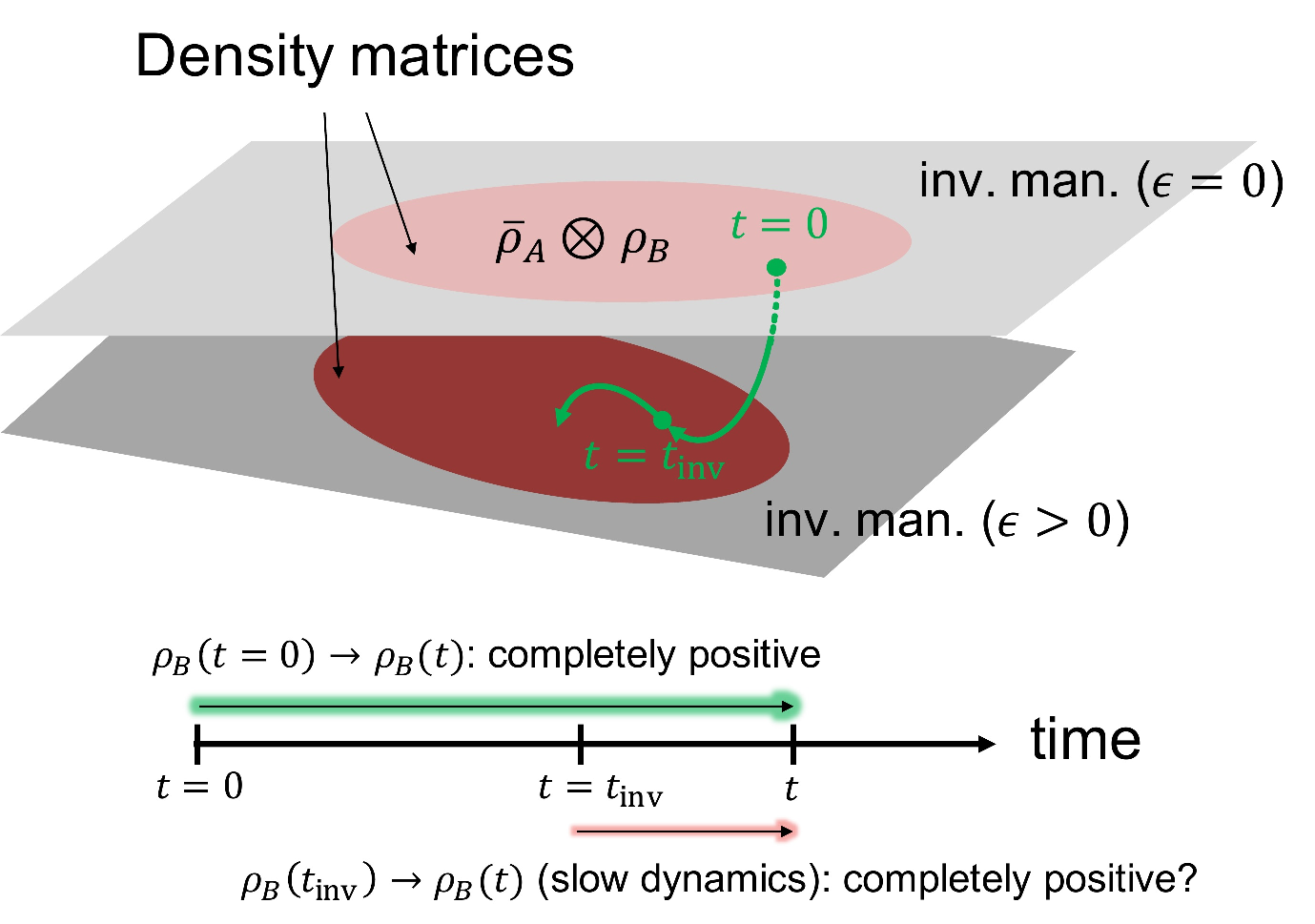}
  \caption{
  Schematic illustration of the total state evolution with an initial product state.
  On the grayish planes showing invariant manifolds with $\epsilon = 0$ (top) and $\epsilon > 0$ (bottom), the reddish areas show a set of density matrices.
  As indicated by the green arrow, the trajectory starting from a product state is first attracted to the invariant manifold with $\epsilon > 0$ (the darker grayish plane) and then is restricted there for $t \geq t_{\rm inv}$ (the slow dynamics).
  While complete positivity is ensured for the partial trace evolution including the attracting phase ($\rho_B(t = 0) \to \rho_B(t)$), it is not always the case for the slow dynamics ($\rho_B(t_{\rm inv}) \to \rho_B(t)$).
  }
  \label{fig:diss_initialization}
\end{figure}
%%%%%%%%%%%%%%%%%%%%%%%%%%%%%%%%%%%%%%%%%%%%%%%%%%%%%%%%%

In the transient regime, a master equation for the partial trace depends explicitly on time owing to the evolution of the fast degrees of freedom.
After relaxation of them, we obtain a time-independent master equation, and that is what we calculate in adiabatic elimination.
In this view, we note a similarity to nonpositivity in the Redfield equation.
The Redfield equation with the asymptotic time-independent coefficient violates (complete) positivity.
On the other hand, positivity is restored by taking into account time-dependence of the coefficient as shown in \cite{Maniscalco04,Whitney08,Hartmann20}.
The same should hold true in the composite Lindblad systems discussed in this paper.
In Appendix \ref{app:exactCP}, this expectation is confirmed for the qudit system in Sec.$\,$\ref{sec:Diag} by computing the exact master equation for the partial trace with a product initial state.
The non-Lindblad generators found in this paper can thus be considered as a signature of non-Markovianianity \cite{Rivas14,Breuer16}, in the sense that they must be associated with particular past evolution in the transient regime.
While the analytic structure of the complete positivity violation is similar between higher-order adiabatic elimination and the Redfield equation, the geometric picture is not applicable to the latter.
We discuss below that it provides us with a clear interpretation of the nonpositive evolution.

\subsection{Correlated initial states}
\label{sec:diss_initial.correlated}

To account for the initialization on the invariant manifold, Eq.$\,$(\ref{eq:diss_InitialProduct}) should be rewritten as
\begin{equation}
  \rho_B (t) = {\rm tr}_A \circ e^{\mathcal{L}_{\rm tot} t} \circ \mathcal{K}^{G=0} (\rho_B(t = 0)).
  \label{eq:diss_rhobt}
\end{equation}
From this equation and complete positivity of ${\rm tr}_A$ and $e^{\mathcal{L}_{\rm tot} t}$,
the complete positivity violation of the time evolution map stems from that of $\mathcal{K}^{G=0}$.
In what follows, we investigate properties of $\mathcal{K}^{G=0}$.
Here we concentrate on the oscillator-qubit system.
A similar analysis for a general class of settings is presented in Appendix \ref{app:K2L3_K2.positive}.

For the oscillator-qubit system, let us assume $\Delta_A = 0$ for simplicity.
Then, the expansion up to the second-order of $g/\gamma$ reads (see Appendix \ref{app:K2L3_K2.positive})
\begin{equation}
  \begin{gathered}
    \mathcal{K}^{G=0} (\rho_B) = W (\bar{\rho}_A \otimes \rho_B) W^\dagger \\
    - \Big[ \frac{4 g^2 (1 + n_{\rm th})}{\gamma^2} (I_A \otimes \sigma_-) (\bar{\rho}_A \otimes \rho_B) (I_A \otimes \sigma_-)^\dagger \\
    + \frac{4 g^2 n_{\rm th}}{\gamma^2} (I_A \otimes \sigma_+) (\bar{\rho}_A \otimes \rho_B) (I_A \otimes \sigma_+)^\dagger \Big],
  \end{gathered}
  \label{eq:diss_KG0}
\end{equation}
with $I_A$ the identity operator on $\mathscr{H}_A$ and $W = I_A \otimes I_B - (2ig/\gamma) (a^\dagger \otimes \sigma_- + a \otimes \sigma_+) - (2 g^2/\gamma^2) (a^\dagger \! a - n_{\rm th} I_A) \otimes I_B$.

Because of the minus signs in the last two lines, the second-order expansion of $\mathcal{K}^{G=0}$ is not completely positive.
It should be noted that the time evolution map can be completely positive even if $\mathcal{K}^{G=0}$ is not.
If $n_{\rm th} = 0$, for instance, $\gamma_\phi^{(4)} = 0$ and the generator $\mathcal{L}_{s}^{G=0}$ is in the Lindblad form up to the fourth-order terms, despite that $\mathcal{K}^{G=0}$ still contains a negative term.

Strikingly, $\mathcal{K}^{G=0}$ in Eq.$\,$(\ref{eq:diss_KG0}) is not even positive.
We prove in Appendix \ref{app:K2L3_K2.positive} that $\mathcal{K}^{G=0}(\rho_B)$ is not positive semidefintie for any pure reduced state $\rho_B$.
This situation can be illustrated as Fig.$\,$\ref{fig:diss_entangled}, where we introduce the following symbols; let $\mathscr{M}_{\rm inv}$ be the set of density matrices on the invariant manifold and $\mathscr{D}(\mathscr{H})$ be the set of density matrices on a Hilbert space $\mathscr{H}$.
The assignment map is $\mathcal{K}^{G=0}: \mathscr{D}(\mathscr{H}_B) \to \mathscr{M}_{\rm inv}$ and, conversely, the partial trace operation ${\rm tr}_A$ can be seen as a map ${\rm tr}_A: \mathscr{M}_{\rm inv} \to \mathscr{D}(\mathscr{H}_B)$.
We first note that, for every $\rho^* \in \mathscr{M}_{\rm inv}$, its partial trace is a density matrix,  $\rho_B^* \equiv {\rm tr}_A (\rho^*) \in \mathscr{D}(\mathscr{H}_B)$,
and we have $\rho^* = \mathcal{K}^{G=0}(\rho_B^*)$ by the construction of $\mathcal{K}^{G=0}$.
In addition to this, we have seen for the oscillator-qubit system that the second-order expansion of $\mathcal{K}^{G=0}$ is not positive.
These imply $\mathscr{M}_{\rm inv} \subset \mathcal{K}^{G=0}(\mathscr{D}(\mathscr{H}_B))$.
Similarly, while the partial trace map is injective from ${\rm tr}_A \circ \mathcal{K}^{G=0} = \mathcal{I}_B$ and the uniqueness of $\mathcal{K}^{G=0}$, it is not surjective because, for the oscillator-qubit system, any pure density matrix on $\mathscr{H}_B$ cannot be obtained by taking the partial trace of states in $\mathscr{M}_{\rm inv}$.
In other words, we have ${\rm tr}_A (\mathscr{M}_{\rm inv}) \subset \mathscr{D}(\mathscr{H}_B)$.

%%%%%%%%%%%%%%%%%%%%%%%%%%%%%%%%%%%%%%%%%%%%%%%%%%%%%%%%%
\begin{figure}[t]
  \includegraphics[keepaspectratio, scale=0.43]{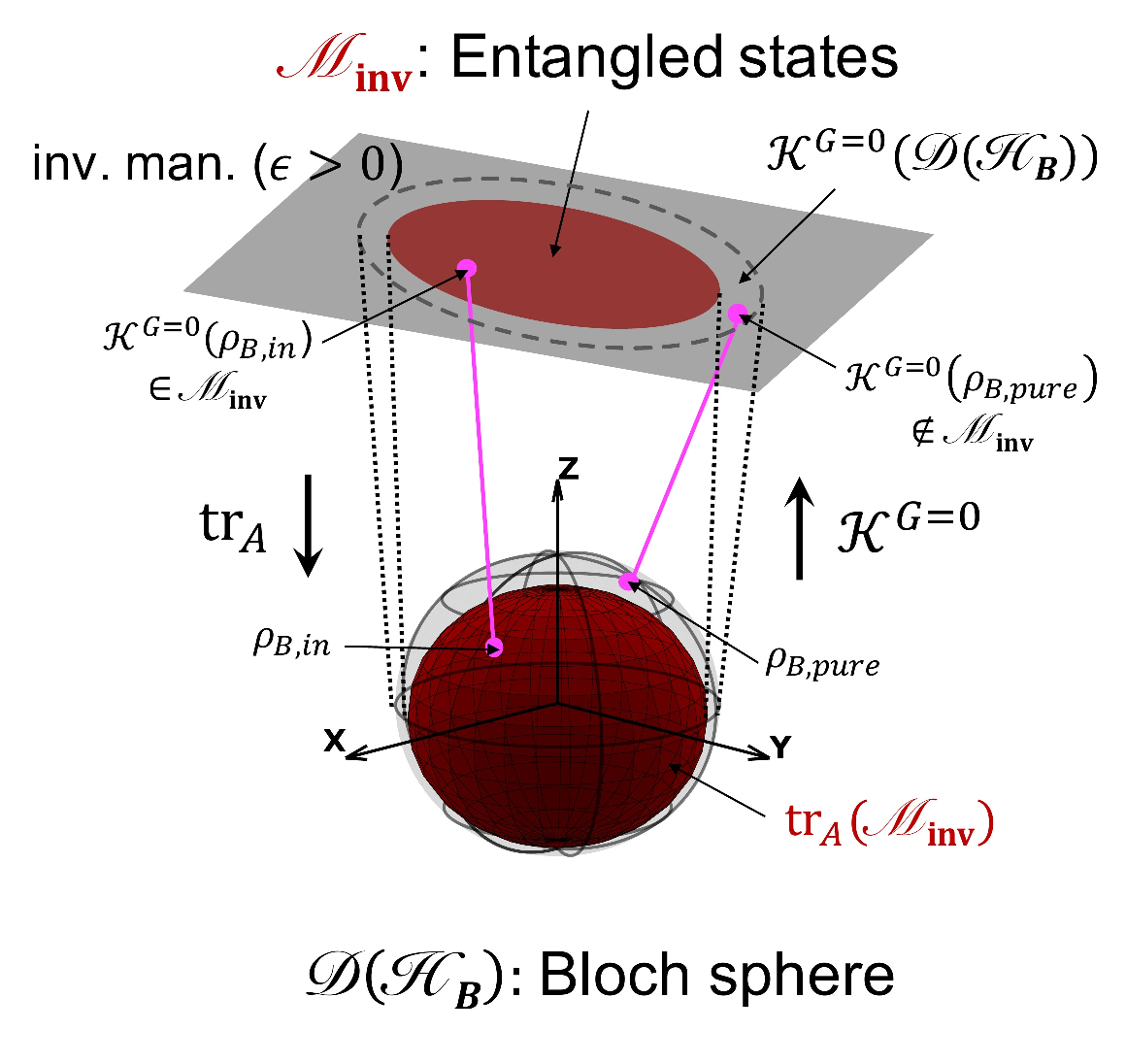}
  \caption{
  Schematic illustration of the nonsurjective property of the partial trace ${\rm tr}_A$ and the positivity violation of $\mathcal{K}^{G=0}$ for the oscillator-qubit system.
  Density matrices on the invariant manifold $\mathscr{M}_{\rm inv}$ shown by the reddish area include entangled states when $\epsilon > 0$.
  As a result, the partial trace ${\rm tr}_A (\mathscr{M}_{\rm inv})$, shown by the reddish spheroid, yields a subset of the Bloch sphere $\mathscr{D}(\mathscr{H}_B)$.
  Conversely, assignment of the whole Bloch sphere $\mathcal{K}^{G=0}(\mathscr{D}(\mathscr{H}_B))$ results in a larger set, enclosed by the dashed line on the invariant manifold, than only the positive-semidefinite matrices of $\mathscr{M}_{\rm inv}$.
  }
  \label{fig:diss_entangled}
\end{figure}
%%%%%%%%%%%%%%%%%%%%%%%%%%%%%%%%%%%%%%%%%%%%%%%%%%%%%%%%%

The positivity violation of $\mathcal{K}^{G=0}$ (the nonsurjective property of the partial trace map) can be understood from the entanglement between the two sub-systems as discussed for Hamiltonian systems in \cite{CPdebate}.
When $\epsilon = 0$, states in $\mathscr{M}_{\rm inv}$ are characterized by $\bar{\rho}_A \otimes \rho_B$ as mentioned above.
For these noncorrelated product states, any reduced state $\rho_B$ can be assigned to a valid total state.
When $\epsilon > 0$, on the other hand, states in $\mathscr{M}_{\rm inv}$ are no longer in the product form due to the interaction term.
At the second-order of $\epsilon$, they include entangled states.
Then, there exist pure reduced states that cannot be assigned to a valid total state.

This perspective also provides an interpretation of why the Lindblad form is attained including up to the second-order contribution.
This topic is further explored in Appendix \ref{app:K2L3_L3}, where we also speculate that $\mathcal{L}_{s}^{G=0}$ might admit the Lindblad form including the third-order contribution in general settings.
Paolo examined a similar problem under some assumptions and found the Lindblad form of the generator including up to the third-order contribution in a gauge choice \cite{PF}.

Complete positivity of the reduced dynamics for correlated initial states have been investigated in various Hamiltonian systems \cite{Stelmachovic01,Jordan04,McCracken13}.
A remarkable result was given in \cite{Rosario08}, in which the authors proved that vanishing quantum discord is sufficient for completely positive reduced dynamics.
Quantum discord is a measure of the quantumness of correlation in a state, that quantifies the amount of information loss when subjected to the least disturbing projective measurement \cite{Henderson01,Ollivier01,Modi13}.
In \cite{Shabani09}, it was claimed that vanishing quantum discord is also necessary for completely positive reduced dynamics.
However, subsequent research revealed counterexamples \cite{Brodutch13,Buscemi14}.
It turned out that vanishing quantum discord is equivalent to completely positive reduced dynamics only for initial states in the form
\begin{equation}
    \rho = \sum_{i,j} c_{ij} \, \phi_{ij} \otimes \ketbra{i}{j},
    \label{eq:diss_neccesary.qdiscord}
\end{equation}
with complex coefficients $\{ c_{ij} \}$, fixed environment operators $\{ \phi_{ij} \}$, and an orthonormal basis in the target system $\{ \ket{i} \}$ \cite{Shabani09}.

Note that, for the dispersive coupling system in Sec.$\,$\ref{sec:Diag}, states on the invariant manifold adhere to the configuration Eq.$\,$(\ref{eq:diss_neccesary.qdiscord}) (see Eq.$\,$(\ref{eq:Diag_KPT})).
Here we estimate quantum discord of this system to examine its potential correlations with the complete positivity violation.
Detailed methodology for estimating quantum discord through random sampling is presented in Appendix \ref{app:qdiscord}, along with a thorough discussion on the connections between the previous studies and the present study.
Fig.$\,$\ref{fig:qdiscord} compares representative values of quantum discord on the invariant manifold and the extent of complete positivity violation across various parameter sets $(\Omega,\Delta)$, with the coupling strength fixed at $\chi/\kappa = 0.1$.
The latter is estimated through the negativity of the coefficients in front of dissipators as 
\begin{equation}
    \eta(\mathscr{M}_{\rm inv}) = \frac{\sum_{i=1}^2 (|\mu_i| - \mu_i)/2}{\sum_{i=1}^2 |\mu_i|},
    \label{eq:diss_negativity}
\end{equation}
with $\mu_i \ (i=1,2)$ being the eigenvalues of the matrix $S^\top \lambda S$ in Eq.$\,$(\ref{eq:vect_LsPT_general}).
It follows that $0 \leq \eta(\mathscr{M}_{\rm inv}) \leq 1$ and that $\eta(\mathscr{M}_{\rm inv}) = 0$ if and only if the reduced dynamics are completely positive.
We see in Fig.$\,$\ref{fig:qdiscord} a qualitative agreement in the parameter dependence; quantum discord on the invariant manifold tends to be large in the parameter regime where there is a significant violation of complete positivity.
We observe a similar trend with different coupling strengths $\chi/\kappa$.
These observations are consistent with what one would expect from our interpretation of the complete positivity violation.

%%%%%%%%%%%%%%%%%%%%%%%%%%%%%%%%%%%%%%%%%%%%%%%%%%%%%%%%%
\begin{figure*}[t]
    \includegraphics[keepaspectratio, scale=0.7]{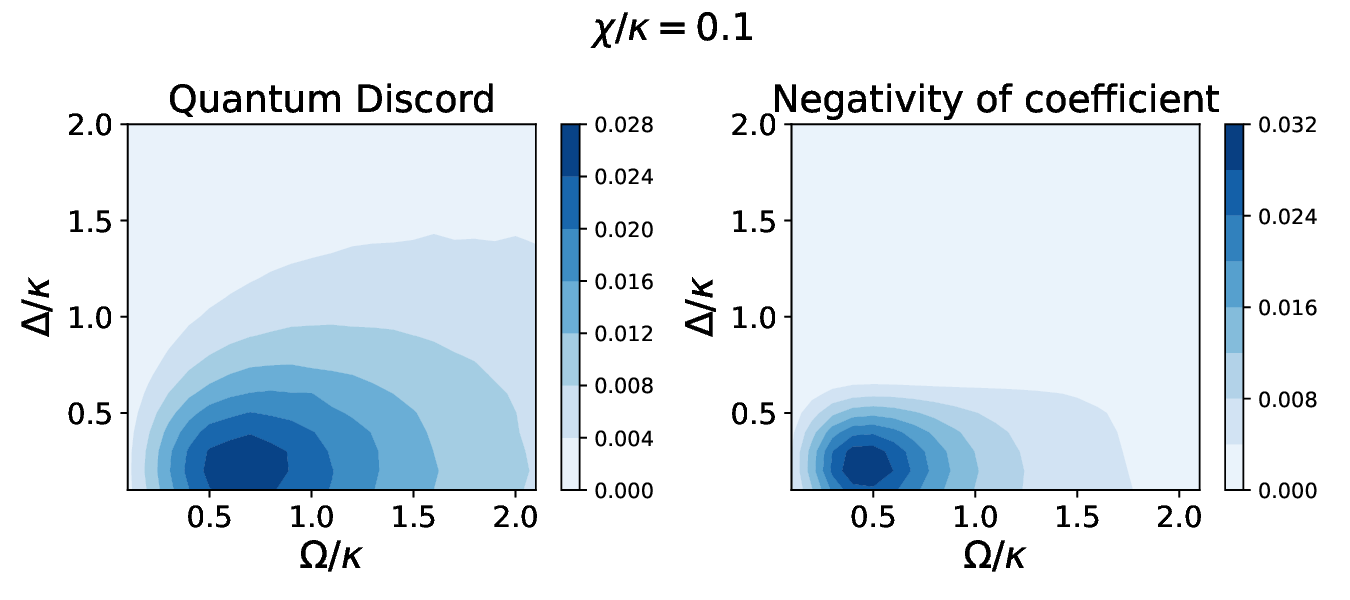}
  \caption{
  Correlation between (left) the strength of correlation on the invarinat manifold measured by quantum discord, $\delta_B(\mathscr{M}_{\rm inv})$ in Eq.$\,$(\ref{eq:qdiscord_def.inv}), and (right) the degree of the complete positivity violation measured by the negativity of the coefficients in front of dissipators, $\eta(\mathscr{M}_{\rm inv})$ in Eq.$\,$(\ref{eq:diss_negativity}).
  To obtain the left figure, we sample $N_{\rm proj} = 500$ orthogonal projectors and $N_{\rm state} = 500$ states on the invariant manifold (see Appendix \ref{app:qdiscord} for details).
  These are the results with the coupling strength $(\chi_1,\chi_2,\chi_3) = (0,\chi,2\chi)$ and $\chi/\kappa = 0.1$, the same as Fig.$\,$\ref{fig:Diag_sign}.
  }
  \label{fig:qdiscord}
\end{figure*}
%%%%%%%%%%%%%%%%%%%%%%%%%%%%%%%%%%%%%%%%%%%%%%%%%%%%%%%%%

\subsection{Role of the gauge degree of freedom}
\label{sec:diss_gauge}

So far we have concentrated on the parametrization via the partial trace.
Now let us consider how the gauge degree of freedom plays its role in the above discussion.
While the second-order expansion of $\mathcal{K}^{G=0}$ is not completely positive, there always exist gauge choices $G$ such that $\mathcal{K}^{G}$ has that property.
This was proved in Theorem 6 of \cite{AzouitThesis} and the calculation is repeated in Appendix \ref{app:K2L3_K2.geneal}.
The gauge choices restoring complete positivity of the assignment are obtained in Eq.$\,$(\ref{eq:K2L3_G2}).
For the oscillator-qubit system, one choice (setting $z = 1$ in Eq.$\,$(\ref{eq:K2L3_G2})) reads
\begin{equation*}
  G(\rho_s) = \frac{4 g^2 (1 + n_{\rm th})}{\gamma^2} \mathcal{D}[\sigma_-] (\rho_s) + \frac{4 g^2 n_{\rm th}}{\gamma^2} \mathcal{D}[\sigma_+] (\rho_s),
\end{equation*}
and $\mathcal{K}^{G}$ up to the second-order expansion is given by
\begin{equation*}
  \mathcal{K}^{G} (\rho_s) = W (\bar{\rho}_A \otimes \rho_s) W^\dagger,
\end{equation*}
where $W$ was introduced in Eq.$\,$(\ref{eq:diss_KG0}).

Even when $\mathcal{K}^G$ is completely positive, however, the Lindblad form of the generator $\mathcal{L}_s^G$ is not guaranteed.
As in Eq.$\,$(\ref{eq:diss_rhobt}), we have
\begin{equation}
  \rho_B (t) = {\rm tr}_A \circ e^{\mathcal{L}_{\rm tot} t} \circ \mathcal{K}^{G} ( \rho_s(0) ).
  \label{eq:diss_rhoBt.rhos0}
\end{equation}
Note that the map $\Upsilon_t^G = {\rm tr}_A \circ e^{\mathcal{L}_{\rm tot} t} \circ \mathcal{K}^{G}$ is completely positive, but now its input is not $\rho_B(0)$ anymore.
To recover a semigroup starting at identity of the form $\exp(\mathcal{L}_s^G t)$, we must either map $\rho_s(0)$ to $\rho_B(0)$, or map $\rho_B(t)$ to $\rho_s(t)$ in Eq.$\,$(\ref{eq:diss_rhoBt.rhos0}). The conclusions will be similar, so let us consider the evolution of $\rho_s$ by using $\rho_s(t) = (\mathcal{I}_B + G)^{-1} (\rho_B(t))$;
\begin{equation}
  \rho_s (t) = (\mathcal{I}_B + G)^{-1} \circ \Upsilon_t^G (\rho_s(0)).
  \label{eq:diss_rhost}
\end{equation}
When $\mathcal{K}^{G}$ is completely positive, so is $(\mathcal{I}_B + G)$.
However, $(\mathcal{I}_B + G)^{-1}$ is not completely positive in general, neither is the time evolution map for $\rho_s$, i.e., $(\mathcal{I}_B + G)^{-1} \circ \Upsilon_t^G$.
We have seen above that the complete positivity violation of the partial trace evolution is a natural consequence of quantum correlations.
On the other hand, we cannot judge from Eq.$\,$(\ref{eq:diss_rhost}) whether it is possible to restore complete positivity or not.
This underlines the importance of the analysis in Sec.$\,$\ref{sec:JC}.

These results add a new perspective to the debate between Pechukas and Alicki regarding complete positivity of the reduced dynamics \cite{CPdebate}.
A representation akin to $\rho_B(t) = \Upsilon_t^G (\rho_s(0))$ was considered by Alicki as an example of completely positive evolution starting from a correlated initial state.
In Pechukas' reply, it was emphasized that one should discuss the property of the map sending $\rho_B(0)$ to $\rho_B(t)$, which might not be completely positive.
This counterargument overlooks the possibility of achieving complete positivity in the map sending $\rho_s(0)$ to $\rho_s(t)$.
Such an alternative representation of the reduced dynamics $\rho_s$, which is different from the conventional partial trace $\rho_B$, naturally emerges in the geometric formulation as a different parametrization of the invariant manifold. 
The analysis in Sec.$\,$\ref{sec:JC} elucidates that achieving complete positivity through this alternative representation is, generally speaking, impossible.

In Appendix \ref{app:gauge}, further discussions are presented on the potential roles that the gauge degree of freedom plays in the practical applications of this formulation, along with the challenges encountered.

%%%%%%%%%%%%%%%%%%%%%%%%%%%%%%%%%%%%%%%%%%%%%%%%%%%%%%%%%%%%%%%%%%%%%%%%%%%%%%%%%%%%%%%%

\section{Concluding remarks}
\label{sec:conc}

Higher-order adiabatic elimination (higher than the second-order) can lead to the slow evolution of the partial trace that is not completely positive.
This is due to the fact that states on an invariant manifold include correlation between the two initial sub-systems, which imposes a restriction on the domain of proper reduced states.
Although the existence of a gauge choice restoring complete positivity was conjectured, we have shown that it is not the case for the oscillator-qubit system discussed in Sec.$\,$\ref{sec:JC}.

On experimental side, the fact that a higher-order reduced model (for the partial trace) can potentially lead to unphysical results can be understood as follows.
A first nontrivial point is to prepare an initial state within the slow invariant manifold.
Subsequently, the evolution of this state could be tracked to deduce rates related to $\mathcal{L}_s^{G=0}$.
A reduced model described by a non-Lindblad generator then implies that, even if the invariant manifold fits the dimensions of e.g. a qubit, we cannot interpret it as a standalone decaying qubit.
To test this, one way is to measure rates at which the slow dynamics decays, such as $T_1$ and $T_2$ relaxation times for a qubit reduced system.
While those values are positive from the stability condition, we expect the violation of relations among them imposed by the complete positivity condition, such as $T_1 \geq T_2/2$.

In closing, we discuss two further questions raised in this paper.
First, the oscillator-qubit system discussed in Sec.$\,$\ref{sec:JC} is so far the only example where we can rigorously prove the impossibility of restoring complete positivity by the gauge transformation.
It is not clear how general this conclusion is. As discussed underneath Eq.$\,$(\ref{eq:diss_rhost}), we have not yet developed an intuitive understanding of such impossibility.
Our proof utilizes the theorem in \cite{WPG10}, which is applicable only to qubit maps.
Future work in this direction would provide us with a new insight into a constraint on the spectrum of quantum maps in general.
A second question is whether it is possible to find a gauge choice leading to a positive and surjective assignment map $\mathcal{K}^G$, which is relevant to practical applications as detailed in Appendix \ref{app:gauge}.
The existence of a such a gauge choice has so far been confirmed only for a dispersively coupled two qubit system \cite{Alain20}.
To prove such surjectivity, we need to characterize density matrices in the total system.
This problem is already complicated for 3-dimensional systems \cite{Kimura03} and a new approach is warranted.

%%%%%%%%%%%%%%%%%%%%%%%%%%%%%%%%%%%%%%%%%%%%%%%%%%%%%%%%%%%%%%%%%%%%%%%%%%%%%%%%%%%%%%%%

\begin{acknowledgments}
  This work has been supported by the Engineering for Quantum Information Processors (EQIP) Inria challenge project, the European Research Council (ERC) under the European Union's Horizon 2020 research and innovation programme (grant agreement No. [884762]), French Research Agency through the ANR grant HAMROQS, and by Plan France 2030 through the project ANR-22-PETQ-0006.
\end{acknowledgments}

%%%%%%%%%%%%%%%%%%%%%%%%%%%%%%%%%%%%%%%%%%%%%%%%%%%%%%%%%%%%%%%%%%

\appendix

%%%%%%%%%%%%%%%%%%%%%%%%%%%%%%%%%%%%%%%%%%%%%%%%%%%%%%%%%%%%%%%%%%%%%%%%%%%%%%%%%%%%%%%%

\section{Perturbation solution to the invariance condition}
\label{app:inv.cond.}

In this appendix, we present a way to evaluate the higher-order contributions in the asymptotic expansions Eq.$\,$(\ref{eq:AE_asymexp}).
Inserting the expansions into the invariance condition Eq.$\,$(\ref{eq:AE_inv1}), the $n$-th order of $\epsilon$ reads
\begin{equation}
  \bar{\rho}_A \otimes \mathcal{L}_{s,n} (\rho_s) = \mathcal{L}_A \otimes \mathcal{I}_B (\mathcal{K}_{n} (\rho_s) )) + \mathcal{L}_n(\rho_s),
  \label{eq:inv.cond._inv2}
\end{equation}
where we have introduced
\begin{equation*}
  \begin{gathered}
    \mathcal{L}_n (\rho_s) = (\mathcal{I}_A \otimes \mathcal{L}_B + \mathcal{L}_{\rm int}) (\mathcal{K}_{n-1}(\rho_s)) \\
    - \sum_{k=1}^{n-1} \mathcal{K}_{k} (\mathcal{L}_{s,n-k}(\rho_s)).
  \end{gathered}
\end{equation*}
Since ${\rm tr}_A ( \mathcal{L}_A \bullet ) = 0$, the partial trace over $\mathscr{H}_A$ leads to $\mathcal{L}_{s,n}$ in the form
\begin{equation}
  \mathcal{L}_{s,n} (\rho_s) = {\rm tr}_A ( \mathcal{L}_n(\rho_s) ).
  \label{eq:inv.cond._Ln}
\end{equation}
Equation (\ref{eq:inv.cond._inv2}) can then be rearranged with known terms on the right-hand side;
\begin{equation}
  \mathcal{L}_A \otimes \mathcal{I}_B (\mathcal{K}_{n} (\rho_s)) =  \bar{\rho}_A \otimes \mathcal{L}_{s,n} (\rho_s) - \mathcal{L}_n(\rho_s).
  \label{eq:AE_LAKn}
\end{equation}
To solve this linear equation for $\mathcal{K}_{n} (\rho_s)$, one needs to invert $\mathcal{L}_A$.
Note that $\mathcal{L}_A$ is singular since one of the eigenvalues is zero.
Thus, this linear equation is underdetermined.
By identifying the Kernel of $\mathcal{L}_A \otimes \mathcal{I}_B$, $\mathcal{K}_{n} (\rho_s)$ is determined only up to $\bar{\rho}_A \otimes {\rm tr}_A (\mathcal{K}_{n} (\rho_s))$ by solving (\ref{eq:AE_LAKn}).
We introduce the undetermined part, $G_n = {\rm tr}_A \circ \mathcal{K}_{n}$, which can be any superoperator on $\mathscr{H}_B$.
This gauge degree of freedom is associated with the nonuniqueness of the parametrization.
With $G_n (\rho_s)$, $\mathcal{K}_{n} (\rho_s)$ reads
\begin{equation}
  \begin{gathered}
    \mathcal{K}_{n} (\rho_s) = \mathcal{L}_A^{+} \otimes \mathcal{I}_B (\bar{\rho}_A \otimes \mathcal{L}_{s,n} (\rho_s) - \mathcal{L}_n(\rho_s)) \\
    + \bar{\rho}_A \otimes G_n(\rho_s),
  \end{gathered}
  \label{eq:AE_Kn}
\end{equation}
where $\mathcal{L}_A^{+}$ is the Moore-Penrose inverse of $\mathcal{L}_A$.
One way to calculate $\mathcal{L}_A^{+}$ by making use of the fact that the evolution only with $\mathcal{L}_A$ exponentially converges to a unique steady state was discussed in \cite{Azouit17}.
We revisit their results using the vectorization method in Appendix \ref{app:vect}.
Notice that the definition of the gauge superoperator is different from that in \cite{Azouit17}.
In this paper, the choice $G_n = 0 \ (n=1,2,\dots)$ corresponds to the parametrization via the partial trace $\rho_s = \rho_B$.

From Eqs.$\,$(\ref{eq:inv.cond._Ln}) and (\ref{eq:AE_Kn}), we can obtain $\mathcal{K}_n$ and $\mathcal{L}_{s,n}$ up to a desired order.
From these equations, we see that $\mathcal{K}_n$ and $\mathcal{L}_{s,n+1}$ depend on $G_1, \dots, G_n$.
Therefore, $\mathcal{K}_n$ and $\mathcal{L}_{s,n+1}$ become, for instance, nonlinear functions of $\rho_s$ if so is even one of $G_1, \dots, G_n$.
To meet the conditions that $\mathcal{K}$ and $\mathcal{L}_s$ are linear and time-independent, we assume that $\{ G_n \}$ have those properties.

%%%%%%%%%%%%%%%%%%%%%%%%%%%%%%%%%%%%%%%%%%%%%%%%%%%%%%%%%%%%%%%%%%%%%%%%%%%%%%%%%%%%%%%%

\section{Adiabatic elimination for a dispersively coupled qudit systems}
\label{app:ae.qudit}

\subsection{All-order adiabatic elimination}
\label{app:ae.qudit_all}

In this appendix, we present the details of the adiabatic elimination calculations for the dispersively coupled qudit system introduced in Sec.$\,$\ref{sec:Diag}.
As discussed in \cite{Alain20}, the invariance condition can be solved without resorting to the asymptotic expansion.
To see this, we recall Eq.$\,$(\ref{eq:Diag_LAmn}), that is, 
\begin{equation*}
  \mathcal{L}_{\rm tot} ( A \otimes E_{m,n} ) = \mathcal{L}_A^{(m,n)}(A) \otimes E_{m,n}.
\end{equation*}
This indicates that the eigenvalue problem of $\mathcal{L}_{\rm tot}$ can be formally solved as
\begin{equation}
  \mathcal{L}_{\rm tot} (Q_{m,n}^{(k)} \otimes E_{m,n}) = \lambda_{m,n}^{(k)} (Q_{m,n}^{(k)} \otimes E_{m,n}),
  \label{eq:ae.qudit_eig_tot}
\end{equation}
for $m,n = 1,\dots,d$ and $k = 1,\dots,d_A^2$, where $\{ Q_{m,n}^{(k)} \}$ satisfy
\begin{equation}
  \mathcal{L}_A^{(m,n)} (Q_{m,n}^{(k)}) = \lambda_{m,n}^{(k)} Q_{m,n}^{(k)}.
  \label{eq:ae.qudit_eig}
\end{equation}
Note $(\lambda_{m,n}^{(k)})^* = \lambda_{n,m}^{(k)}$ and $(Q_{m,n}^{(k)})^\dagger = Q_{n,m}^{(k)}$.
The stability condition reads ${\rm Re}(\lambda_{m,n}^{(k)}) \leq 0$ with ${\rm Re}$ denoting real part.
Given $(m,n)$, the eigenoperator contributing to the slow dynamics is the one with the smallest value of ${\rm Re}(-\lambda_{m,n}^{(k)})$.
For each set $(m,n)$, we rearrange $\{ k \}$ so that $k = 1$ has this property and denote $\lambda_{m,n} = \lambda_{m,n}^{(1)}$ and $Q_{m,n} = Q_{m,n}^{(1)}$ in the following.
In addition, we normalize $Q_{m,n}$ as ${\rm tr}_A(Q_{m,n}) = 1$.
From Eq.$\,$(\ref{eq:ae.qudit_eig}), $Q_{m,m}$ is an eigenoperator of a Lindbladian.
Since the eigenvalues of Lindbladians include $0$ due to the trace preserving property,
we have $\lambda_{m,m} = 0$ for $m = 1,\dots,d$.
To summarize, we have the following properties of the eigenvalues,
\begin{equation}
  \lambda_{m,n}^* = \lambda_{n,m}, \ \ \ \lambda_{m,m} = 0, \ \ \ {\rm Re}(-\lambda_{m,n}) \geq 0,
  \label{eq:ae.quditDiag_lmd.prop}
\end{equation}
and of the eigenoperators,
\begin{equation*}
  Q_{m,n}^\dagger = Q_{n,m}, \ \ \ {\rm tr}_A(Q_{m,n}) = 1.
\end{equation*}

When the inequality
\begin{equation}
  \max_{m,n}{\rm Re}(-\lambda_{m,n}) < \min_{k>1,m,n} {\rm Re}(-\lambda_{m,n}^{(k)}),
  \label{eq:ae.quditDiag_ineq}
\end{equation}
is satisfied, the modes with $k>1$ decay faster than those with $k=0$.
In this case, the total density matrix after the decay of the modes with $k>1$ can be expressed as
\begin{equation*}
  \rho = \sum_{m,n = 1}^d [\rho_B]_{m,n} \ Q_{m,n} \otimes E_{m,n},
\end{equation*}
where $[\rho_B]_{m,n}$ is the $mn$-element of $\rho_B$
the evolution of which is given by
\begin{equation*}
  \frac{d}{dt} [\rho_B]_{m,n} = \lambda_{m,n} \ [\rho_B]_{m,n}.
\end{equation*}
Therefore, the maps $\mathcal{K}$ and $\mathcal{L}_s$ associated with the partial trace parametrization ($G=0$) read
\begin{equation*}
  \mathcal{K}^{G=0} (\rho_B) = \sum_{m,n = 1}^d Q_{m,n} \otimes E_m \rho_B E_n,
\end{equation*}
and
\begin{equation*}
  \mathcal{L}_s^{G=0} (\rho_B) = \sum_{m,n = 1}^d \lambda_{m,n} \ E_m \rho_B E_n.
\end{equation*}

\subsection{Fourth-order adiabatic elimination}
\label{app:ae.qudit_4th}

Using the result in Appendix \ref{app:ae.qudit_all}, we found in Fig.$\,$\ref{fig:Diag_sign} the presence of a parameter region where the generator $\mathcal{L}_s^{G=0}$ is not in the Lindblad form.
The above all-order analysis, however, requires numerical computation of the eigenvalues $\{ \lambda_{m,n} \}$.
It is then unclear whether the non-Lindblad form is obtained even with infinitesimal coupling constants.
To make clear the coupling constant dependence, we here present the results of fourth-order adiabatic elimination.

The setting considered here is presented in Sec.$\,$\ref{sec:Diag}.
The interaction is assumed to be given by Eq.$\,$(\ref{eq:Diag_int}).
To make the coupling constant dependence explicit, we rewrite it as
\begin{equation*}
  \epsilon \mathcal{L}_{\rm int} \bullet = i [\chi ( V_A \otimes B ), \bullet],
\end{equation*}
with $B = \sum_{m=1}^d (\chi_m/\chi) E_m$ and $\chi \in \mathbb{R}_{>0}$ the characteristic coupling strength such that $ \chi_m/\chi  \ (m = 1,\dots,d)$ are in the order of unity.
We consider adiabatic elimination in a frame where the qudit internal dynamics becomes trivial as $\mathcal{L}_B = 0$.
The sub-system $\mathscr{H}_A$ is assumed to be a qubit system described by Eq.$\,$(\ref{eq:Diag_qubit}).
The timescale separation parameter $\epsilon$ then reads $\epsilon = \chi/\kappa$.

We assume that the eigenvalue problem for the generator $\mathcal{L}_A$ is solvable.
Although it might require numerical computation, we note that the results are independent of $\chi$ the coupling constant.
In the vectorized form, we introduce the following notations (see Eq.$\,$(\ref{eq:vect_eig}))
\begin{equation*}
  \hat{\mathcal{L}}_A \dket{X_\alpha} = \kappa \nu_\alpha \dket{X_\alpha}, \ \ \ \dbra{\bar{X}_\alpha} \hat{\mathcal{L}}_A = \kappa \nu_{\alpha} \dbra{\bar{X}_\alpha},
\end{equation*}
for $\alpha = 0,1,2,3$.
Since there exists a unique steady state, we set $\{ \nu_\alpha \}$ such that $\nu_0 = 0$ and ${\rm Re}(\nu_\alpha) < 0 \ (\alpha = 1,2,3)$.
Assuming the linear independence of $\{ X_\alpha \}$, we can obtain the Moore-Penrose inverse of $\hat{\mathcal{L}}_A$ as (see Eq.$\,$(\ref{eq:vect_MPinverse}))
\begin{equation*}
  \hat{\mathcal{L}}_A^+ = \sum_{\alpha = 1}^{3} \frac{1}{\kappa \nu_\alpha} \dketbra{X_\alpha}{\bar{X}_\alpha}.
\end{equation*}

With the Moore-Penrose inverse, we can now calculate the higher-order contributions.
For the partial trace parametrization, $\mathcal{L}_s^{G=0}$ up to the fourth-order expansion reads
\begin{equation}
  \begin{gathered}
    \mathcal{L}_s^{G=0} (\rho_B) = - i \epsilon x_1 B \rho_B + \epsilon^2 x_2 [B \rho_B, B] \\
    + i \epsilon^3 [x_3 B \rho_B B + y_3 \rho_B B^2, B] \\
    + \epsilon^4 [x_4 B^2 \rho_B B - y_4 B^3 \rho_B, B] \\
    + (h.c.),
  \end{gathered}
  \label{eq:app:ae.qudit_LsG0}
\end{equation}
where $(h.c.)$ means the Hermitian conjugate of all the terms prior.
The coefficients $x_k \ (k = 1,2,3,4)$ and $y_k \ (k = 3,4)$, which are independent of $\chi$ the coupling constant, are defined by
\begin{equation*}
  x_1 / \kappa = (V_A)_0,
\end{equation*}
\begin{equation*}
  \begin{gathered}
    x_2 / \kappa
    = - \sum_{\alpha=1}^{3} (V_A)_{0,\bar{\alpha}}^* (V_A)_\alpha \frac{1}{\nu_\alpha^*},
  \end{gathered}
\end{equation*}
\begin{equation*}
  \begin{gathered}
    x_3 / \kappa
    = - (V_A)_{0} \sum_{\alpha=1}^{3} (V_A)_{0,\bar{\alpha}}^* (V_A)_\alpha \frac{1}{(\nu_\alpha^2)^*} \\
    + \sum_{\alpha,\beta=1}^{3} (V_A)_{0,\bar{\alpha}} (V_A)_{\alpha,\bar{\beta}}^* (V_A)_\beta \frac{1}{\nu_\alpha \nu_\beta^*},
  \end{gathered}
\end{equation*}
\begin{equation*}
  \begin{gathered}
    y_3 / \kappa
    = - (V_A)_{0} \sum_{\alpha=1}^{3} (V_A)_{\bar{\alpha},0}^* (V_A)_\alpha \frac{1}{(\nu_\alpha^2)^*} \\
    + \sum_{\alpha,\beta=1}^{3} (V_A)_{0,\bar{\alpha}} (V_A)_{\bar{\beta},\alpha}^* (V_A)_\beta \frac{1}{\nu_\alpha \nu_\beta^*},
  \end{gathered}
\end{equation*}
\begin{equation*}
  \begin{gathered}
    x_4 / \kappa = \\
    \sum_{\alpha=1}^3 \Big[ 2 {\rm Re} \Big( (V_A)_{0,\bar{\alpha}} (V_A)_\alpha^* \frac{x_2}{\nu_\alpha^2} \Big) + (V_A)_{\bar{\alpha},0} (V_A)_\alpha^* \frac{x_2}{\nu_\alpha^2} \Big] \\
    - (V_A)_{0}^2 \sum_{\alpha = 1}^{3} \Big[ 2 (V_A)_{0,\bar{\alpha}}^* (V_A)_\alpha \frac{1}{(\nu_\alpha^*)^3}
    + (V_A)_{0,\bar{\alpha}} (V_A)_\alpha^* \frac{1}{\nu_\alpha^3} \Big] \\
    + \sum_{\alpha,\beta = 1}^{3} \Big\{ 2 {\rm Re} \Big( (V_A)_{0} (V_A)_{0,\bar{\alpha}}^* (V_A)_{\alpha,\bar{\beta}} (V_A)_{\beta}^* \frac{\nu_\alpha^* + \nu_\beta}{(\nu_\alpha^* \nu_\beta)^2} \Big) \\
    + (V_A)_{0} (V_A)_{0,\bar{\alpha}}^* (V_A)_{\bar{\beta},\alpha} (V_A)_{\beta}^* \frac{\nu_\alpha^* + \nu_\beta}{(\nu_\alpha^* \nu_\beta)^2} \Big\} \\
    - \sum_{\alpha,\beta,\gamma = 1}^{3} \Big\{ (V_A)_{0,\bar{\alpha}}^* (V_A)_\gamma \Big[ (V_A)_{\alpha,\bar{\beta}} (V_A)_{\beta,\bar{\gamma}}^* + (V_A)_{\bar{\beta},\alpha} (V_A)_{\bar{\gamma},\beta}^* \Big] \\
    \times \frac{1}{\nu_\alpha^* \nu_\beta \nu_\gamma^*} + \Big[ (V_A)_{0,\bar{\alpha}}^* (V_A)_{\alpha,\bar{\beta}} (V_A)_{\bar{\gamma},\beta}^* (V_A)_\gamma \frac{1}{\nu_\alpha^* \nu_\beta \nu_\gamma^*} \Big]^* \Big\},
  \end{gathered}
\end{equation*}
\begin{equation*}
  \begin{gathered}
    y_4 / \kappa
    = \sum_{\alpha=1}^{3} (V_A)_{\bar{\alpha},0} (V_A)_\alpha^* \frac{x_2}{\nu_\alpha^2} \\
    - (V_A)_{0}^2 \sum_{\alpha = 1}^{3} (V_A)_{0,\bar{\alpha}}^* (V_A)_\alpha \frac{1}{(\nu_\alpha^*)^3} \\
    + \sum_{\alpha,\beta = 1}^{3} (V_A)_{0} (V_A)_{0,\bar{\alpha}}^* (V_A)_{\bar{\beta},\alpha} (V_A)_{\beta}^* \frac{\nu_\alpha^* + \nu_\beta}{(\nu_\alpha^* \nu_\beta)^2} \\
    - \sum_{\alpha,\beta,\gamma = 1}^{3} (V_A)_{0,\bar{\alpha}}^* (V_A)_{\bar{\beta},\alpha} (V_A)_{\beta,\bar{\gamma}}^* (V_A)_\gamma \frac{1}{\nu_\alpha^* \nu_\beta \nu_\gamma^*},
  \end{gathered}
\end{equation*}
where we have introduce the notation,
\begin{equation*}
  (V_A)_{\alpha,\bar{\beta},\gamma,\bar{\sigma}} = \dbraket{ X_\alpha \bar{X}_\beta X_\gamma \bar{X}_\sigma}{V_A},
\end{equation*}
for the matrix elements.

To investigate the structure of $\mathcal{L}_s^{G=0}$ given by Eq.$\,$(\ref{eq:app:ae.qudit_LsG0}), we rewrite it as follows, which is valid up to order $\epsilon^4$;
\begin{equation}
  \mathcal{L}_s^{G=0} \bullet = -i [h_B, \bullet]
  + \Big[ \epsilon^2 \mathcal{D}[l_B] \bullet + \epsilon^4 c_B \mathcal{D}[l_B^2] \bullet \Big].
  \label{eq:app:ae.qudit_LsG0_1}
\end{equation}
The first term includes the Hermitian operator $h_B$ defined by
\begin{equation*}
  \begin{gathered}
    h_B = \epsilon x_1 B  + \epsilon^2 {\rm Im} (x_2) B^2 \\
    - \epsilon^3 {\rm Re} (y_3) B^3 - \epsilon^4 {\rm Im} (y_4) B^4.
  \end{gathered}
\end{equation*}
In the terms inside the square bracket, the operator $l_B$ is defined by
\begin{equation*}
  l_B = c_1 B + \epsilon c_2 B^2 + \epsilon^2 c_3 B^3,
\end{equation*}
with $|c_1|^2 = 2 {\rm Re} (x_2)$, $c_2 = - i ({y_3}^* + 2 {\rm Re} (x_3))/c_1^*$, and $c_3 = - (x_4 + y_4)/c_1^*$.
The assumption ${\rm Re} (x_2) > 0$ has been made, which is justified since the second-order expansion always admits the Lindblad form as shown in \cite{Azouit17}.
With these $\{ c_k \}$, we have introduced $c_B = |c_1|^{-4} (2 {\rm Re} (x_4) - |c_2|^2) \in \mathbb{R}$.
Now, we denote $l_B$ as $l_B = \sum_{m=1}^d l_m E_m$ with $l_m = c_1 (\chi_m/\chi) + \epsilon c_2 (\chi_m/\chi)^2 + \epsilon^2 c_3 (\chi_m/\chi)^3$.
The terms inside the square bracket in Eq.$\,$(\ref{eq:app:ae.qudit_LsG0_1}) then reads as
\begin{equation*}
  \begin{gathered}
    \epsilon^2 \mathcal{D}[l_B] (\rho_B) + \epsilon^4 c_B \mathcal{D}[l_B^2] (\rho_B) \\
    = \sum_{m,n = 1}^d \beta_{m,n} \Big[ E_m \rho_B E_n - \frac{1}{2} (E_n E_m \rho_B + \rho_B E_n E_m) \Big],
  \end{gathered}
\end{equation*}
with $\beta_{m,n} = \epsilon^2 l_m l_n^* (1 + \epsilon^2 c_B l_m l_n^* )$.
Since this has a similar form to Eq.$\,$(\ref{eq:vect_LsPT2}), as discussed underneath Eq.$\,$(\ref{eq:vect_LsPT_general}),
$\mathcal{L}_s^{G=0}$ is a Lindbladian if and only if the condition $S^\top \beta S \geq 0$ is satisfied,
where the matrix $\beta$ is defined by $[\beta]_{m,n} = \beta_{m,n}$ and the matrix $S$ is defined in Eq.$\,$(\ref{eq:vect_MatrixS}).

As in Sec.$\,$\ref{sec:Diag}, we consider the case $d=3$ as an example.
When $\epsilon \ll 1$, the trace of the matrix $S^\top \beta S$ is dominated by the terms at order $\epsilon^2$, which are positive.
The condition $S^\top \beta S \geq 0$ then reads
\begin{equation*}
 \epsilon^4 c_B | (l_1 - l_2)(l_1 - l_3)(l_2 - l_3) |^2 \geq 0,
\end{equation*}
which is equivalent to $c_B \geq 0$.
Since $c_B$ is independent of $\chi$ the coupling constant, we can now discuss the Lindblad structure of $\mathcal{L}_s^{G=0}$ with infinitesimal coupling constants.
By numerically computing $c_B$ with the same settings as in Fig.$\,$\ref{fig:Diag_sign}, where the linear independence of the eigenvectors $\{ X_\alpha \}$ was numerically confirmed, we examined its sign in various parameter sets.
As a result, we found that the results closely resembled those depicted in Fig.$\,$\ref{fig:Diag_sign}.
Therefore, in a parameter region near the blue region of Fig.$\,$\ref{fig:Diag_sign}, $\mathcal{L}_s^{G=0}$ is not a Lindbladian, even with infinitesimal coupling constants.

%%%%%%%%%%%%%%%%%%%%%%%%%%%%%%%%%%%%%%%%%%%%%%%%%%%%%%%%%%%%%%%%%%%%%%%%%%%%%%%%%%%%%%%%
\section{Vectorization}
\label{app:vect}

\subsection{Introduction to vectorization}

In this appendix, we introduce a vectorization of operators, which is a convenient representation in studies of open quantum systems.
Let $A$ be an operator acting on a Hilbert space $\mathscr{H}$ with dimension $d$ and have a matrix representation as
\begin{equation*}
  A = \sum_{i,j = 1}^d (A)_{ij} \ketbra{i}{j},
\end{equation*}
with an orthonormal basis set $\{ \ket{i} \}_{i=1,\dots,d}$ in $\mathscr{H}$.
Following \cite{Havel03}, we map this operator to a vector as
\begin{equation*}
  \dket{A} \equiv \sum_{i,j = 1}^d (A)_{ij} \ket{j} \otimes \ket{i}.
\end{equation*}
We also introduce the dual state as the Hermitian conjugation of $\dket{A}$,
\begin{equation*}
  \dbra{A} \equiv [\dket{A}]^\dagger = \sum_{i,j = 1}^d (A)_{ij}^* \bra{j} \otimes \bra{i}.
\end{equation*}
With $B$ and $C$ being operators acting on $\mathscr{H}$, we can show the Hilbert-Schmidt inner product as
\begin{equation*}
  \dbraket{A}{B} = {\rm tr}(A^\dagger B),
\end{equation*}
with ${\rm tr}$ the trace operation over $\mathscr{H}$, and
\begin{equation}
  \dket{ABC} = C^\top \otimes A \, \dket{B},
  \label{eq:vect_ABC}
\end{equation}
where $\top$ denotes the matrix transpose.

\subsection{Moore-Penrose inverse of a Lindbladian with a unique steady state}

Next we consider superoperators in the vectorized form.
As a vectorized operator is a $d^2$-dimensional vector, a superoperator is represented by a $d^2$-dimensional matrix.
We denote such supermatrix (the name is taken from \cite{Havel03}) by attaching the hat symbol $\hat{}$, that is, if $\mathcal{L}$ is a superoperator acting on $\mathscr{H}$, we denote
\begin{equation}
  \mathcal{L}(A) \to \dket{\mathcal{L}(A)} = \hat{\mathcal{L}} \dket{A}.
  \label{eq:vect_supermatrix}
\end{equation}
For instance, if $\mathcal{L}(B) = ABC$, then $\hat{\mathcal{L}} = C^\top \otimes A$ as shown in Eq.$\,$(\ref{eq:vect_ABC}).
Suppose now that the eigenvalue problem of $\hat{\mathcal{L}}$ is solved as
\begin{equation}
  \hat{\mathcal{L}} \dket{X_\alpha} = \lambda_\alpha \dket{X_\alpha}, \ \ \ \dbra{\bar{X}_\alpha} \hat{\mathcal{L}} = \lambda_{\alpha} \dbra{\bar{X}_\alpha},
  \label{eq:vect_eig}
\end{equation}
for $\alpha = 0,1,\dots,d^2-1$. If $\{ \dket{X_\alpha} \}_{\alpha = 0,1,\dots,d^2-1}$ are linearly independent and are normalized so that the orthonormality relations $\dbraket{\bar{X}_\alpha}{X_{\beta}} = \delta_{\alpha,\beta} \ (\alpha,\beta = 0,1,\dots,d^2-1)$ hold, we have the completeness relation
\begin{equation}
  \sum_{\alpha=0}^{d^2-1} \dketbra{X_\alpha}{\bar{X}_\alpha} = I_d \otimes I_d,
  \label{eq:vect_complete}
\end{equation}
with the $d$-dimensional identity matrix $I_d$.
The spectral decomposition of $\hat{\mathcal{L}}$ then reads
\begin{equation}
  \hat{\mathcal{L}} = \sum_{\alpha=0}^{d^2-1} \lambda_\alpha \dketbra{X_\alpha}{\bar{X}_\alpha}.
  \label{eq:vect_eigdecomp}
\end{equation}

Let us assume $\lambda_0 = 0$ and ${\rm Re}(\lambda_\alpha) < 0 \ (\alpha = 1,\dots,d^2-1)$, as we have assumed for $\mathcal{L}_A$ in Sec.$\,$\ref{sec:AE}.
In this case, the Moore-Penrose inverse of $\hat{\mathcal{L}}$, $\hat{\mathcal{L}}^+$, can be written as
\begin{equation}
  \hat{\mathcal{L}}^+ = \sum_{\alpha = 1}^{d^2-1} \frac{1}{\lambda_\alpha} \dketbra{X_\alpha}{\bar{X}_\alpha}.
  \label{eq:vect_MPinverse}
\end{equation}
To derive another expression, let us use $1/\lambda_{\alpha} = - \int_0^\infty ds \, \exp(\lambda_\alpha s)$ for $\alpha = 1,\dots,d^2-1$,
\begin{equation*}
  \begin{gathered}
    \hat{\mathcal{L}}^+ = - \sum_{\alpha=1}^{d^2-1} \int_0^\infty ds \, e^{\lambda_\alpha s} \dketbra{X_\alpha}{\bar{X}_\alpha} \\
    = - \int_0^\infty ds \, e^{\hat{\mathcal{L}} s} \sum_{\alpha=1}^{d^2-1} \dketbra{X_\alpha}{\bar{X}_\alpha}.
  \end{gathered}
\end{equation*}
Using the completeness relation Eq.$\,$(\ref{eq:vect_complete}), we obtain
\begin{equation*}
  \hat{\mathcal{L}}^+ = - \int_0^\infty ds \, e^{\hat{\mathcal{L}} s} (I_d \otimes I_d - \dketbra{X_0}{\bar{X}_0} ),
\end{equation*}
or its operation given by
\begin{equation}
  \mathcal{L}^+ (A) = - \int_0^\infty ds \, e^{\mathcal{L} s} (A - {\rm tr} (\bar{X}_0^\dagger A) X_0 ).
  \label{eq:vect_MPinverse_exp_op}
\end{equation}
This expression was derived in \cite{Azouit17}.

\subsection{Choi matrix}
\label{app:vect_choi}

Given a superoperator $\mathcal{T}$, we introduce the Choi matrix as \cite{Havel03,Choi75}
\begin{equation*}
  Choi(\hat{\mathcal{T}}) = \mathcal{I} \otimes \mathcal{T} ( \ketbra{\alpha}{\alpha} ),
\end{equation*}
with $\mathcal{I}$ being the identity superoperator on $\mathscr{H}$ and $\ket{\alpha} \in \mathscr{H} \otimes \mathscr{H}$ is the (non-normalized) maximally entangled state defined by $\ket{\alpha} = \sum_{i=1}^d \ket{i} \otimes \ket{i}$.

The Choi matrix provides a useful way to judge complete positivity of superoperators.
To see this, suppose the operation of $\mathcal{T}$ is given by
\begin{equation}
  \mathcal{T}(A) = \sum_{p,q = 1}^{d^2} \tau_{p,q} T_p A T_q^\dagger,
  \label{eq:vect_Top}
\end{equation}
with an orthonormal operator basis $\{ T_p \}$ satisfying $(T_p| T_q) = \delta_{p,q} \ (p,q = 1,2,\dots,d^2)$ and a $d^2$-dimensional Hermitian matrix $\tau$.
One then finds
\begin{equation}
  \begin{gathered}
    Choi(\hat{\mathcal{T}}) = \sum_{i,j=1}^d \ketbra{i}{j} \otimes \mathcal{T} (\ketbra{i}{j}) \\
    = \sum_{i,j,k,l=1}^d \sum_{p,q = 1}^{d^2} \tau_{p,q} (T_p)_{ki} (T_q)_{lj}^* \ketbra{i}{j} \otimes \ketbra{k}{l} \\
    = \sum_{p,q = 1}^{d^2} \tau_{p,q} \dketbra{T_p}{T_q}.
  \end{gathered}
  \label{eq:vect_ChoiT}
\end{equation}
From Eq.$\,$(\ref{eq:vect_Top}) and the linear independence of $\{ T_p \}$, $\mathcal{T}$ is a completely positive map if and only if $\tau \geq 0$.
On the other hand, $\tau \geq 0$ is equivalent to $Choi(\hat{\mathcal{T}}) \geq 0$.
Combining these results, $\mathcal{T}$ is a completely positive map if and only if $Choi(\hat{\mathcal{T}}) \geq 0$.
Similarly, $\mathcal{T}$ preserves the Hermitian property if and only if $Choi(\hat{\mathcal{T}})^\dagger = Choi(\hat{\mathcal{T}})$.

While the complete positivity condition is stronger than the positivity one in general, we can prove the following;
\begin{lemma}
  Suppose that $\mathcal{T}$ takes the form of
  \begin{equation}
    \mathcal{T}(A) = \sum_{m,n = 1}^{d} \mathfrak{p}_{m,n} E_m A E_n,
    \label{eq:vect_Tdiag}
  \end{equation}
  with a $d$-dimensional Hermitian matrix $\mathfrak{p}$ and $\{ E_m \}$ the projectors introduced in Sec.$\,$\ref{sec:Diag}.
  Then, $\mathcal{T}$ is positive if and only if it is completely positive.
  \label{theorem:diagonal}
\end{lemma}
\begin{proof}
From the linear independence of $\{ E_m \}$, $\mathfrak{p} \geq 0$ is equivalent to $\mathcal{T}$ being completely positive.
We prove that  $\mathfrak{p} \geq 0$ is also an equivalent condition to the positivity of $\mathcal{T}$.
Since the sufficiency is evident, we need to prove only the necessity.
Let $U \in \mathbb{C}^{d \times d}$ be a unitary matrix diagonalizing $\mathfrak{p}$ as $\mathfrak{p}_{m,n} = \sum_{l = 1}^d \mathfrak{p}_l U_{m,l} U_{n,l}^*$ with $\{ \mathfrak{p}_l \}$ the eigenvalues $\mathfrak{p}$.
With those, the above equation reads
\begin{equation*}
  \mathcal{T}(A) = \sum_{l = 1}^{d} \mathfrak{p}_{l} P_l A P_l^\dagger,
\end{equation*}
with $P_l = \sum_{m=1}^d U_{m,l} E_m$.
Suppose that $\mathfrak{p}$ is not positive semidefinite.
Then, at least one of the elements $\{ \mathfrak{p}_l \}$ is negative.
When $\mathfrak{p}_n < 0$, for instance, we introduce $\ket{\psi_n} = [U_{1,n}^* \ U_{2,n}^* \ \dots \ U_{d,n}^* ]^\top$ and $\ket{\varphi} = [1 \ 1 \ \dots 1]^\top $.
We then find
\begin{equation*}
  \braket{\varphi|P_l|\psi_n} = (U^\dagger U)_{n,l} = \delta_{n,l},
\end{equation*}
from the definition of $\{ E_m \}$ and the unitarity of $U$.
This yields
\begin{equation*}
  \braket{\varphi| \mathcal{T}(\ketbra{\psi_n}{\psi_n}) |\varphi} = \sum_{l = 1}^d \mathfrak{p}_l |\braket{\varphi|P_l|\psi_n}|^2 = \mathfrak{p}_n < 0,
\end{equation*}
which implies that $\mathcal{T}$ is not positive.
By contraposition, therefore, $\mathfrak{p} \geq 0$ is necessary for $\mathcal{T}$ being positive.
\end{proof}

\subsection{Lindbladian}
\label{app:vect_Lindbladian}

To judge if a given superoperator is a Lindbladian or not, the following lemma is convenient;
\begin{lemma}
  Suppose that $\mathcal{L}$ preserves the Hermitian property, $[\mathcal{L}(A)]^\dagger = \mathcal{L}(A^\dagger)$, and have trace zero, ${\rm tr} \circ \mathcal{L} = 0$.
  Then, $\mathcal{L}$ is a Lindbladian if and only if the supermatrix $\hat{\mathcal{P}}^I Choi(\hat{\mathcal{L}}) \hat{\mathcal{P}}^I$, with $\hat{\mathcal{P}}^I = I_d \otimes I_d - \dketbra{I_d}{I_d}/d$, is positive semidefinite.
  \label{theorem:Lindbladform}
\end{lemma}
\begin{proof}
  The proof presented here is taken from \cite{Havel03}.
  To start, note that $\hat{\mathcal{P}}^I$ is an Hermitian supermatrix, $[\hat{\mathcal{P}}^I]^2 = \hat{\mathcal{P}}^I$, and its operation reads $\mathcal{P}^I (A) = A - ({\rm tr}(A)/d) I_d$.
  Thus, $\mathcal{P}^I$ is the orthogonal projector onto traceless operators, and the kernel of $\hat{\mathcal{P}}^I$ is ${\rm ker}(\hat{\mathcal{P}}^I) = \{ \dket{I_d} \}$.
  We use these properties.

  If $\mathcal{L}$ is a Lindbladian, it is generally given by
  \begin{equation*}
    \mathcal{L} (A) = -i [H, A] + \sum_{\mu} \mathcal{D} [L_\mu] A, 
  \end{equation*}
  with arbitrary matrices $L_\mu \in C^{d \times d}$ and a Hermitian matrix $H \in C^{d \times d}$.
  From the definition of the Choi matrix Eq.$\,$(\ref{eq:vect_ChoiT}), we find
  \begin{equation*}
    \begin{gathered}
      Choi(\hat{\mathcal{L}}) = - i [ \dketbra{H}{I_d} - \dketbra{I_d}{H}] \\
      + \sum_{\mu} \Big[ \dketbra{L_\mu}{L_\mu} - \frac{1}{2} [ \dketbra{L_\mu^\dagger L_\mu}{I_d} + \dketbra{I_d}{L_\mu^\dagger L_\mu}] \Big].
    \end{gathered}
  \end{equation*}
  By applying $\hat{\mathcal{P}}^I$ from the left and right sides, we can extract the dissipation part
  \begin{equation*}
    \hat{\mathcal{P}}^I Choi(\hat{\mathcal{L}}) \hat{\mathcal{P}}^I = \sum_{\mu} \dketbra{\mathcal{P}^I (L_\mu)}{\mathcal{P}^I (L_\mu)}.
  \end{equation*}
  Thus, $\hat{\mathcal{P}}^I Choi(\hat{\mathcal{L}}) \hat{\mathcal{P}}^I \geq 0$.

  If $\hat{\mathcal{P}}^I Choi(\hat{\mathcal{L}}) \hat{\mathcal{P}}^I$ is positive semidefinite, on the other hand, there exist traceless operators $\{ \overline{L}_\mu \}_{\mu =  1,\dots,D} \ (D < d^2)$ such that
  \begin{equation*}
    \hat{\mathcal{P}}^I Choi(\hat{\mathcal{L}}) \hat{\mathcal{P}}^I = \sum_{\mu = 1}^{D} \dketbra{\overline{L}_\mu}{\overline{L}_\mu}.
  \end{equation*}
  Since ${\rm ker}(\hat{\mathcal{P}}^I) = \{ \dket{I_d} \}$, $Choi(\hat{\mathcal{L}})$ generally reads
  \begin{equation*}
    Choi(\hat{\mathcal{L}}) = \dketbra{K}{I_d} + \dketbra{I_d}{K} + \sum_{\mu = 1}^{D} \dketbra{\overline{L}_\mu}{\overline{L}_\mu},
  \end{equation*}
  where $K$ is a $d$-dimensional complex matrix and we have used $Choi(\hat{\mathcal{L}})^\dagger = Choi(\hat{\mathcal{L}})$ which is derived from the Hermitian preservation.
  Note that we have not imposed ${\rm tr}(\mathcal{L} \bullet) = 0$ yet.
  The operation of $\mathcal{L}$ reads
  \begin{equation*}
    \mathcal{L}(A) =  K A + A K^\dagger + \sum_{\mu = 1}^{D} \overline{L}_\mu A \overline{L}_{\mu}^{\dagger}.
  \end{equation*}
  From the condition ${\rm tr} \mathcal{L} = 0$ or, equivalently,
  \begin{equation*}
    {\rm tr}(\mathcal{L}(A)) =  {\rm tr} \Big( \Big[ K + K^\dagger + \sum_{\mu = 1}^{D} \overline{L}_\mu^\dagger \overline{L} \Big] A  \Big) = 0,
  \end{equation*}
  for any $A$, we obtain $K + K^\dagger = - \sum_{\mu = 1}^{D} \overline{L}_\mu^\dagger \overline{L}$. Thus,
  \begin{equation*}
    \mathcal{L}(A) = -i \Big[ \frac{i(K-K^\dagger)}{2}, A \Big] + \sum_{\mu = 1}^D \mathcal{D} [\overline{L}_\mu] (A),
  \end{equation*}
  and $\mathcal{L}$ is a Lindbladian.
\end{proof}

As an example, suppose that $\mathcal{L}$ is given by
\begin{equation}
  \mathcal{L}(A) = -i [H, A] + \sum_{\alpha,\beta = 1}^{D} \gamma_{\alpha,\beta} \Big[ V_\alpha A V_\beta^\dagger - \frac{1}{2} \{ V_\beta^\dagger V_\alpha, A \}  \Big],
  \label{eq:vect_nondiag}
\end{equation}
with arbitrary matrices $V_\alpha \in \mathbb{C}^{d \times d}$ and Hermitian matrices $H \in \mathbb{C}^{d \times d}$ and $\gamma \in \mathbb{C}^{D \times D}$. In view of Lemma \ref{theorem:Lindbladform}, $\gamma \geq 0$ is a sufficient condition for $\mathcal{L}$ to be a Lindbladian, but not necessary in general.
Note that
\begin{equation*}
  \hat{\mathcal{P}}^I Choi(\hat{\mathcal{L}}) \hat{\mathcal{P}}^I = \sum_{\alpha,\beta = 1}^{D} \gamma_{\alpha,\beta} \dketbra{\mathcal{P}^I (V_\alpha)}{\mathcal{P}^I (V_\beta)}.
\end{equation*}
When $\{ \mathcal{P}^I (V_\alpha) \}$ are linearly dependent, the right-hand side can be positive semidefinite even if $\gamma$ is not.
A trivial example is $D = 2$, $\mathcal{P}^I (V_1) = \mathcal{P}^I (V_2)$, $\gamma_{1,2} = 0$, $\gamma_{1,1} > 0$, and $\gamma_{2,2} = - \gamma_{1,1}/2 < 0$.
In this case, even though $\gamma$ is not positive semidefinite, $\hat{\mathcal{P}}^I Choi(\hat{\mathcal{L}}) \hat{\mathcal{P}}^I = (\gamma_{1,1}/2) \dketbra{\mathcal{P}^I (V_1)}{\mathcal{P}^I (V_1)} \geq 0$.
In fact, the operation of $\mathcal{L}$ is given by the following Lindblad form,
\begin{equation*}
  \begin{gathered}
    \mathcal{L}(A) = -i \Big[ H - \frac{\gamma_{1,1}}{2id}({\rm tr}(V_1)^* V_1 - {\rm tr}(V_1) V_1^\dagger) \\
    + \frac{\gamma_{1,1}}{4id}({\rm tr}(V_2)^* V_2 - {\rm tr}(V_2) V_2^\dagger), A \Big]
    + \frac{\gamma_{1,1}}{2} \mathcal{D} [ \mathcal{P}^I (V_1) ] (A).
  \end{gathered}
\end{equation*}

When $\{ \mathcal{P}^I (V_\alpha) \}$ are linearly independent, $\mathcal{L}$ is a Lindbladian if and only if $\gamma \geq 0$.
Note that $\{ \mathcal{P}^I (V_\alpha) \}$ might not be linearly independent even if $\{ V_\alpha \}$ are.
For instance, let us consider the projectors in Sec.$\,$\ref{sec:Diag}, $\{ E_m \}$.
Although it is a linearly independent set, from $\sum_{m = 1}^d E_m = I_d$ and $\mathcal{P}^I (E_m) = E_m - I_d/d$, we obtain $\sum_{m = 1}^d \mathcal{P}^I (E_m) = 0$.
Thus, $\{ \mathcal{P}^I (E_m) \}$ are not independent.
If $\{ V_\alpha \}$ are traceless, then $\mathcal{P}^I(V_\alpha) = V_\alpha$. This leads to the following corollary;
\begin{corollary}
  In Eq.$\,$(\ref{eq:vect_nondiag}), suppose that $\{ V_\alpha \}$ are traceless and linearly independent. Then, $\mathcal{L}$ is a Lindbladian if and only if $\gamma \geq 0$.
  \label{theorem:Lindbladform_coll}
\end{corollary}

As a concrete example, we consider $\mathcal{L}_s^{G = 0}$ in Sec.$\,$\ref{sec:Diag}, see Eq.$\,$(\ref{eq:Diag_LsPT}).
Note first that 
\begin{equation}
  \begin{gathered}
    \mathcal{L}_s^{G=0} (\rho_B) = \sum_{m,n = 1}^d \lambda_{m,n} \Big[ E_m \rho_B E_n \\ - \frac{1}{2} (E_n E_m \rho_B + \rho_B E_n E_m) \Big],
  \end{gathered}
  \label{eq:vect_LsPT2}
\end{equation}
where we have used $\lambda_{m,m} = 0 \ (m = 1,\dots,d)$ (see Eq.$\,$(\ref{eq:ae.quditDiag_lmd.prop})).
As discussed above, $\{ \mathcal{P}^I (E_m) \}$ are not independent.
Hence, $\lambda \geq 0$ is only a sufficient condition for $\mathcal{L}_s^{G=0}$ to be a Lindbladian, but not necessary.

To judge whether $\mathcal{L}_s^{G=0}$ is a Lindbladian or not, we represent it with linearly independent operators that are traceless.
A possible set is $\{ \tau_k \}_{k = 1,\dots,d-1}$ defined by
\begin{equation*}
  \tau_k = \sum_{l = 1}^k E_l - k E_{k+1},
\end{equation*}
which are, with normalization, commonly used as basis matrices of the $\mathfrak{su}(d)$ Lie algebras together with nondiagonal ones.
With the $d$-dimensional identity matrix $I_d$, the transformation between $\{ E_m \}$ and $I_d \cup \{ \tau_k \}$ bases read
\begin{equation*}
  \begin{pmatrix}
  E_1 \\
  E_2 \\
  \vdots \\
  E_{d-1} \\
  E_d
  \end{pmatrix}
  =
  \begin{pmatrix}
    \ \ \ \ \ \ \ \ \  \scalebox{1.5}{$S$} \ \ \ \ \ \ \ \ \  \begin{matrix} 1 \\ 1 \\ \vdots \\ 1 \\ 1 \end{matrix}
  \end{pmatrix}
  \begin{pmatrix}
    \tau_1 \\
    \tau_2 \\
    \vdots \\
    \tau_{d-1} \\
    I_d/d
  \end{pmatrix},
\end{equation*}
where $S \in \mathbb{R}^{d \times (d-1)}$ is defined by
\begin{equation}
  S =
  \begin{pmatrix}
    q(2) & q(3) & \dots & q(d-1) & q(d) \\
    -1/2 & q(3) & \dots & q(d-1) & q(d) \\
    0 & -1/3 & \dots & q(d-1) & q(d) \\
    & & \vdots \\
    0 & 0 & \dots & q(d-1) & q(d) \\
    0 & 0 & \dots & -1/(d-1) & q(d) \\
    0 & 0 & \dots & 0 & -1/d \\
  \end{pmatrix},
  \label{eq:vect_MatrixS}
\end{equation}
with $q(k) = (k(k-1))^{-1}$.
Using $\{ \tau_k \}$, we can express $\mathcal{L}_s^{G=0}$ as
\begin{equation}
  \begin{gathered}
    \mathcal{L}_s^{G=0} (\rho_B) = \\
    - \Big[ \sum_{\substack{m,n=1 \\ (m<n)}}^d \sum_{k = 1}^{d-1} \frac{\lambda_{m,n} - \lambda_{m,n}^*}{2d} (S_{m,k}-S_{n,k})  \tau_k, \rho_B \Big] \\
    + \sum_{k,l = 1}^{d-1} (S^\top \lambda S)_{k,l} \Big[ \tau_k \rho_B \tau_l - \frac{1}{2} ( \tau_l \tau_k \rho_B + \rho_B \tau_l \tau_k) \Big].
  \end{gathered}
  \label{eq:vect_LsPT_general}
\end{equation}
From Corollary \ref{theorem:Lindbladform_coll}, thus, $S^\top \lambda S \geq 0$ is necessary and sufficient conditions for $\mathcal{L}_s^{G = 0}$ to be a Lindbladian.

%%%%%%%%%%%%%%%%%%%%%%%%%%%%%%%%%%%%%%%%%%%%%%%%%%%%%%%%%%%%%%%%%%%%%%%%%%%%%%%%%%%%%%%%
\section{Adiabatic elimination for a damped oscillator system}
\label{app:damposc}

Let $\mathscr{H}_A$ be a damped oscillator system with Eq.$\,$(\ref{eq:JC_LA}) being coupled to a slow sub-system $\mathscr{H}_B$ via
\begin{equation*}
  \epsilon \mathcal{L}_{\rm int} \bullet = - i [g ( a^\dagger \otimes B + a \otimes B^\dagger ), \bullet],
\end{equation*}
with $B$ an operator on $\mathscr{H}_B$. We do not specify the form of $\mathcal{L}_B$.
In this appendix, we present an efficient way to perform adiabatic elimination in this setting.

We first solve the eigenvalue problem of $\mathcal{L}_A$.
The eigenoperators of $\mathcal{L}_A$ can be represented in a compact form with the normal ordering and the generalized Laguerre polynomials \cite{Briegel93}.
Here we present an alternative way to diagonalize $\mathcal{L}_A$, which can be derived from the ladder superoperator technique \cite{Prosen10,BZ21}.
To this end, we introduce a Hermitian preserving map
\begin{equation*}
  \mathcal{W}_A (O_A) = \sum_{p,q = 0}^\infty \frac{(-n_{\rm th})^p}{p! \, q!} (a^\dagger)^p a^q \, O_A \, (a^\dagger)^q a^p,
\end{equation*}
with $O_A$ an arbitrary operator on $\mathscr{H}_A$.
In the supermatrix form (see Eq.$\,$(\ref{eq:vect_supermatrix})), it reads $\hat{\mathcal{W}}_A = \exp(-n_{\rm th} a^\top \! \otimes a^\dagger) \exp( a^* \! \otimes a)$,
from which the inverse reads $\hat{\mathcal{W}}_A^{-1} = \exp( - a^* \! \otimes a) \exp( n_{\rm th} a^\top \! \otimes a^\dagger)$. The similarity transformation with $\mathcal{W}_A$ gives
\begin{align}
  \begin{split}
    \mathcal{W}_A (a \mathcal{W}_A^{-1} (O_A) ) &= a \, O_A + n_{\rm th} O_A a, \\
    \mathcal{W}_A (\mathcal{W}_A^{-1} (O_A) a^\dagger) &= O_A \, a^\dagger + n_{\rm th} a^\dagger O_A, \\
    \mathcal{W}_A (a^\dagger \mathcal{W}_A^{-1} (O_A) ) &= (1 + n_{\rm th}) a^\dagger O_A + O_A \, a^\dagger, \\
    \mathcal{W}_A (\mathcal{W}_A^{-1} (O_A) a) &= (1 + n_{\rm th}) \, O_A a + a \, O_A. \\
  \end{split}
  \label{eq:damposc_WSimTrans}
\end{align}
Using these relations, we find
\begin{equation}
  \begin{gathered}
    \mathcal{M}_A (O_A) \equiv \mathcal{W}_A \circ \mathcal{L}_A \circ \mathcal{W}_A^{-1} (O_A) \\
    = - \frac{1}{2} ( \bar{\gamma} \, a^\dagger \! a \, O_A + \bar{\gamma}^* O_A \, a^\dagger \! a ),
  \end{gathered}
  \label{eq:damposc_MA}
\end{equation}
with $\bar{\gamma} = \gamma + 2i\Delta_A$.
This is diagonal in the Fock basis of the oscillator which we denote by $\{ \ket{n} \}_{n = 0,1,2,\dots}$.
Thus, the eigenvalue problem of $\mathcal{L}_A$ is solved in the vectorization representation (see Eq.$\,$(\ref{eq:vect_eig})) as
\begin{equation*}
  \begin{gathered}
    \hat{\mathcal{L}}_A \dket{X_{m,n}} = \lambda_{m,n} \dket{X_{m,n}}, \\
    \dbra{\bar{X}_{m,n}}  \hat{\mathcal{L}}_A = \lambda_{m,n} \dbra{\bar{X}_{m,n}},
  \end{gathered}
\end{equation*}
for $m,n = 0,1,2,\dots$, with $\lambda _{m,n} = - (\bar{\gamma} m + \bar{\gamma}^* n)/2$, $\dket{X_{m,n}} = \hat{\mathcal{W}}_A^{-1} \dket{ \ketbra{m}{n} }$, and $\dbra{\bar{X}_{m,n}} = \dbra{\ketbra{m}{n}} \hat{\mathcal{W}}_A$.
The eigenvalues satisfy $\lambda_{0,0} = 0$ and ${\rm Re}(\lambda_{m,n}) < 0 \ (m + n > 0)$.
Therefore, the evolution only with $\mathcal{L}_A$ exponentially converges to a unique steady state $\bar{\rho}_A$ given by
\begin{equation}
  \bar{\rho}_A = X_{0,0} =  \Big( \frac{n_{\rm th}}{1+n_{\rm th}} \Big)^{a^\dagger \! a} / {\rm tr}_A \Big[ \Big( \frac{n_{\rm th}}{1+n_{\rm th}} \Big)^{a^\dagger \! a} \Big].
  \label{eq:damposc_steadystate}
\end{equation}
Note that $n_{\rm th}$ is the average oscillator quantum number with $\bar{\rho}_A$
\begin{equation}
  {\rm tr}_A (a^\dagger \! a \, \bar{\rho}_A) = n_{\rm th}.
  \label{eq:damposc_nth}
\end{equation}

Adiabatic elimination is greatly simplified in the eigenbasis of $\mathcal{L}_A$.
To implement this, we introduce $\xi = \mathcal{W}_A \otimes \mathcal{I}_B (\rho)$, which is Hermitian $\xi^\dagger = \xi$, and consider the master equation for $\xi$,
\begin{equation}
  \frac{d}{dt} \xi = \mathcal{M}_A \otimes \mathcal{I}_B (\xi) + \epsilon \mathcal{M}_{\rm ad} (\xi) \equiv \mathcal{M}_{\rm tot} (\xi),
  \label{eq:damposc_master}
\end{equation}
where $\mathcal{M}_A$ is defined in Eq.$\,$(\ref{eq:damposc_MA}) and
$\mathcal{M}_{\rm ad} = (\mathcal{W}_A \otimes \mathcal{I}_B) \circ (\mathcal{I}_A \otimes \mathcal{L}_B + \mathcal{L}_{\rm int} ) \circ (\mathcal{W}_A^{-1} \otimes \mathcal{I}_B)$.
From Eq.$\,$(\ref{eq:damposc_WSimTrans}), the operation of $\mathcal{M}_{\rm ad}$ reads
\begin{equation*}
  \begin{gathered}
    \mathcal{M}_{\rm ad} (\xi) = \mathcal{I}_A \otimes \mathcal{L}_B (\xi) \\
    + (a \otimes I_B) \mathcal{I}_A \otimes \mathcal{B}_L (\xi) + \mathcal{I}_A \otimes \mathcal{B}_D (\xi) (a^\dagger \otimes I_B) \\
    + (a^\dagger \otimes I_B) \mathcal{I}_A \otimes \mathcal{B}_R (\xi) + \mathcal{I}_A \otimes \mathcal{B}_U (\xi) (a \otimes I_B),
  \end{gathered}
\end{equation*}
where $I_B$ the identity operator on $\mathscr{H}_B$ and $\{ \mathcal{B}_{X} \}_{X=L,D,R,U}$ are suoperoperators on $\mathscr{H}_B$ defined by
\begin{align*}
  \begin{split}
    \mathcal{B}_L (O_B) &= i [O_B, B^\dagger], \ \ \mathcal{B}_D (O_B) = i [O_B, B], \\
    \mathcal{B}_R (O_B) &= i (n_{\rm th} O_B B - (1+n_{\rm th}) B O_B), \\
    \mathcal{B}_U (O_B) &= -i (n_{\rm th} B^\dagger O_B - (1+n_{\rm th}) O_B B^\dagger), \\
  \end{split}
\end{align*}
with $O_B$ an arbitrary operator on $\mathscr{H}_B$.
We can easily show $\mathcal{B}_L (O_B) ^\dagger = \mathcal{B}_D (O_B^\dagger)$ and $\mathcal{B}_R (O_B)^\dagger = \mathcal{B}_U (O_B^\dagger)$ as expected from the Hermitian preserving property of $\mathcal{M}_{\rm ad}$.

Instead of Eq.$\,$(\ref{eq:AE_master}), we perform adiabatic elimination for Eq.$\,$(\ref{eq:damposc_master}).
We introduce $\xi = \mathcal{J} (\rho_s)$, which is related to $\mathcal{K}$ in Sec.$\,$\ref{sec:AE} as $\mathcal{K} = \mathcal{W}_A^{-1} \otimes \mathcal{I}_B \circ \mathcal{J}$.
The invariance condition Eq.$\,$(\ref{eq:AE_inv1}) now reads $\mathcal{J} (\mathcal{L}_s (\rho_s)) = \mathcal{M}_{\rm tot} (\mathcal{J}(\rho_s))$.
Since this equation cannot be solved exactly to our knowledge, as in Eq.$\,$(\ref{eq:AE_asymexp}), we expand $\mathcal{J}$ with respect to $\epsilon$, $\mathcal{J} = \sum_{n = 0}^\infty \epsilon^n \mathcal{J}_n$.
Then, similarly to the derivation of Eq.$\,$(\ref{eq:AE_asymexp0}), we find the zeroth order $\mathcal{J}_0 (\rho_s) = \ketbra{0}{0} \otimes \rho_s$.
Substituting the asymptotic expansions into the invariance condition, the $n$-th order of $\epsilon$ reads
\begin{equation}
  \ketbra{0}{0} \otimes \mathcal{L}_{s,n} (\rho_s) = \mathcal{M}_A \otimes \mathcal{I}_B (\mathcal{J}_{n} (\rho_s) ) + \mathcal{M}_n(\rho_s),
  \label{eq:damposc_inv.cond.n}
\end{equation}
with
\begin{equation*}
  \mathcal{M}_n(\rho_s) = \mathcal{M}_{\rm ad} (\mathcal{J}_{n-1} (\rho_s)) - \sum_{k=1}^{n} \mathcal{J}_{k} (\mathcal{L}_{s,n-k}(\rho_s)).
\end{equation*}
Note $\braket{0| \mathcal{M}_A |0} = 0$, which can be shown directly from the definition Eq.$\,$(\ref{eq:damposc_MA}).
Thus, sandwiching both sides of Eq.$\,$(\ref{eq:damposc_inv.cond.n}) with $\ket{0}$, we find $\mathcal{L}_{s,n}$
\begin{equation*}
  \mathcal{L}_{s,n} (\rho_s) = \braket{0|\mathcal{M}_n(\rho_s)|0}.
\end{equation*}
Furthermore, we obtain
\begin{equation*}
  \begin{gathered}
    \mathcal{J}_{n} (\rho_s) = \mathcal{M}_A^{+} \otimes \mathcal{I}_B (\ketbra{0}{0} \otimes \mathcal{L}_{s,n} (\rho_s) - \mathcal{M}_n(\rho_s)) \\
    + \ketbra{0}{0} \otimes G_n(\rho_s),
  \end{gathered}
\end{equation*}
as in the derivation of Eq.$\,$(\ref{eq:AE_Kn}).
In this equation, $G_n = \braket{0|\mathcal{J}_{n}|0}$ is the gauge superoperator associated with the singularity of $\mathcal{M}_A$.
This definition of $G_n$ is consistent with the one in Sec.$\,$\ref{sec:AE} because of $\braket{0|\mathcal{J}_{n}|0} = {\rm tr}_A \circ \mathcal{K}_n$, which can be shown from the identity $\braket{0|\mathcal{W}_A (O_A)|0} = {\rm tr}_A (O_A)$.
In what follows, we consider the partial trace parametrization or, equivalently, we set $G_n = 0 \ (n = 1,2,\dots)$.

Before proceeding with the calculation, we list several convenient formulas.
We introduce $O_{m,n} = \ketbra{m}{n} \otimes O_B + (h.c.)$.
To calculate $\mathcal{L}_{s,n}$, we use
\begin{equation*}
  \braket{0|\mathcal{M}_{\rm ad}(\ketbra{0}{0} \otimes O_B )|0} = \mathcal{L}_B (O_B),
\end{equation*}
and
\begin{equation*}
  \begin{gathered}
    \braket{0|\mathcal{M}_{\rm ad}( O_{m,n} )|0} = \delta_{m,0} \delta_{n,0} \mathcal{L}_B (O_B) \\
    + \delta_{m,1} \delta_{n,0} \mathcal{B}_L (O_B) +  \delta_{m,0} \delta_{n,1} \mathcal{B}_D (O_B) + (h.c.).
  \end{gathered}
\end{equation*}
To evaluate $\mathcal{J}_n$, we need to calculate $\mathcal{M}_A^{+}(\bullet)$ and $\ketbra{0}{0} \otimes \braket{0|\mathcal{M}_{\rm ad} (\bullet)|0}  - \mathcal{M}_{\rm ad} (\bullet)$.
From Eq.$\,$(\ref{eq:vect_MPinverse}), the operation of $\mathcal{M}_A^{+}$ reads
\begin{equation*}
  \begin{gathered}
    \mathcal{M}_A^{+} (\ketbra{m}{n} + (h.c.)) = - y_{m,n} \ketbra{m}{n} + (h.c.), \\
    (m + n > 0),
  \end{gathered}
\end{equation*}
with $y_{m,n} = - 1/\lambda_{m,n} = 2/(\bar{\gamma} m + \bar{\gamma}^* n)$.
To calculate the latter, $\ketbra{0}{0} \otimes \braket{0|\mathcal{M}_{\rm ad} (\bullet)|0}  - \mathcal{M}_{\rm ad} (\bullet)$, we use
\begin{equation*}
  \begin{gathered}
    \ketbra{0}{0} \otimes \braket{0|\mathcal{M}_{\rm ad} (\ketbra{0}{0} \otimes \rho_B)|0}  - \mathcal{M}_{\rm ad} (\ketbra{0}{0} \otimes \rho_B) \\
    = - \ketbra{1}{0} \otimes \mathcal{B}_R (\rho_B) + (h.c.),
  \end{gathered}
\end{equation*}
for $\rho_B$ any Hermitian operator on $\mathscr{H}_B$ and
\begin{equation*}
  \begin{gathered}
    \ketbra{0}{0} \otimes \braket{0|\mathcal{M}_{\rm ad} (O_{m,n})|0} - \mathcal{M}_{\rm ad} (O_{m,n}) \\
    = - \sqrt{m+1} \ketbra{m+1}{n} \otimes \mathcal{B}_R(O_B) \\
    - \sqrt{n + 1} \ketbra{m}{n+1} \otimes \mathcal{B}_U(O_B) \\
    - (1 - \delta_{m,1}\delta_{n,0}) \sqrt{m} \ketbra{m-1}{n} \otimes \mathcal{B}_L(O_B) \\
    - (1 - \delta_{m,0}\delta_{n,1}) \sqrt{n} \ketbra{m}{n-1} \otimes \mathcal{B}_D(O_B) \\
    - \ketbra{m}{n} \otimes \mathcal{L}_B (O_B) + (h.c.).
  \end{gathered}
\end{equation*}

With the aid of these formulas, we can calculate the higher-order contributions.
For the partial trace parametrization, $\{ \mathcal{L}_{s,n} \}_{n=1,2,3,4}$ and $\{ \mathcal{J}_{n} \}_{n = 1,2,3}$ read
\begin{equation*}
  \mathcal{L}_{s,1}(\rho_B) = \mathcal{L}_B (\rho_B),
\end{equation*}
\begin{equation*}
  \mathcal{J}_{1}(\rho_B) = y_{1,0} \ketbra{1}{0} \otimes \mathcal{B}_R (\rho_B) + (h.c.),
\end{equation*}
\begin{equation*}
  \mathcal{L}_{s,2}(\rho_B) = y_{1,0} \mathcal{B}_R \circ \mathcal{B}_L (\rho_B) + (h.c.),
\end{equation*}
\begin{equation*}
  \begin{gathered}
    \mathcal{J}_{2}(\rho_B) = y_{1,0}^2 \ketbra{1}{0} \otimes [\mathcal{L}_B, \mathcal{B}_R] (\rho_B) \\
    + \sqrt{2} \, y_{1,0} y_{2,0} \ketbra{2}{0} \otimes \mathcal{B}_R \circ \mathcal{B}_R (\rho_B) \\
    + y_{1,0} y_{1,1} \ketbra{1}{1} \otimes \mathcal{B}_U \circ \mathcal{B}_R (\rho_B) + (h.c.),
  \end{gathered}
\end{equation*}
\begin{equation*}
  \mathcal{L}_{s,3}(\rho_B) = y_{1,0}^2 \mathcal{B}_L \circ [\mathcal{L}_B, \mathcal{B}_R] (\rho_B) + (h.c.),
\end{equation*}
\begin{equation*}
  \begin{gathered}
    \mathcal{J}_{3}(\rho_B) = y_{1,0}^3 \ketbra{1}{0} \otimes [\mathcal{L}_B, [\mathcal{L}_B, \mathcal{B}_R]] (\rho_B) \\
    + 2 \, y_{1,0}^2 y_{2,0} \ketbra{1}{0} \otimes \mathcal{B}_L \circ \mathcal{B}_R \circ \mathcal{B}_R (\rho_B) \\
    + |y_{1,0}|^2 y_{1,1} \ketbra{0}{1} \otimes \mathcal{B}_L \circ \mathcal{B}_U \circ \mathcal{B}_R (\rho_B) \\
    + y_{1,0}^2 y_{1,1} \ketbra{1}{0} \otimes \mathcal{B}_D \circ \mathcal{B}_U \circ \mathcal{B}_R (\rho_B) \\
    - y_{1,0}^2 \ketbra{1}{0} \otimes \mathcal{B}_R \circ \mathcal{L}_{s,2} (\rho_B) \\
    + \dots + (h.c.),
  \end{gathered}
\end{equation*}
\begin{equation*}
  \begin{gathered}
    \mathcal{L}_{s,4}(\rho_B) = y_{1,0}^3 \mathcal{B}_L \circ [\mathcal{L}_B, [\mathcal{L}_B, \mathcal{B}_R]] (\rho_B) \\
    + 2 \, y_{1,0}^2 y_{2,0} \mathcal{B}_L \circ \mathcal{B}_L \circ \mathcal{B}_R \circ \mathcal{B}_R (\rho_B) \\
    + |y_{1,0}|^2 y_{1,1} \mathcal{B}_D \circ \mathcal{B}_L \circ \mathcal{B}_U \circ \mathcal{B}_R (\rho_B) \\
    + y_{1,0}^2 y_{1,1} \mathcal{B}_L \circ \mathcal{B}_D \circ \mathcal{B}_U \circ \mathcal{B}_R (\rho_B) \\
    - y_{1,0}^2 \mathcal{B}_L \circ \mathcal{B}_R \circ \mathcal{L}_{s,2} (\rho_B) + (h.c.),
  \end{gathered}
\end{equation*}
where $\dots$ in $\mathcal{J}_{3}$ are the terms with $\ketbra{m}{n} \ (m+n > 1)$, which do not contribute to $\mathcal{L}_{s,4}$.
With $B = \sigma_-$ and $\mathcal{L}_B = 0$, we obtain Eq.$\,$(\ref{eq:JC_LsPT}).

%%%%%%%%%%%%%%%%%%%%%%%%%%%%%%%%%%%%%%%%%%%%%%%%%%%%%%%%%%%%%%%%%%%%%%%%%%%%%%%%%%%%%%%%
\section{Properties of Eq.$\,$(\ref{eq:JC_LsPT})}
\label{app:damposc.4thL}

In this appendix, we discuss several properties of $\mathcal{L}_s^{G=0}$ defined by Eq.$\,$(\ref{eq:JC_LsPT}).

\subsection{Spectrum}
\label{app:damposc.4thL_spec}

Let us first calculate the spectrum of $\mathcal{L}_s^{G=0}$.
With $I_B$ the identity matrix on $\mathscr{H}_B$ ($2$-dimensional identity matrix), the set $\{ I_B/\sqrt{2}, \sigma_x/\sqrt{2}, \sigma_y/\sqrt{2}, \sigma_z/\sqrt{2} \}$ is an orthonormal basis with the Hilbert-Schmidt inner product.
Let $[\mathcal{L}_s^{G=0}]$ be the $4 \times 4$ matrix representation of $\mathcal{L}_{s}^{G=0}$ in this basis. It reads
\begin{equation}
  [\mathcal{L}_{s}^{G=0}] =
  \begin{pmatrix}
    0 & 0 & 0 & 0 \\
    0 & -1/T_2 & - \omega_B^{(4)} & 0 \\
    0 & \omega_B^{(4)} & - 1/T_2 & 0 \\
    R_{z}/T_1 & 0 & 0 & - 1/T_1
  \end{pmatrix},
  \label{eq:damposc.4thL_matrix.form}
\end{equation}
with $1/T_1 = \gamma_-^{(4)} + \gamma_+^{(4)}$, $1/T_2 = 1/(2 T_1) + 2 \gamma_\phi^{(4)}$, and $R_{z} = - (\gamma_-^{(4)} - \gamma_+^{(4)}) T_1$.
From this, the spectrum of $\mathcal{L}_{s}^{G=0}$ is given by $\{ 0, -1/T_2 + i \omega_B^{(4)}, -1/T_2 - i \omega_B^{(4)}, -1/T_1 \}$.

\subsection{Positivity}
\label{app:damposc.4thL_positive}

We next show that the time evolution map is positive even when $\gamma_\phi^{(4)} < 0$.
To see this, we introduce the Bloch vector $\bm{r}(t) = (r_x(t), r_y(t), r_z(t))^\top \in \mathbb{R}^3$, which is related to the partial trace as $\rho_B (t) = (I_B + \sum_{i = x,y,z} r_i(t) \sigma_i)/2$.
From Eq.$\,$(\ref{eq:JC_LsPT}), the evolution of the Bloch vector reads
\begin{equation*}
  \begin{gathered}
    \frac{d}{dt} r_x(t) = - \frac{r_x(t)}{T_2} - \omega_B^{(4)} r_y(t), \\
    \frac{d}{dt} r_y(t) = - \frac{r_y(t)}{T_2} + \omega_B^{(4)} r_x(t), \\
    \frac{d}{dt} r_z(t) = - \frac{1}{T_1} (r_z(t) - R_{z}).
  \end{gathered}
\end{equation*}
These equations mean that that $R_{z}$ is the asymptotic value of $r_z(t)$, $\lim_{t \to \infty} \bm{r}(t) = (0, 0, R_{z})^\top$.
They also lead to
\begin{equation*}
  \frac{d}{dt} \bm{r}^2(t) = - 2 \Big( \frac{\bm{r}^2(t) - r_z^2(t)}{T_2} + \frac{r_z(t) (r_z(t) - R_{z})}{T_1} \Big).
\end{equation*}
The partial trace $\rho_B(t)$ is positive semidefinite if and only if $\bm{r}^2(t) \leq 1$.
If we have $(d/dt) \bm{r}^2(t) < 0$ under the constraint $\bm{r}^2(t) = 1$ for any $t$, then $\bm{r}^2(t)$ does not exceed unity along the time evolution and thus positivity is preserved.
Substituting $\bm{r}^2(t) = 1$ into the above equation, we obtain
\begin{equation*}
  \begin{gathered}
    \frac{d}{dt} \bm{r}^2(t) = \\
    - \frac{2}{\Delta T}  \Big[ \Big( r_z(t) - \frac{\Delta T}{2 T_1} R_{z} \Big)^2 + (\Delta T)^2 ( \gamma_-^{(4)} \gamma_+^{(4)} - 4 (\gamma_\phi^{(4)})^2 ) \Big],
  \end{gathered}
\end{equation*}
with $1/\Delta T = 1/T_1 - 1/T_2 = (\gamma_-^{(4)} + \gamma_+^{(4)})/2 - 2 \gamma_\phi^{(4)}$.
When $|g|/\gamma \ll 1$, we have $\gamma_{\pm}^{(4)} \gg |\gamma_\phi^{(4)}|$.
This leads to $\Delta T > 0$ and $\gamma_-^{(4)} \gamma_+^{(4)} > 4 (\gamma_\phi^{(4)})^2$.
Therefore, $(d/dt) \bm{r}^2(t) < 0$ whenever $\bm{r}^2(t) = 1$.
This proves that the time evolution map is positive even with a negative value of $\gamma_\phi^{(4)}$.

\subsection{Incompatibility with a Kraus map}
\label{app:damposc.4thL_WPG}

From the result in Appendix \ref{app:damposc.4thL_spec}, the spectrum of the time evolution map $\exp(\mathcal{L}_s^{G = 0} t)$ is given by $\{ 1, \exp(-t/T_2 + i \omega_B^{(4)} t), \exp( -t/T_2 - i \omega_B^{(4)} t), \exp( -t/T_1 ) \}$.
Here we prove that a Kraus map whose spectrum is given by this set does not exist when $\gamma_\phi^{(4)} < 0$ and for an infinitesimal time $t$.
To this end, we use the theorem in \cite{WPG10}, which was explicitly presented in Sec.$\,$\ref{sec:JC}.
Let $\Lambda$ be the spectrum of the time evolution map $\exp( \mathcal{L}_{s}^{G=0} t )$, that is,
\begin{equation*}
  \Lambda = \{ 1, e^{-t/T_2 + i \omega_B^{(4)}t}, e^{-t/T_2 - i \omega_B^{(4)} t}, e^{-t/T_1} \}.
\end{equation*}
We recall that $T_1$ and $T_2$ are defined underneath Eq.$\,$(\ref{eq:damposc.4thL_matrix.form}).
It is easy to check that $\Lambda$ is closed under complex conjugation.
To see the condition Eq.$\,$(\ref{eq:JC_thm}), we note that $\{ s_i \}$ are now all positive from their definitions. 
In such cases, assuming $s_1 \geq s_2 \geq s_3$, the condition Eq.$\,$(\ref{eq:JC_thm}) reads $s_1 \leq 1$ and $s_1 + s_2 \leq 1 + s_3$ (see \cite{WPG10}).
When $\gamma_\phi^{(4)} < 0$, we always have $1/T_1 > 1/T_2$, and $\exp( -t/T_2) > \exp( -t/T_1 )$ for $t > 0$.
Thus, we set $s_1 = s_2 = \exp( -t/T_2)$ and $s_3 = \exp( -t/T_1)$.
Then, while the first condition $s_1 \leq 1$ is always satisfied, the second condition that reads
\begin{equation}
  2 e^{-t/T_2} \leq 1 + e^{-t/T_1},
  \label{eq:damposc.4thL_CPinequality}
\end{equation}
is nontrivial.
For an infinitesimal time $t$ such that $\exp( -t/T_i ) \simeq 1 - t/T_i \ (i = 1,2)$, this condition reads $\gamma_\phi^{(4)} \geq 0$ and is violated if $\gamma_\phi^{(4)} < 0$.
Therefore, the condition Eq.$\,$(\ref{eq:JC_thm}) is not satisfied when $\gamma_\phi^{(4)} < 0$ and for an infinitesimal time $t$ and this concludes the proof. 

%%%%%%%%%%%%%%%%%%%%%%%%%%%%%%%%%%%%%%%%%%%%%%%%%%%%%%%%%%%%%%%%%%%%%%%%%%%%%%%%%%%%%%%%
\section{Exact master equation with a product initial state and complete positivity}
\label{app:exactCP}

When the initial state is a product state,the time evolution of the partial trace is in general completely positive as discussed following Eq.$\,$(\ref{eq:diss_InitialProduct}).
The complete positivity violation found in this paper stems from the transient dynamics that is discarded in adiabatic elimination.
We expect complete positivity to be restored by accounting for the short-time regime.
In this appendix, we confirm this expectation.
To this end, we consider the qudit system discussed in Sec.$\,$\ref{sec:Diag} and derive the exact master equation for the partial trace with an initially product state.

In order to derive the master equation, we first calculate the time evolution of the partial trace.
Using Eq.$\,$(\ref{eq:Diag_LAmn}),
the evolution of the total density matrix with a product initial state $\rho(0) = \rho_A \otimes \rho_B(0)$, with $\rho_A$ a fixed initial state on $\mathscr{H}_A$, reads
\begin{equation*}
  \begin{gathered}
    \rho(t) = e^{\mathcal{L}_{\rm tot} t} (\rho_A \otimes \rho_B(0)) \\
    = \sum_{m,n = 1}^d e^{\mathcal{L}_A^{(m,n)} t} (\rho_A) \otimes E_m \rho_B(0) E_n.
  \end{gathered}
\end{equation*}
By taking ${\rm tr}_A$ of this equation, we obtain the evolution of the partial trace
\begin{equation}
  \rho_B(t) = \sum_{m,n = 1}^d [\mathcal{T}_{B,t}]_{m,n} E_m \rho_B(0) E_n,
  \label{eq:exactCP_rhoBt}
\end{equation}
with $[\mathcal{T}_{B,t}]_{m,n} = {\rm tr}_A \circ \exp (\mathcal{L}_A^{(m,n)} t) (\rho_A)$.
From $(\mathcal{L}_A^{(m,n)} (A))^\dagger = \mathcal{L}_A^{(n,m)} (A^\dagger)$, which can be directly shown from the definition of $\mathcal{L}_A^{(m,n)}$, the coefficients satisfy $[\mathcal{T}_{B,t}]_{m,n}^* = [\mathcal{T}_{B,t}]_{n,m}$.
This relation ensures the Hermitian property of $\rho_B(t)$.
Furthermore, since $\mathcal{L}_A^{(m,m)}$ is a Lindbladian, the trace preservation leads to $[\mathcal{T}_{B,t}]_{m,m} = {\rm tr}_A (\rho_A) = 1$.
This ensures the trace preservation of the partial trace ${\rm tr}_B (\rho_B(t)) = 1$ with ${\rm tr}_B$ the trace over $\mathscr{H}_B$.

In what follows, we assume $[\mathcal{T}_{B,t}]_{m,n} \ne 0$ for any $m,n = 1,\dots,d$ and $t \in \mathbb{R}_{\geq 0}$.
Under this assumption, the master equation for the partial trace can be put into a compact form as
\begin{equation}
  \frac{d}{dt} \rho_B (t) = \sum_{m,n = 1}^d \lambda_{m,n} (t) E_m \rho_B(t) E_n,
  \label{eq:exactCP_LB}
\end{equation}
with $\lambda_{m,n} (t) = (d/dt) \ln ([\mathcal{T}_{B,t}]_{m,n})$.
Long after decay time of the fast sub-system, the time dependence of $[\mathcal{T}_{B,t}]_{m,n}$ is governed by the eigenmode of $\mathcal{L}_A^{(m,n)}$ whose real part is the closest to zero.
In Appendix \ref{app:ae.qudit_all}, we have defined the mode $k = 1$ to have this property.
Therefore, the time dependence of $[\mathcal{T}_{B,t}]_{m,n}$ asymptotically reads $[\mathcal{T}_{B,t}]_{m,n} \propto \exp(\lambda_{m,n}t)$, from which we obtain
\begin{equation}
  \lim_{t \to \infty} \lambda_{m,n} (t) = \lambda_{m,n}.
  \label{eq:exactCP_lmdrelax}
\end{equation}
In other words, while the master equation Eq.$\,$(\ref{eq:exactCP_LB}) explicitly depends on time before the relaxation is completed, it asymptotically agrees with the evolution equation in adiabatic elimination Eq.$\,$(\ref{eq:Diag_LsPT}).

From the properties of $[\mathcal{T}_{B,t}]_{m,n}$ derived above, we have $\lambda_{m,n}(t)^* = \lambda_{n,m}(t)$ and $\lambda_{m,m}(t) = 0$.
As in Eq.$\,$(\ref{eq:vect_LsPT_general}), thus, Eq.$\,$(\ref{eq:exactCP_LB}) can be rewritten as
\begin{equation}
  \begin{gathered}
    \frac{d}{dt} \rho_B (t) = \\
    - \Big[ \sum_{\substack{m,n=1 \\ (m<n)}}^d \sum_{k = 1}^{d-1} \frac{\lambda_{m,n}(t) - \lambda_{m,n}^*(t)}{2d} (S_{m,k}-S_{n,k})  \tau_k, \rho_B (t) \Big] \\
    + \sum_{k,l = 1}^{d-1} (S^\top \lambda (t) S)_{k,l} \Big[ \tau_k \rho_B (t) \tau_l - \frac{1}{2} ( \tau_l \tau_k \rho_B (t) + \rho_B (t) \tau_l \tau_k) \Big],
  \end{gathered}
  \label{eq:exactCP_LB_Lindblad}
\end{equation}
where the matrix $\lambda(t)$ is defined by $[\lambda(t)]_{m,n} = \lambda_{m,n}(t)$.
Note that the eigenvalues of $S^\top \lambda (t) S$, which correspond to the decay rate, are now dependent on time.

As a concrete example, we consider a qutrit $(d = 3)$ coupled to a dissipative qubit governed by Eq.$\,$(\ref{eq:Diag_qubit}).
We are interested in a parameter set where (complete) positivity is violated in the slow dynamics.
According to Fig.$\,$\ref{fig:Diag_sign}, the violation occurs, for instance, when $(\chi_1,\chi_2,\chi_3) = (0,\chi,2\chi)$, $\chi/\kappa = 0.1$, $\Omega/\kappa = 0.5$, and $\Delta/\kappa = 0.5$, and we use these values in our simulation.
We assume that $\rho_A$ is the unique steady state of $\mathcal{L}_A$, that is, the initial states are on the invariant manifold without the coupling.
In our simulation, we computed $\rho_A$ from the right eigenvector of $\mathcal{L}_A$ the eigenvalue of which is the closest to zero.
With the steady state $\rho_A$, we computed $[\mathcal{T}_{B,t}]_{m,n}$ in the vectorized representation $[\mathcal{T}_{B,t}]_{m,n} = \dbra{I_2} \exp (\hat{\mathcal{L}}_A^{(m,n)} t) \dket{\rho_A}$,
where the matrix exponentiation $\exp (\hat{\mathcal{L}}_A^{(m,n)} t)$ was directly evaluated using a SciPy function.
We confirmed that $[\mathcal{T}_{B,t}]_{m,n} \ne 0$ is satisfied for $m,n = 1,2,3$ within the time range under study ($0 \leq \kappa t \leq 20$).
Lastly, we computed the coefficients $\{ \lambda_{m,n}(t) \}$ in the discretized form $\lambda_{m,n}(t) = (\ln ([\mathcal{T}_{B,t+\delta t}]_{m,n}) - \ln ([\mathcal{T}_{B,t}]_{m,n}))/\delta t$ with $\kappa \delta t = 10^{-2}$.

%%%%%%%%%%%%%%%%%%%%%%%%%%%%%%%%%%%%%%%%%%%%%%%%%%%%%%%%%
\begin{figure}[t]
  \includegraphics[keepaspectratio, scale=0.45]{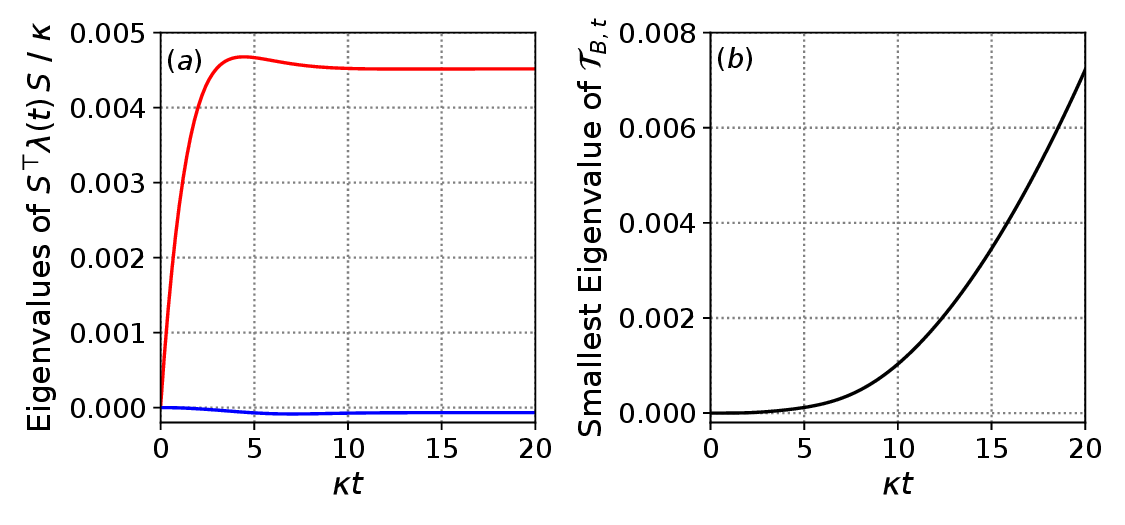}
  \caption{Coefficients in the exact master equation and confirmation of complete positivity for the qubit-qutrit system with $(\chi_1,\chi_2,\chi_3) = (0,\chi,2\chi)$, $\chi/\kappa = 0.1$, $\Omega/\kappa = 0.5$, and $\Delta/\kappa = 0.5$.
  $(a)$ Time-dependent coefficients in front of the dissipators in the exact master equation.
  From Eq.$\,$(\ref{eq:exactCP_LB_Lindblad}), we evaluated them from the eigenvalues of the $2$-dimensional matrix $S^\top \lambda(t) S$.
  The larger eigenvalue (the red curve) is positive all the time, whereas the smaller eigenvalue (the blue curve) is negative.
  $(b)$ Smallest eigenvalue of the matrix $\mathcal{T}_{B,t}$.
  The fact that this is non-negative indicates that the time evolution map is completely positive (see the main text).
  }
  \label{fig:exactCP_ExactMaster}
\end{figure}
%%%%%%%%%%%%%%%%%%%%%%%%%%%%%%%%%%%%%%%%%%%%%%%%%%%%%%%%%

Figure \ref{fig:exactCP_ExactMaster} $(a)$ shows the two eigenvalues of the coefficient matrix $S^\top \lambda (t) S$ as a function of time.
As expected, the eigenvalues become constant after $\kappa t \simeq 10$.
We numerically checked the relation Eq.$\,$(\ref{eq:exactCP_lmdrelax}) in this long-time regime.
One of the eigenvalues shown by the blue curve is always negative.
This negativity asymptotically results in the (complete) positivity violation found in Sec.$\,$\ref{sec:Diag}.
On the other hand, the time evolution map including the transient regime ($0 \leq \kappa t \leq 10$ in the current setting) is completely positive.
To see this numerically, we consider the Choi matrix of the time evolution map.
From Eq.$\,$(\ref{eq:exactCP_rhoBt}), it reads $\sum_{m,n = 1}^d [\mathcal{T}_{B,t}]_{m,n} \, \dketbra{E_m}{E_n}$.
As shown in Appendix \ref{app:vect}, the time evolution map is completely positive if and only if this matrix is positive semidefinite.
Since $\{ \dket{E_m} \}_{m = 1,2,3}$ is an orthonormal set, we only need to consider the property of the matrix $\mathcal{T}_{B,t} \in \mathbb{C}^{3 \times 3}$ the $mn$-elements of which are $[\mathcal{T}_{B,t}]_{m,n} \ (m,n = 1,2,3)$.
In Fig.$\,$\ref{fig:exactCP_ExactMaster} $(b)$, we show the smallest eigenvalue of $\mathcal{T}_{B,t}$ along the time evolution.
It is non-negative as expected.
Therefore, the time evolution map is completely positive including the transient regime in this example.

%%%%%%%%%%%%%%%%%%%%%%%%%%%%%%%%%%%%%%%%%%%%%%%%%%%%%%%%%%%%%%%%%%%%%%%%%%%%%%%%%%%%%%%%
\section{$\mathcal{K}$ and $\mathcal{L}_s$ for a general class of settings}
\label{app:K2L3}

In this appendix, we consider a general class of settings introduced in Sec.$\,$\ref{sec:AE} ($\mathcal{L}_B$ and $\mathcal{L}_{\rm int}$ are assumed to contain only Hamiltonian terms) and calculate the second-order expansion of $\mathcal{K}$ and the third-order expansion of $\mathcal{L}_s$.
These are for clarification and extension of results in \cite{AzouitThesis}.
As in \cite{AzouitThesis}, we denote the interaction by $\mathcal{L}_{\rm int} (\rho) = - i [H_{\rm int}, \rho]$ with $H_{\rm int} = \sum_{k} A_k^\dagger \otimes B_k = \sum_k A_k \otimes B_k^\dagger$.
Occasionally, we consider the oscillator-qubit system in Sec.$\,$\ref{sec:JC} as a concrete example. 

\subsection{Second-order expansion of $\mathcal{K}$}
\label{app:K2L3_K2.geneal}

In Theorem 6 of \cite{AzouitThesis}, the second-order expansion of $\mathcal{K}$ was calculated for a general class of settings.
The theorem states that, there always exist gauge choices such that $\mathcal{K}$ is completely positive within the expansion order.
Although not proved explicitly, the result also implies that $\mathcal{K}$ for the partial trace is always noncompletely positive.
Here we repeat the calculation to clarify this point and to find a specific form of gauge superoperators leading to completely positive $\mathcal{K}$.

We first outline how we proceed the calculation.
The right-hand side of Eq.$\,$(\ref{eq:AE_LAKn}) is traceless and thus generally takes the form
\begin{equation*}
  \begin{gathered}
    \bar{\rho}_A \otimes \mathcal{L}_{s,n} (\rho_s) - \mathcal{L}_n(\rho_s) =  \mathcal{S}(X_{1} \bar{\rho}_A) \otimes \mathcal{B}_1 (\rho_s) \\
    + \mathcal{S}(\bar{\rho}_A X_{2}) \otimes \mathcal{B}_2 (\rho_s)
    + \mathcal{S}(X_{3} \bar{\rho}_A X_{4}) \otimes \mathcal{B}_3 (\rho_s),
  \end{gathered}
\end{equation*}
where the superoperator $\mathcal{S}$ is defined by $\mathcal{S}(X) = X - {\rm tr}_A(X) \bar{\rho}_A$, $\{ X_{i} \}_{i =1,2,3,4}$ are operators on $\mathscr{H}_A$, and $\{ \mathcal{B}_j \}_{j = 1,2,3}$ are superoperators on $\mathscr{H}_B$.
Then, Eq.$\,$(\ref{eq:AE_Kn}) reads
\begin{equation*}
  \begin{gathered}
    \mathcal{K}_n^{G} (\rho_s) = \mathcal{L}_A^{+} (\mathcal{S}(X_{1} \bar{\rho}_A)) \otimes \mathcal{B}_1 (\rho_s) \\
    + \mathcal{L}_A^{+} ( \mathcal{S}(\bar{\rho}_A X_{2}) ) \otimes \mathcal{B}_2 (\rho_s)
    + \mathcal{L}_A^{+} ( \mathcal{S}(X_{3} \bar{\rho}_A X_{4}) ) \otimes \mathcal{B}_3 (\rho_s) \\
    + \bar{\rho}_A \otimes G_n(\rho_s).
  \end{gathered}
\end{equation*}
Regarding the first and second terms on the right-hand side,
it was shown in Lemma 4 of \cite{Azouit17} that there exist $\{ Y_{i} \}_{i = 1,2}$ such that $\mathcal{L}_A^{+} (\mathcal{S}(X_{1} \bar{\rho}_A)) = Y_{1} \bar{\rho}_A$
and $\mathcal{L}_A^{+} ( \mathcal{S}(\bar{\rho}_A X_{2}) ) = \bar{\rho}_A Y_{2}$, even when $\bar{\rho}_A$ is not full rank.
These relations are essential in proving the Kraus map form as shown below.
The fourth term on the right-hand side represents the gauge degree of freedom.
Here we restrict our gauge choice to the following form which only modifies the other terms without generating independent ones;
\begin{equation*}
  G_n(\rho_s) = g_1 \mathcal{B}_1 (\rho_s) + g_2 \mathcal{B}_2 (\rho_s) + g_3 \mathcal{B}_3 (\rho_s),
\end{equation*}
with $g_i \in \mathbb{C} \ (i = 1,2,3)$.
Redefining $Y_{i}$ so that the $g_i$ term is absorbed for $i = 1,2$, we obtain
\begin{equation*}
  \begin{gathered}
    \mathcal{K}_n^G (\rho_s) = Y_{1} \bar{\rho}_A \otimes \mathcal{B}_1 (\rho_s)
    + \bar{\rho}_A Y_{2} \otimes \mathcal{B}_2 (\rho_s)
    + Z \otimes \mathcal{B}_3 (\rho_s),
  \end{gathered}
\end{equation*}
with
\begin{equation*}
  Y_{1} \bar{\rho}_A = \mathcal{L}_A^{+} (\mathcal{S}(X_{1} \bar{\rho}_A)) + g_1 \bar{\rho}_A,
\end{equation*}
\begin{equation*}
  \bar{\rho}_A Y_{2} = \mathcal{L}_A^{+} ( \mathcal{S}(\bar{\rho}_A X_{2}) ) + g_2 \bar{\rho}_A,
\end{equation*}
and
\begin{equation*}
  Z = \mathcal{L}_A^{+} ( \mathcal{S}(X_{3} \bar{\rho}_A X_{4}) ) + g_3 \bar{\rho}_A.
\end{equation*}

Using these formulas, we can calculate the first- and second-order contributions as follows;
\begin{equation}
  \mathcal{K}_1^G (\rho_s) = - i \sum_k (F_k \otimes B_k^\dagger) (\bar{\rho}_A \otimes \rho_s) + (h.c.),
  \label{eq:K2L3_K_1}
\end{equation}
and
\begin{equation}
  \begin{gathered}
    \mathcal{K}_2^G (\rho_s) = i \sum_k V_k \bar{\rho}_A \otimes (\mathcal{C}_B^k (\rho_s) - \mathcal{L}_B (B_k^\dagger) \rho_s) \\
    + i \sum_k (V_k \otimes \mathcal{L}_B (B_k^\dagger) ) (\bar{\rho}_A \otimes \rho_s) \\
    - \sum_{k,j} \braket{A_k^\dagger} (V_j \otimes B_j^\dagger B_k) (\bar{\rho}_A \otimes \rho_s) \\
    + \sum_{k,j} ( U_{k,j} \otimes B_k B_j^\dagger ) (\bar{\rho}_A \otimes \rho_s) \\
    - \sum_{k,j} Z_{k,j} \otimes B_k^\dagger \rho_s B_j + (h.c.).
  \end{gathered}
  \label{eq:K2L3_K_2}
\end{equation}
with $\mathcal{C}_B^k (\rho_s) = \mathcal{L}_B (B_k^\dagger \rho_s) - B_k^\dagger \mathcal{L}_B (\rho_s)$.
In these equations, we have introduced $\braket{A_k} = {\rm tr}_A (A_k \bar{\rho}_A)$.
The inverse of $\mathcal{L}_A$ has led to
\begin{equation*}
  F_k \bar{\rho}_A = - \mathcal{L}_A^{+} (\mathcal{S} (A_k \bar{\rho}_A)) + f_k \bar{\rho}_A,
\end{equation*}
\begin{equation*}
  V_k \bar{\rho}_A = \mathcal{L}_A^{+} (\mathcal{S} (F_k \bar{\rho}_A)) + v_k \bar{\rho}_A,
\end{equation*}
\begin{equation*}
  U_{k,j} \bar{\rho}_A = \mathcal{L}_A^{+} (\mathcal{S} (A_k^\dagger F_j \bar{\rho}_A)) + u_{k,j} \bar{\rho}_A,
\end{equation*}
and
\begin{equation*}
  Z_{k,j} = \mathcal{L}_A^{+} (\mathcal{S} (F_k \bar{\rho}_A \bar{A}_j^\dagger)) - z_{k,j} \bar{\rho}_A,
\end{equation*}
where $\bar{A}_j = A_j - \braket{A_j} I_A$.
The gauge superoperators have been set as
\begin{equation*}
  G_1 (\rho_s) = - i \sum_k f_k B_k^\dagger \rho_s + (h.c.),
\end{equation*}
and
\begin{equation*}
  \begin{gathered}
    G_2 (\rho_s) = i \sum_k v_k \mathcal{C}_B^k (\rho_s)
    - \sum_{k,j} \braket{A_k^\dagger} v_j B_j^\dagger B_k \rho_s \\
    + \sum_{k,j} u_{k,j} B_k B_j^\dagger \rho_s
    + \sum_{k,j} z_{k,j} B_k^\dagger\rho_s B_j
    + (h.c.).
  \end{gathered}
\end{equation*}
When $\mathcal{L}_B$ contains only Hamiltonian terms, $\mathcal{C}_B^k (\rho_s) = \mathcal{L}_B (B_k^\dagger) \rho_s$ and the first line on the right side of Eq.$\,$(\ref{eq:K2L3_K_2}) vanishes.
Thus, we neglect this term in the following.

Note
\begin{equation*}
  \begin{gathered}
    \sum_{k,j} Z_{k,j} \otimes B_k^\dagger \rho_s B_j + (h.c.) = \\
    \sum_{k,j} (Z_{k,j} + Z_{j,k}^\dagger) \otimes B_k^\dagger \rho_s B_j,
  \end{gathered}
\end{equation*}
where
\begin{equation*}
  \begin{gathered}
    Z_{k,j} + Z_{j,k}^\dagger = \\
    \mathcal{L}_A^{+} ( \mathcal{S} (F_k \bar{\rho}_A \bar{A}_j^\dagger + \bar{A}_k \bar{\rho}_A F_j^\dagger ) ) - (z_{k,j} + z_{j,k}^*) \bar{\rho}_A.
  \end{gathered}
\end{equation*}
When $\mathcal{L}_A$ is given by Eq.$\,$(\ref{eq:AE_LA}), it was shown in \cite{AzouitThesis} that the $\mathcal{L}_A^{+}$ part on the right-hand side reads
\begin{equation*}
  \begin{gathered}
    \mathcal{L}_A^{+} ( \mathcal{S} (F_k \bar{\rho}_A \bar{A}_j^\dagger + \bar{A}_k \bar{\rho}_A F_j^\dagger ) ) = \\
    \mathcal{S} (F_k \bar{\rho}_A F_j^\dagger) - \mathcal{L}_A^{+} ( \sum_\mu [L_{A,\mu}, F_k] \bar{\rho}_A [L_{A,\mu}, F_j]^\dagger ).
  \end{gathered}
\end{equation*}

Combining all the terms together, we eventually obtain, up to the second-order of $\epsilon$,
\begin{equation}
  \begin{gathered}
    \mathcal{K}^G (\rho_s) = W (\bar{\rho}_A \otimes \rho_s) W^\dagger \\
    - \epsilon^2 \mathcal{L}_A^{+} \Big( \sum_\mu P_\mu (\bar{\rho}_A \otimes \rho_s) P_\mu^\dagger \Big) \\
    + \epsilon^2 \sum_{k,j} \mu_{k,j} (I_A \otimes B_k^\dagger) (\bar{\rho}_A \otimes \rho_s) (I_A \otimes B_j), \\
  \end{gathered}
  \label{eq:K2L3_K2nd}
\end{equation}
with $W = I_A \otimes I_B - i \epsilon M + \epsilon^2 N$, $M = \sum_k F_k \otimes B_k^\dagger$,
$N = \sum_{k,j} (U_{k,j} \otimes B_k B_j^\dagger - \braket{A_k^\dagger} V_j \otimes B_j^\dagger B_k) + i \sum_k V_k \otimes \mathcal{L}_B(B_k^\dagger)$,
$P_\mu = \sum_k [L_{A,\mu}, F_k] \otimes B_k^\dagger$, and $\mu_{k,j} = z_{k,j} + z_{j,k}^* - {\rm tr}_A (F_k \bar{\rho}_A F_j^\dagger)$.
The second line is completely positive because $- \mathcal{L}_A^{+}$ is a completely positive map from Lemma 1 of \cite{Azouit17}.
Therefore, $\mathcal{K}^G$ is completely positive if and only if the matrix $\mu$, the elements of which are defined by $[\mu]_{k,j} = \mu_{k,j}$, is positive semidefinite.

The partial trace parametrization corresponds to $G_1 = G_2 = 0$, and thus $z_{k,j} = 0$.
In this case, $\mu_{k,j} = - {\rm tr}_A (F_k^\dagger \bar{\rho}_A F_j)$ and it is negative semidefinite.
This consideration concludes that $\mathcal{K}^{G = 0}$ is always a noncompletely positive map.
On the other hand, if one sets
\begin{equation}
  z_{k,j} = \frac{z}{2} {\rm tr}_A (F_k \bar{\rho}_A F_j^\dagger),
  \label{eq:K2L3_zkj}
\end{equation}
with $z \in \mathbb{R}_{\geq 1}$, the matrix $\mu$ is positive semidefinite.
Therefore, there always exist gauge choices ensuring complete positivity of $\mathcal{K}^{G}$.

Complete positivity of $\mathcal{K}^{G}$ is attained irrespective of the values of $f_k$, $v_k$, and $u_{k,j}$.
If we set $f_k = v_k = u_{k,j} = 0$, however, $G(\rho_s) = \rho_s + \sum_{k,j} (z_{k,j} + z_{j,k}^*)  B_k^\dagger \rho_s B_j$ and $\mathcal{K}$ does not preserve the trace.
The simplest gauge choice recovering the trace preservation is $f_k = v_k = 0$ and $u_{k,j} = - z_{j,k}$.
With Eq.$\,$(\ref{eq:K2L3_zkj}), the gauge superoperator reads
\begin{equation}
  G(\rho_s) = \epsilon^2 z \sum_{k,j} {\rm tr}_A (F_k \bar{\rho}_A F_j^\dagger) (B_k^\dagger \rho_s B_j - \frac{1}{2} \{ B_j B_k^\dagger, \rho_s \} ).
  \label{eq:K2L3_G2}
\end{equation}

\subsection{Entanglement and non-positivity of $\mathcal{K}^{G = 0}$}
\label{app:K2L3_K2.positive}

While $\mathcal{K}^{G = 0}$ is not completely positive, it is not certain whether it can be positive or not.
To address this issue, let us first consider the oscillator-qubit system in Sec.$\,$\ref{sec:JC} as an example. 
Using the formula Eq.$\,$(\ref{eq:K2L3_K2nd}), $\mathcal{K}^{G = 0}$ up to the second-order expansion is given by Eq.$\,$(\ref{eq:diss_KG0}) in  this case.
To see that it is not positive, let $\ket{0} \in \mathscr{H}_A$ be the vacuum state of the oscillator and $\ket{\psi} \in \mathscr{H}_B$ be a state.
The matrix element of $\mathcal{K}^{G=0}$ with respect to $\ket{0,\psi} = \ket{0} \otimes \ket{\psi}$ reads
\begin{equation*}
  \begin{gathered}
    \braket{0,\psi|\mathcal{K}^{G=0} (\rho_B)|0,\psi} = \braket{0|\bar{\rho}_A|0} \Big[ \Big( 1 + \frac{4g^2 n_{\rm th}}{\gamma^2} \Big) \braket{\psi|\rho_B|\psi} \\
    - \frac{4 g^2 (1 + n_{\rm th})}{\gamma^2} \, \braket{\psi_+|\rho_B|\psi_+} - \frac{4 g^2 n_{\rm th}^2}{\gamma^2 (1+n_{\rm th})} \, \braket{\psi_-|\rho_B|\psi_-} \Big].
  \end{gathered}
\end{equation*}
with $\ket{\psi_\pm} = \sigma_\pm \ket{\psi}$.
If $\rho_B$ is a pure state, its kernel is not empty. Suppose $\ket{\psi}$ is in the kernel of $\rho_B$, that is, $\rho_B \ket{\psi} = 0$.
Since either $\braket{\psi_+|\rho_B|\psi_+}$ or $\braket{\psi_-|\rho_B|\psi_-}$ is nonzero, we obtain $\braket{0,\psi|\mathcal{K}^{G=0} (\rho_B)|0,\psi} < 0$.
Consequently, $\mathcal{K}^{G=0}$ is not positive.

To extend the analysis to a general class of settings, we note a theorem regarding assignment maps, which was first proved for qubit reduced systems in \cite{CPdebate} and later generalized in \cite{Jordan04}.
With the notations in Sec.$\,$\ref{sec:diss}, let $\Phi:\mathscr{D}(\mathscr{H}_B) \to \mathscr{D}(\mathscr{H}_A \otimes \mathscr{H}_B)$ be a map assigning a reduced state in $\mathscr{H}_B$ a total state in $\mathscr{D}(\mathscr{H}_A \otimes \mathscr{H}_B)$.
Then, the following statements are equivalent;
\begin{itemize}
  \item $\Phi$ is linear, consistent (${\rm tr}_A \circ \Phi = \mathcal{I}_B$), and positive ($\Phi(\rho_B) \in \mathscr{D}(\mathscr{H}_A \otimes \mathscr{H}_B)$ for every $\rho_B \in \mathscr{D}(\mathscr{H}_B)$).
  \item For every $\rho_B \in \mathscr{D}(\mathscr{H}_B)$, $\Phi(\rho_B) = \rho_A \otimes \rho_B$ where $\rho_A \in \mathscr{D}(\mathscr{H}_A)$ is independent of $\rho_B$.
\end{itemize}
We now apply this theorem to $\mathcal{K}^{G=0}$.
Since it is linear and consistent, $\mathcal{K}^{G=0}$ is positive if and only if $\mathscr{M}_{\rm inv}$ is characterized by product states with a fixed state on $\mathscr{D}(\mathscr{H}_A)$.
Due to the interaction term, $\mathscr{M}_{\rm inv}$ at the second-order expansion includes entangled states in general.
Therefore, the corresponding $\mathcal{K}^{G=0}$ is nonpositive.

One might argue that nonproduct states should also be obtained in the first-order expansion of $\epsilon$.
On this point, we note that the first-order expansion of $\mathcal{K}^{G=0}$ reads (see Eq.$\,$(\ref{eq:K2L3_K2nd}))
\begin{equation}
  \begin{gathered}
    \mathcal{K}^{G=0} (\rho_B) = \\
    \bar{\rho}_A \otimes \rho_B - i \epsilon (M (\bar{\rho}_A \otimes \rho_B) - (\bar{\rho}_A \otimes \rho_B) M^\dagger),
  \end{gathered}
  \label{eq:diss_1stOrder}
\end{equation}
where $M$ satisfies ${\rm tr}_A (M \bar{\rho}_A) = 0$.
In agreement with the above theorem, this is not positive.
We here note that this can be rewritten as
\begin{equation}
  \begin{gathered}
    \mathcal{K}^{G=0} (\rho_B) = \\
    (I_A \otimes I_B - i \epsilon M) (\bar{\rho}_A \otimes \rho_B) (I_A \otimes I_B - i \epsilon M)^\dagger \\
    - \epsilon^2 M (\bar{\rho}_A \otimes \rho_B) M^\dagger.
  \end{gathered}
  \label{eq:diss_1stOrderRewritten}
\end{equation}
Within the accuracy of the first-order expansion, we can neglect the third line and obtain a Kraus map form of $\mathcal{K}^{G=0}$.
Since the neglect of the third line violates the consistency at order $\epsilon^2$, this result does not contradict with the above theorem either.
In the second-order expansion, on the other hand, there is a negative term that cannot be eliminated within that accuracy.
This results in the positivity violation of $\mathcal{K}^{G=0}$.

\subsection{Third-order expansion of $\mathcal{L}_s$}
\label{app:K2L3_L3}

While the map $\mathcal{K}^{G=0}$ resulting from the first-order expansion, i.e., Eq.$\,$(\ref{eq:diss_1stOrder}), is not completely positive, its rewriting with Eq.$\,$(\ref{eq:diss_1stOrderRewritten}) shows that the nonpositivity only yields a term of order $\epsilon^2$.
This explains the reason why the Lindblad form must always be obtained including the second-order contribution, as shown in \cite{Azouit17}.

The nonpositivity of $\mathcal{K}^{G=0}$ at second-order indicates danger for the third-order expansion and higher.
For the oscillator-qubit system, the third-order contribution $\mathcal{L}_{s,3}^{G=0}$ vanishes and the non-Lindblad form is obtained from the fourth-order expansion.
Even when $\mathcal{L}_{s,3}^{G=0}$ does not vanish, it was shown in Theorem 9 of \cite{AzouitThesis} that, if the interaction Hamiltonian includes one tensor-product term and the second-order contribution does not vanish, $\mathcal{L}_{s}^{G=0}$ admits the Lindblad form up to the third-order expansion.
Here we consider the third-order expansion with a general interaction Hamiltonian $H_{\rm int} = \sum_{k} A_k^\dagger \otimes B_k = \sum_k A_k \otimes B_k^\dagger$.

Since ${\rm tr}_A \circ \mathcal{K}_n^{G=0} = 0 \ (n = 1,2,\dots)$, from Eq.$\,$(\ref{eq:inv.cond._Ln}), $\mathcal{L}_{s,n}^{G=0}(\rho_B) = {\rm tr}_A \Big( (\mathcal{I}_A \otimes \mathcal{L}_B + \mathcal{L}_{\rm int}) (\mathcal{K}_{n-1}^{G=0} (\rho_B) ) \Big)$.
With Eqs.$\,$(\ref{eq:AE_asymexp0}), (\ref{eq:K2L3_K_1}), and (\ref{eq:K2L3_K_2}), we obtain $\{ \mathcal{L}_{s,n}^{G=0} \}_{n=1,2,3}$ as
\begin{equation*}
  \mathcal{L}_{s,1}^{G=0} (\rho_B) = - i \sum_k \braket{A_k} B_k^\dagger \, \rho_B + (h.c.),
\end{equation*}
\begin{equation*}
  \mathcal{L}_{s,2}^{G=0} (\rho_B) = \sum_{j,k} \braket{A_j^\dagger F_k} [B_k^\dagger \rho_B, B_j] + (h.c.),
\end{equation*}
\begin{equation*}
  \begin{gathered}
    \mathcal{L}_{s,3}^{G=0} (\rho_B) = - \sum_{j,k} \braket{A_j^\dagger V_k} [ \mathcal{L}_B (B_k^\dagger) \rho_B, B_j] \\
    - \sum_{i,j,k} i \braket{A_i^\dagger V_j} \braket{A_k^\dagger} [B_j^\dagger B_k \rho_B, B_i] \\
    + \sum_{i,j,k} i \braket{A_i^\dagger U_{k,j} } [B_k B_j^\dagger \rho_B, B_i] \\
    - \sum_{i,j,k} i {\rm tr}_A (A_i^\dagger Z_{k,j}) [B_k^\dagger \rho_B B_j, B_i] + (h.c.).
  \end{gathered}
\end{equation*}

Since $\mathcal{L}_s^{G=0}$ preserves the Hermitian property and have trace zero, we only need to focus on the dissipation part to see if $\mathcal{L}_s^{G=0}$ is a Lindbladian.
The dissipation part of $\mathcal{L}_{s,2}^{G=0}$ reads $\sum_{j,k} \eta_{k,j} B_k^\dagger \rho_B B_j$ with $\eta_{k,j} = \braket{A_j^\dagger F_k} + \braket{A_k^\dagger F_j}^*$.
As proved in \cite{Azouit17}, the matrix $\eta$ defined by $[\eta]_{k,j} = \eta_{k,j}$ is positive semidefinite in general.
The dissipation part of $\mathcal{L}_{s,3}^{G=0}$ reads $\sum_{k} (Q_k^\dagger \rho_B B_k + B_k^\dagger \rho_B Q_k)$ with
$Q_k^\dagger = - \sum_i \braket{A_k^\dagger V_i} [\mathcal{L}_B (B_i^\dagger) + i \sum_j \braket{A_j^\dagger} B_i^\dagger B_j ] + i \sum_{i,j} [ \braket{A_k^\dagger U_{i,j}} + {\rm tr}_A (A_i^\dagger (Z_{j,k} + Z_{k,j}^\dagger)) ] B_i B_j^\dagger$.
If we assume that $\eta$ is invertible (or, equivalently, $\eta$ is positive definite), then the dissipation part of $\mathcal{L}_s^{G=0}$ reads, up to the third-order of $\epsilon$,
\begin{equation*}
  \sum_{j,k} \eta_{k,j} R_k^\dagger \rho_B R_j,
\end{equation*}
with $R_k^\dagger = \epsilon B_k^\dagger + \epsilon^2 \sum_j [\eta^{-1}]_{j,k} Q_j^\dagger$.
Therefore, $\mathcal{L}_s^{G=0}$ admits the Lindblad form.

When $\eta$ is not invertible, the third-order contribution can generate a new dissipator that is independent of the second-order contribution.
Studies on such cases are in progress.
Note that, when $\mathcal{L}_B$ and $\mathcal{L}_{\rm int}$ contain only Hamiltonian terms, $\mathcal{L}_{\rm tot}$ is a Lindbladian irrespective of the sign of $\epsilon$.
Thus, the structure of $\mathcal{L}_s^{G=0}$ is not affected by flipping the sign of $\epsilon$ as $\epsilon \to - \epsilon$.
This implies that the new dissipator generated at the third-order should either always include a negative coefficient or vanish.
To elaborate on it, let us consider an extreme case with $\eta = 0$.
In general, we can write $Q_k$ as $Q_k = q_k B_k + \bar{Q}_k$ with $q_k$ a complex number and $\bar{Q}_k$ an operator on $\mathscr{H}_B$ that is linearly independent of $B_k$.
The dissipation part of $\mathcal{L}_{s}^{G=0}$ then reads $\epsilon^3 \sum_k [(\mathcal{D}[\bar{Q}_k+B_k] - \mathcal{D}[\bar{Q}_k-B_k])/2 + 2 {\rm Re}(q_k) \mathcal{D}[B_k]]$ up to the third-order expansion.
On the one hand, the first part with $\bar{Q}_k$ always includes a negative coefficient, and the structure of $\mathcal{L}_{s}^{G=0}$ is independent of the sign of $\epsilon$.
On the other hand, the second part with ${\rm Re}(q_k)$ flips its sign by the transformation $\epsilon \to - \epsilon$.
According to the above discussion, we expect to have ${\rm Re}(q_k) = 0$ under the condition of $\eta = 0$.
However, we have not been able to prove it yet.

%%%%%%%%%%%%%%%%%%%%%%%%%%%%%%%%%%%%%%%%%%%%%%%%%%%%%%%%%%%%%%%%%%%%%%%%%%%%%%%%%%%%%%%%
\section{Quantum discord and complete positivity violation}
\label{app:qdiscord}

In Sec.$\,$\ref{sec:diss_initial.correlated}, we interpret the complete positivity violation in terms of correlations present in states on an invariant manifold.
To support this interpretation, we examine the relation between the strength of correlations on an invariant manifold and the degree of complete positivity violation.
We take the qubit-qutrit system in Sec.$\,$\ref{sec:Diag} for this analysis.
States on the invariant manifold are characterized by Eq.$\,$(\ref{eq:Diag_KPT}) in this case.
For this class of bipartite states, there is a strong connection between complete positivity of the reduced dynamics and quantum discord, as we elaborate in the next paragraph.
Hence, we take quantum discord as a measure of correlations.

Prior to discussing a methodology for estimating quantum discord, we make several remarks on connections between complete positive reduced dynamics and quantum discord.
These remarks are intended to complement the historical overview given in Sec.$\,$\ref{sec:diss_initial.correlated}.
It was shown in \cite{Rosario08} that, for Hamiltonian systems, vanishing quantum discord implies completely positive reduced dynamics.
In connection with the present study, their proof can be straightforwardly extended to any completely positive processes in the total system, including the time evolution with a Lindblad generator.
Hence, we must have completely positive reduced dynamics if all states on the invariant manifold have vanishing quantum discord. 
In \cite{Shabani09}, it was shown that, for initial states in the form of Eq.$\,$(\ref{eq:Diag_KPT}), the converse is also true, that is, completely positive reduced dynamics imply vanishing quantum discord.
Here we should emphasize that the condition is complete positive reduced dynamics for all possible unitary transformations.
In fact, the proof in \cite{Shabani09} relies on this fact as it employs a special form of unitary transformation to elucidate restrictions on the environment operators imposed by the complete positivity.
Therefore, the criterion is only sufficient and we should not necessarily have vanishing quantum discord even if we observe complete positive reduced dynamics in the present analysis.
As an example, let us consider the dispersively coupled qubit-qubit system in Sec.$\,$\ref{sec:Diag}.
As mentioned in the last paragraph in Sec.$\,$\ref{sec:Diag_any.order}, the generator $\mathcal{L}_s^{G=0}$ is in the Lindblad form, ensuring complete positive reduced dynamics.
However, this merely confirms the complete positivity for the total system evolution in that case and does not guarantee it for all unitary transformations.
In fact, we estimated quantum discord in this case following the procedure described below and observed nonzero quantum discord.

Now we delve into the estimation of quantum discord.
When considering local projective measurements on $\mathscr{H}_B$, quantum discord of a quantum state $\rho \in \mathscr{D}(\mathscr{H}_A \otimes \mathscr{H}_B)$, denoted as $\delta_B(\rho)$, is defined by \cite{Henderson01,Ollivier01,Modi13}
\begin{equation}
    \delta_B(\rho) = S(\rho_B) - S(\rho) + \min_{ \{ \Pi_j^B \}  } S(A|\{ \Pi_j^B \}),
    \label{eq:qdiscord_def}
\end{equation}
with $S(\rho) = - {\rm tr} (\rho \ln \rho)$ the von Neumann entropy. 
The last term on the right-hand side, $\min_{ \{ \Pi_j^B \}  } S(A|\{ \Pi_j^B \})$, encapsulates the conditional entropy of subsystem $\mathscr{H}_A$ subsequent to the measurement on $\mathscr{H}_B$ associated with rank-1 orthogonal projectors $\{ \Pi_j^B \} \ (j = 1,\dots,{\rm dim}(\mathscr{H}_B))$.
The conditional entropy is calculated using the formula
\begin{equation*}
    S(A|\{ \Pi_j^B \}) = \sum_{j=1}^{{\rm dim}(\mathscr{H}_B)} p_j S(\rho_{A|\Pi_j^B}),
\end{equation*}
with $p_j = {\rm tr}(\Pi_j^B \rho)$ the probability of the outcome labeled by $j$ and $\rho_{A|\Pi_j^B} = \Pi_j^B \rho \Pi_j^B / p_j$ the associated postmeasurement state. 
The minimization is performed over all possible projective measurements $\{ \Pi_j^B \}$.

Currently, no efficient method exists to compute $\min_{ \{ \Pi_j^B \} } S(A|\{ \Pi_j^B \})$ to our knowledge. Consequently, we estimate this minimum by employing a brute-force approach, which involves random sampling of projectors $\{ \Pi_j^B \}$.
By diagonalizing randomly generated observables in $\mathscr{H}_B$, we obtain a set of orthogonal projectors.
Let $N_{\rm proj}$ be the number of generated observables. 
Given a state, we compute the minimum value of the conditional entropy in this set, from which we estimate the quantum discord of the state using Eq.$\,$(\ref{eq:qdiscord_def}).
Since sampling all orthogonal projectors $\{ \Pi_j^B \}$ is infeasible in practice, the estimated quantum discord should be considered an upper bound.

Our current objective is to estimate a representative value of quantum discord of the invariant manifold $\mathscr{M}_{\rm inv}$, which we denote as $\delta_B (\mathscr{M}_{\rm inv})$.
To achieve this, we adopt a straightforward approach, in which we randomly sample $N_{\rm state}$ states on the invariant manifold $\rho_a \in \mathscr{M}_{\rm inv} \ (a = 1, 2, \dots, N_{\rm state})$ and then estimate $\delta_B (\mathscr{M}_{\rm inv})$ through the average value of quantum discord among those states;
\begin{equation}
    \delta_B (\mathscr{M}_{\rm inv}) = \frac{1}{N_{\rm state}} \sum_{a=1}^{N_{\rm state}} \delta_B (\rho_a).
    \label{eq:qdiscord_def.inv}
\end{equation}
A set $\{ \rho_a \}$ can be prepared by randomly generating $\rho_B \in \mathscr{D}(\mathscr{H}_B)$ and then inserting them into Eq.$\,$(\ref{eq:Diag_KPT}).
Given that the map $\mathcal{K}^{G=0}$ is not positive, we check the eigenvalues of $\mathcal{K}^{G=0}(\rho_B)$ to ensure that $\{ \rho_a \}$ are all positive semidefinite.

The resulting quantum discord is shown in the left figure in Fig.$\,$\ref{fig:qdiscord}.
For the numbers of random sampling, we take $N_{\rm proj} = N_{\rm state} = 500$.
We confirmed that the results do not change significantly by employing larger values of $N_{\rm proj}$ and $N_{\rm state}$.

%%%%%%%%%%%%%%%%%%%%%%%%%%%%%%%%%%%%%%%%%%%%%%%%%%%%%%%%%%%%%%%%%%%%%%%%%%%%%%%%%%%%%%%%
\section{Gauge choice for practical applications}
\label{app:gauge}

The geometric formulation introduced in Sec.$\,$\ref{sec:AE} involves the gauge degree of freedom.
It was previously anticipated in \cite{Azouit17} that the gauge degree of freedom can be leveraged to restore the Lindblad form of the generator $\mathcal{L}_s$.
In Sec.$\,$\ref{sec:JC}, however, we found an example where such a restoration is unattainable.
Given this background, this appendix explores the potential roles that the gauge degree of freedom plays in the practical applications.

In practical applications, care should be taken when the initial state is determined in the reduced system.
For the partial trace parametrization, the domain of initial reduced states must be restricted \cite{Jordan04} to ${\rm tr}_A (\mathscr{M}_{\rm inv})$ (see Sec.$\,$\ref{sec:diss} for the notation), which is a subset of $\mathscr{D}(\mathscr{H}_B)$ in general.
By linearity, extending the domain of $\mathcal{K}^{G=0}$ to $\mathscr{D}(\mathscr{H}_B)$ (all density matrices) is possible.
However, such extension leads to an initial total state that is not positive semidefinite and thus is unphysical.
Recall that $\mathcal{K}^{G=0}$ is not positive at the second-order expansion in general.
Therefore, strictly speaking, the domain restriction is necessary once we consider contributions higher than the second-order, even when $\mathcal{L}_s^{G=0}$ admits the Lindblad form.

In order to cover all the states in $\mathscr{M}_{\rm inv}$, the map $\mathcal{K}^{G}: \mathscr{D}(\mathscr{H}_B) \to \mathscr{M}_{\rm inv}$ needs to be positive and surjective.
If $G = G_{ps}$ has those properties, then the time evolution map with this gauge choice $\exp(\mathcal{L}_s^{G_{ps}} t)$ at least preserves positivity.
Indeed, since $\mathcal{K}^{G_{ps}}$ maps $\mathscr{D}(\mathscr{H}_B)$ onto all the positive states in the subspace spanned by $\mathscr{M}_{\rm inv}$,
by injectivity a nonpositive $\rho_s$ obtained with $\exp(\mathcal{L}_s^{G_{ps}} t)$ can also be mapped only onto a nonpositive $\rho$ in the space spanned by $\mathscr{M}_{\rm inv}$.

A question then arises whether such a $G_{ps}$ exists in general.
To our knowledge, the existence of $G_{ps}$ has been confirmed only for a dispersively coupled two qubit system in \cite{Alain20}.
The rest of this appendix examines this question for the qubit-qutrit system introduced in Sec.$\,$\ref{sec:Diag}.
In fact, we see that the approach employed for the two qubit system is not applicable.

As employed in \cite{Alain20}, we consider the following gauge choice;
\begin{equation}
  (\mathcal{I}_B + G) (\rho_s) = \sum_{m,n = 1}^d c_{m,n} E_m \rho_s E_n,
  \label{eq:surj_diagonalG}
\end{equation}
where the complex coefficients $c_{m,n}$ satisfy $c_{m,n}^* = c_{n,m}$ and $c_{m,m} = 1$ for the Hermitian and trace preservation of $\mathcal{K}^{G}$, respectively.
This is one of the simplest gauge choices because the corresponding supermatrix (see Appendix \ref{app:vect}) is diagonal as $\hat{\mathcal{I}}_B + \hat{G} = \sum_{m,n = 1}^d c_{m,n} E_n \otimes E_m$.
From this representation, it is clear that $(\mathcal{I}_B + G)^{-1}$ exists if and only if all the coefficients $\{ c_{m,n} \}$ are nonzero.
In that case, one can calculate $\mathcal{L}_s$ with different gauge choices using Eq.$\,$(\ref{eq:AE_LG}).
In fact, within the form of Eq.$\,$(\ref{eq:surj_diagonalG}), the generator $\mathcal{L}_s^G$ is independent of gauge choice because $\mathcal{L}_s^{G=0}$ is diagonal (see Eq.$\,$(\ref{eq:Diag_LsPT})) and thus commutes with $(\mathcal{I}_B + G)$.

%%%%%%%%%%%%%%%%%%%%%%%%%%%%%%%%%%%%%%%%%%%%%%%%%%%%%%%%%
\begin{figure}[t]
  \includegraphics[keepaspectratio, scale=0.55]{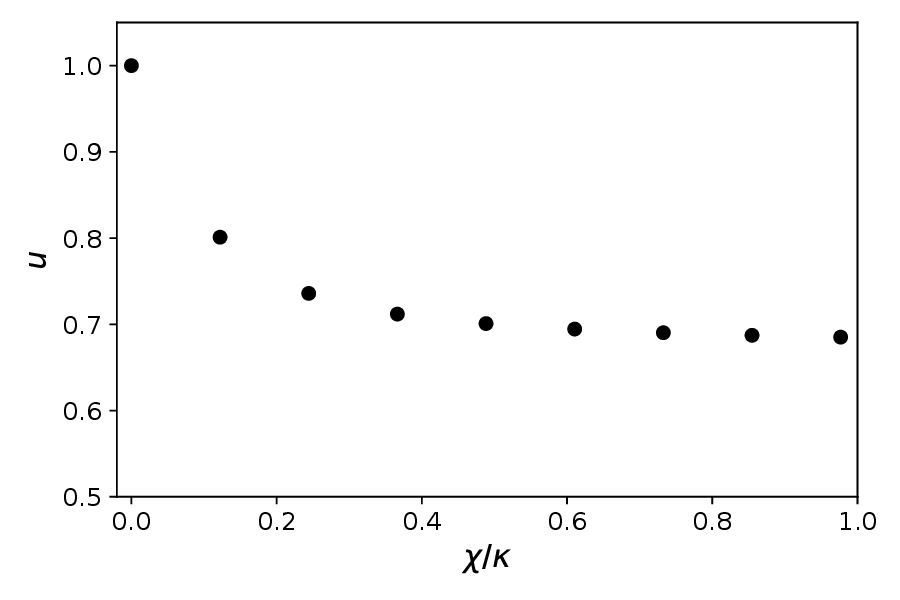}
  \caption{Maximial multiplying factor $u$ such that, with $c_{m,n} = u \tilde{c}_{m,n}$, positivity is ensured for the composite qubit-qutrit state in the case of any three-state superposition, as a function of the coupling strength ($(\chi_1,\chi_2,\chi_3) = (0,\chi,2\chi)$) in unit of the qubit damping rate. The figure is obtained by varying the value of $\chi$, while fixing $\Omega / \kappa = 0.5$ and $\Delta / \kappa = - 0.5$.}
  \label{fig:surj_maxfactor}
\end{figure}
%%%%%%%%%%%%%%%%%%%%%%%%%%%%%%%%%%%%%%%%%%%%%%%%%%%%%%%%%

From Eqs.$\,$(\ref{eq:AE_KG}) and (\ref{eq:Diag_KPT}), $\mathcal{K}^G$ with Eq.$\,$(\ref{eq:surj_diagonalG}) reads
\begin{equation*}
  \mathcal{K}^G (\rho_s) = \sum_{m,n = 1}^d c_{m,n} Q_{m,n} \otimes E_m \rho_s E_n.
\end{equation*}
The absolute values of $\{ c_{m,n} \}$ are bounded above by the necessity of $\mathcal{K}^G$ to be positive.
On the other hand, taking too small values results in nonsurjective $\mathcal{K}^G$.
The problem is thus to find their optimal values ensuring both the surjectivity and positivity.
We can use any qudit state, $\rho_s$, to obtain constraints on these coefficients and explore the space of the total state they allow to span.
In particular, if we consider a qudit in a superposition of two levels between $m$ and $n$, we can use the results of \cite{Alain20} to show that there exists an optimal value $\tilde{c}_{m,n}$ that ensures both the surjectivity and positivity.
As this is valid for any pair $(m,n)$, one might be tempted to choose the set $\{ \tilde{c}_{m,n} \}$ as the optimal one.
However, strikingly, this set is not admissible for $d > 2$ as it does not ensure the positivity of $\mathcal{K}^G$.
In other words, one needs to take smaller values of the coefficients than $\{ \tilde{c}_{m,n} \}$ as $c_{m,n} = u \tilde{c}_{m,n}$ with $u \leq 1$ to obtain positive $\mathcal{K}^G$.
For instance, for the qubit-qutrit system ($d = 3$) with Eq.$\,$(\ref{eq:Diag_qubit}), we plot the maximum value of $u$ that ensures the positivity of $\mathcal{K}^G$ in Fig.$\,$\ref{fig:surj_maxfactor}.
The necessity to take $u < 1$ implies that not all the qutrit states can be captured in this gauge choice.
In particular, some superpositions of only two levels require $u = 1$.
Thus, this counterexample allows us to conclude that, within the diagonal gauge given by Eq.$\,$(\ref{eq:surj_diagonalG}), there exists no set $\{ c_{m,n} \}$
leading to positive and surjective $\mathcal{K}^G$ when $d > 2$.
This result motivates us to consider nondiagonal gauge superoperators.
Such an extension is currently under investigation.

%%%%%%%%%%%%%%%%%%%%%%%%%%%%%%%%%%%%%%%%%%%%%%%%%%%%%%%%%%%%%%%%%%%%%%%%%%%%%%%%%%%%%%%%

%%%%%%%%%%%%%%%%%%%%%%%%%%%%%%%%%%%%%%%%%%%%%%%%%%%%%%%%%%%%%%%%%%%%%%%%%%%%%%%%%%%%%%%%


\begin{thebibliography}{66}

  \bibitem{Breuer02} H.-P. Breuer and F. Petruccione,
  {\it The Theory of Open Quantum Systems}, (Oxford University Press, New York, 2002).

  \bibitem{NielsenChuang} M. A. Nielsen and I. L. Chuang, {\it Quantum Computation and Quantum Information}, (Cambridge University Press, Cambridge, 2002).

  \bibitem{GKS} V. Gorini, A. Kossakowski, and E. C. G. Sudarshan, Completely positive dynamical semigroups of N-level systems,
  \href{https://aip.scitation.org/doi/abs/10.1063/1.522979}{J. Math. Phys. {\bf 17}, 821 (1976).}

  \bibitem{Lindblad} G. Lindblad, On the Generators of Quantum Dynamical Semigroups,
  \href{https://link.springer.com/article/10.1007/BF01608499}{Commun. Math. Phys. {\bf 48}, 119 (1976).}

  \bibitem{Havel03} T. F. Havel, Robust procedures for converting among Lindblad, Kraus and matrix representations of quantum dynamical semigroups,
  \href{https://aip.scitation.org/doi/abs/10.1063/1.1518555?journalCode=jmp}{J. Math. Phys. {\bf 44}, 534 (2003).}

  \bibitem{Kotovic-review} P. V. Kotovic, Application of Singular Perturbation Techniques to Control Problems,
  \href{https://www.jstor.org/stable/2030977}{SIAM Review {\bf 26}, 501 (1984).}

  %%% Born-Markov %%%%%%%%%%%%%%%%%%%%%%%%%%%%%%%%%%%%%%%%%%%%%
  \bibitem{Cirac92} J. I. Cirac,
  Interaction of a two-level atom with a cavity mode in the bad-cavity limit,
  \href{https://journals.aps.org/pra/abstract/10.1103/PhysRevA.46.4354}{Phys. Rev. A {\bf 46}, 4354 (1992).}

  \bibitem{Wiseman93} H. M. Wiseman and G. J. Milburn,
  Quantum theory of field-quadrature measurements,
  \href{https://journals.aps.org/pra/abstract/10.1103/PhysRevA.47.642}{Phys. Rev. A {\bf 47}, 642 (1993).}

  \bibitem{Ripoll09} J. J. García-Ripoll, S. Dürr, N. Syassen, D. M. Bauer, M. Lettner, G. Rempe, and J. I. Cirac,
  Dissipation-induced hard-core boson gas in an optical lattice,
  \href{https://iopscience.iop.org/article/10.1088/1367-2630/11/1/013053}{New J. Phys. 11, 013053 (2009).}

  \bibitem{Reiter12} Florentin Reiter and Anders S. Sørensen,
  Effective operator formalism for open quantum systems,
  \href{https://journals.aps.org/pra/abstract/10.1103/PhysRevA.85.032111}{Phys. Rev. A {\bf 85}, 032111 (2012).}

  \bibitem{Lesanovsky13} Igor Lesanovsky and Juan P. Garrahan,
  Kinetic constraints, hierarchical relaxation and onset of glassiness in strongly interacting and dissipative Rydberg gases,
  \href{https://journals.aps.org/prl/abstract/10.1103/PhysRevLett.116.240404}{Phys. Rev. Lett. {\bf 111}, 215305 (2013).}

  \bibitem{Karabanov15} A. Karabanov, D. Wiśniewski, I. Lesanovsky, and W. Köckenberger, Dynamic Nuclear Polarization as Kinetically Constrained Diffusion,  \href{https://journals.aps.org/prl/abstract/10.1103/PhysRevLett.115.020404}{Phys. Rev. Lett. {\bf 115}, 020404 (2015).}

  \bibitem{Tomita17} T. Tomita, S. Nakajima, I. Danshita, Y. Takasu, and Y. Takahashi,
  Observation of the Mott insulator to superfluid crossover of a driven-dissipative Bose-Hubbard system,
  \href{https://www.science.org/doi/10.1126/sciadv.1701513}{Sci. Adv. 3, e1701513 (2017).}

  \bibitem{Damanet19} François Damanet, Andrew J. Daley, and Jonathan Keeling,
  Atom-only descriptions of the driven-dissipative Dicke model,
  \href{https://journals.aps.org/pra/abstract/10.1103/PhysRevA.99.033845}{Phys. Rev. A {\bf 99}, 033845 (2019).}

  \bibitem{Viana22} Alejandro Vivas-Viaña, Alejandro González-Tudela, and Carlos Sánchez Muñoz,
  Unconventional mechanism of virtual-state population through dissipation,
  \href{https://journals.aps.org/pra/abstract/10.1103/PhysRevA.106.012217}{Phys. Rev. A {\bf 106}, 012217 (2022).}
  
 \bibitem{Yang22} Bo Yang, Bo Xiong, Zilong Liu, and Bo Zhang,
  Effective fourth-order Lindblad equation for a weakly-coupled dissipative quantum system
  \href{https://www.sciencedirect.com/science/article/abs/pii/S0375960121007453}{Phys. Lett. A {\bf 425}, 127881 (2022).}
  %%%%%%%%%%%%%%%%%%%%%%%%%%%%%%%%%%%%%%%%%%%%%%%%

  %%% Laplace %%%%%%%%%%%%%%%%%%%%%%%%%%%%%%%%%%%%%%%%%%%%%
  \bibitem{F-Shapiro20} Daniel Finkelstein-Shapiro, David Viennot, Ibrahim Saideh, Thorsten Hansen, Tõnu Pullerits, and Arne Keller,
  Adiabatic elimination and subspace evolution of open quantum systems
  \href{https://journals.aps.org/pra/abstract/10.1103/PhysRevA.101.042102}{Phys. Rev. A {\bf 101}, 042102 (2020).}

  \bibitem{Saideh20} Ibrahim Saideh, Daniel Finkelstein-Shapiro, Tõnu Pullerits, and Arne Keller,
  Projection-based adiabatic elimination of bipartite open quantum systems,
  \href{https://journals.aps.org/pra/abstract/10.1103/PhysRevA.102.032212}{Phys. Rev. A {\bf 102}, 032212 (2020).}
  %%%%%%%%%%%%%%%%%%%%%%%%%%%%%%%%%%%%%%%%%%%%%%%%

  %%% SW %%%%%%%%%%%%%%%%%%%%%%%%%%%%%%%%%%%%%%%%%%%%%
  \bibitem{Kessler12} E. M. Kessler,
  Generalized Schrieffer-Wolff formalism for dissipative systems,
  \href{https://journals.aps.org/pra/abstract/10.1103/PhysRevA.86.012126}{Phys. Rev. A {\bf 86}, 012126 (2012).}
  
  \bibitem{Burgarth21} Daniel Burgarth, Paolo Facchi, Hiromichi Nakazato, Saverio Pascazio, and Kazuya Yuasa,
  Eternal adiabaticity in quantum evolution,
  \href{https://journals.aps.org/pra/abstract/10.1103/PhysRevA.103.032214}{Phys. Rev. A {\bf 103}, 032214 (2021).}

  \bibitem{Jager22} Simon B. Jäger, Tom Schmit, Giovanna Morigi, Murray J. Holland, and Ralf Betzholz,
  Lindblad Master Equations for Quantum Systems Coupled to Dissipative Bosonic Modes,
  \href{https://journals.aps.org/prl/abstract/10.1103/PhysRevLett.129.063601}{Phys. Rev. Lett. {\bf 129}, 063601 (2022).}
  %%%%%%%%%%%%%%%%%%%%%%%%%%%%%%%%%%%%%%%%%%%%%%%%

  \bibitem{Macieszczak16} K. Macieszczak, M. Guţă, I. Lesanovsky, and J. P. Garrahan, Towards a Theory of Metastability in Open Quantum Dynamics,
  \href{https://journals.aps.org/prl/abstract/10.1103/PhysRevLett.116.240404}{Phys. Rev. Lett. {\bf 116}, 240404 (2016).}

  \bibitem{Zanardi16} P. Zanardi, J. Marshall, and L. Campos Venuti, Dissipative universal Lindbladian simulation,
  \href{https://journals.aps.org/pra/abstract/10.1103/PhysRevA.93.022312}{Phys. Rev. A {\bf 93}, 022312 (2016).}

  \bibitem{Azouit17} R. Azouit, F. Chittaro, A. Sarlette, and P. Rouchon, Towards generic adiabatic elimination for bipartite open quantum systems,
  \href{https://iopscience.iop.org/article/10.1088/2058-9565/aa7f3f}{Quantum Sci. Technol. {\bf 2}, 044011 (2017).}

  \bibitem{Fenichel79} N. Fenichel. Geometric singular perturbation theory for ordinary differential equations,
  \href{https://www.sciencedirect.com/science/article/pii/0022039679901529}{J. Di. Equations, {\bf 31}:53-98, (1979).}

  \bibitem{Michiel23} M. Burgelman, P. Forni, and A. Sarlette, Quantum dynamical decoupling by shaking the close environment,
  \href{https://www-sciencedirect-com.kyoto-u.idm.oclc.org/science/article/pii/S0016003222005580}{J. Franklin Inst {\bf 360} 14022 (2023).}

  \bibitem{FM24} Francois-Marie Le Régent and Pierre Rouchon,
  Adiabatic elimination for composite open quantum systems: Reduced-model formulation and numerical simulations,
  \href{https://journals.aps.org/pra/abstract/10.1103/PhysRevA.109.032603}{Phys. Rev. A {\bf 109}, 032603 (2024).}

  \bibitem{Choi75} M. D. Choi, Completely Positive Linear Maps on Complex Matrices,
  \href{https://www.sciencedirect.com/science/article/pii/0024379575900750}{Linear Algebra and its Applications, {\bf 10}, 285 (1975).}

  \bibitem{AzouitThesis} R. Azouit, {\it Adiabatic elimination for open quantum systems},
  \href{https://pastel.archives-ouvertes.fr/tel-01743808/document}{Doctoral thesis (2017).}

  \bibitem{Alain20} A. Sarlette, P. Rouchon, A. Essig, Q. Ficheux and B. Huard, Quantum adiabatic elimination at arbitrary order for photon number measurement,
  \href{https://www.sciencedirect.com/science/article/pii/S2405896320303876}{IFACPapersOnLine {\bf 53}, 250 (2020).}

  \bibitem{WPG10} M. M. Wolf and D. Perez-Garcia, The inverse eigenvalue problem for quantum channels,
  \href{https://arxiv.org/abs/1005.4545}{arXiv:quant-ph/1005.4545.}

  \bibitem{Tokieda22} M. Tokieda, C. Elouard, A. Sarlette, and P. Rouchon,
  Complete Positivity Violation in Higher-order Quantum Adiabatic Elimination,
  \href{https://www.sciencedirect.com/science/article/pii/S2405896323021882}{IFACPapersOnLine {\bf 56}, 1333 (2023).}

  \bibitem{Antoine21} A. Essig, Q. Ficheux, T. Peronnin, N. Cottet, R. Lescanne, A. Sarlette, P. Rouchon, Z. Leghtas, and B. Huard,
  Multiplexed Photon Number Measurement,
  \href{https://journals.aps.org/prx/abstract/10.1103/PhysRevX.11.031045}{Phys. Rev. X {\bf 11}, 031045 (2021).}

  \bibitem{Muller17} Clemens Müller and Thomas M. Stace,
  Deriving Lindblad master equations with Keldysh diagrams: Correlated gain and loss in higher order perturbation theory,
  \href{https://journals.aps.org/pra/abstract/10.1103/PhysRevA.95.013847}
  {Phys. Rev. A {\bf 95}, 013847 (2017).}

  \bibitem{Redfield57} A. G. Redfield, 
  On the theory of relaxation processes,
  \href{https://doi.org/10.1147/rd.11.0019}{IBM J. Res. Deli., {\bf 1}, 19 (1957).}

  \bibitem{inv.close} For the oscillator-qubit system in Sec.$\,$\ref{sec:JC}, for instance, the fast relaxing modes typically decay at the slowest rate as $\exp(- \gamma t)$ for small enough $g$. This implies that the trajectory is $\delta$-close to the invariant manifold when $t = \gamma^{-1} \ln \delta^{-1}$.

  \bibitem{Maniscalco04} S. Maniscalco, F. Intravaia, J. Piilo, and A. Messina,
  Misbeliefs and misunderstandings about the non-Markovian dynamics of a damped harmonic oscillator,
  \href{https://iopscience-iop-org.kyoto-u.idm.oclc.org/article/10.1088/1464-4266/6/3/016}{J. Opt. B: Quantum Semiclass. Opt. {\bf 6}, S98 (2004).}

  \bibitem{Whitney08} R. S. Whitney, Staying positive: going beyond Lindblad with perturbative master equations,
  \href{https://iopscience.iop.org/article/10.1088/1751-8113/41/17/175304}{J. Phys. A: Math. Theor. {\bf 41}, 175304 (2008).}

  \bibitem{Hartmann20} R. Hartmann and W. T. Strunz, Accuracy assessment of perturbative master equations: Embracing nonpositivity,
  \href{https://journals.aps.org/pra/abstract/10.1103/PhysRevA.101.012103}{Phys. Rev. A {\bf 101}, 012103 (2020).}

  \bibitem{Rivas14} Ángel Rivas, Susana F Huelga, and Martin B Plenio,
  Quantum non-Markovianity: characterization, quantification and detection,
  \href{https://iopscience.iop.org/article/10.1088/0034-4885/77/9/094001}{Rep. Prog. Phys. {\bf 77}, 094001 (2014).}

  \bibitem{Breuer16} Heinz-Peter Breuer, Elsi-Mari Laine, Jyrki Piilo, and Bassano Vacchini,
  Colloquium: Non-Markovian dynamics in open quantum systems,
  \href{https://journals.aps.org/rmp/abstract/10.1103/RevModPhys.88.021002}{Rev. Mod. Phys. {\bf 88}, 021002 (2016).}

  \bibitem{CPdebate} P. Pechukas, Reduced Dynamics Need Not Be Completely Positive,
  \href{https://journals.aps.org/prl/abstract/10.1103/PhysRevLett.73.1060}{Phys. Rev. Lett. {\bf 73}, 1060 (1994)};
  R. Alicki, \href{https://journals.aps.org/prl/abstract/10.1103/PhysRevLett.75.3020}{{\it ibid}. {\bf 75}, 3020 (1995)};
  P. Pechukas. \href{https://journals.aps.org/prl/abstract/10.1103/PhysRevLett.75.3021}{{\it ibid}. {\bf 75}, 3021 (1995)}.

  \bibitem{PF} Paolo Forni (private communication).

  \bibitem{Stelmachovic01} Peter Štelmachovič and Vladimír Bužek,
  Dynamics of open quantum systems initially entangled with environment: Beyond the Kraus representation
  \href{https://journals.aps.org/pra/abstract/10.1103/PhysRevA.64.062106}{Phys. Rev. A {\bf 64}, 062106 (2001).}

  \bibitem{Jordan04} T. F. Jordan, A. Shaji, and E.C.G. Sudarshan,
  Dynamics of initially entangled open quantum systems,
  \href{https://journals.aps.org/pra/abstract/10.1103/PhysRevA.70.052110}{Phys. Rev. A {\bf 70}, 052110 (2004).}
  
  \bibitem{McCracken13} James M. McCracken,
  Hamiltonian composite dynamics can almost always lead to negative reduced dynamics,  \href{https://journals.aps.org/pra/abstract/10.1103/PhysRevA.88.022103}{Phys. Rev. A {\bf 88}, 022103 (2013).}

  \bibitem{Rosario08} César A Rodríguez-Rosario, Kavan Modi, Aik-meng Kuah, Anil Shaji, and E C G Sudarshan,
  Completely positive maps and classical correlations,
  \href{https://iopscience.iop.org/article/10.1088/1751-8113/41/20/205301}{J. Phys. A: Math. Theor. {\bf 41}, 205301 (2008).}
  
  \bibitem{Henderson01} L. Henderson and V. Vedral,
  Classical, quantum and total correlations,
  \href{https://iopscience.iop.org/article/10.1088/0305-4470/34/35/315/meta}{J. Phys. A: Math. Gen. 34, 6899 (2001).}

  \bibitem{Ollivier01} Harold Ollivier and Wojciech H. Zurek,
  Quantum Discord: A Measure of the Quantumness of Correlations,
  \href{https://journals.aps.org/prl/abstract/10.1103/PhysRevLett.88.017901}{Phys. Rev. Lett. {\bf 88}, 017901 (2001).}

  \bibitem{Modi13} Kavan Modi,
  A Pedagogical Overview of Quantum Discord,
  \href{https://www.worldscientific.com/doi/abs/10.1142/S123016121440006X}{Open Syst. Inf. Dyn. 21, 1440006 (2014).}

  \bibitem{Shabani09} Alireza Shabani and Daniel A. Lidar
  Vanishing Quantum Discord is Necessary and Sufficient for Completely Positive Maps,
  \href{https://journals.aps.org/prl/abstract/10.1103/PhysRevLett.102.100402}{Phys. Rev. Lett. {\bf 102}, 100402 (2009)};
  \href{https://journals.aps.org/prl/abstract/10.1103/PhysRevLett.116.049901}{{\it ibid}. {\bf 116}, 049901 (2016)}.

  \bibitem{Brodutch13} Aharon Brodutch, Animesh Datta, Kavan Modi, Ángel Rivas, and César A. Rodríguez-Rosario,
  Vanishing quantum discord is not necessary for completely positive maps,
  \href{https://journals.aps.org/pra/abstract/10.1103/PhysRevA.87.042301}{Phys. Rev. A {\bf 87}, 042301 (2013).}

  \bibitem{Buscemi14} Francesco Buscemi,
  Complete Positivity, Markovianity, and the Quantum Data-Processing Inequality, in the Presence of Initial System-Environment Correlations,
  \href{https://journals.aps.org/prl/abstract/10.1103/PhysRevLett.113.140502}{Phys. Rev. Lett. {\bf 113}, 140502 (2014).}
  
  \bibitem{Kimura03} G. Kimura, The Bloch Vector for $N$-Level
  Systems,
  \href{https://www.sciencedirect.com/science/article/abs/pii/S0375960103009411?via%3Dihub}{Phys. Lett. A {\bf 314}, 339 (2003).}

  \bibitem{Briegel93} H.-J. Briegel and B.-G. Englert, Quantum optical master equations: The use of damping bases,
  \href{https://journals.aps.org/pra/abstract/10.1103/PhysRevA.47.3311}{Phys. Rev. A {\bf 47}, 3311 (1993).}

  \bibitem{Prosen10} T. Prosen and T. H. Seligman, Quantization over boson operator spaces,
  \href{https://iopscience.iop.org/article/10.1088/1751-8113/43/39/392004/meta}{J. Phys. A: Math. Theor. {\bf 43}, 392004 (2010).}

  \bibitem{BZ21} T. Barthel and Y. Zhang, Solving quasi-free and quadratic Lindblad master equations for open fermionic and bosonic systems,
  \href{https://iopscience.iop.org/article/10.1088/1742-5468/ac8e5c/meta}{J. Stat. Mech. 113101 (2022).}

\end{thebibliography}
\end{document}